\documentclass[reqno]{amsart}

\usepackage{amsfonts,latexsym,enumerate}
\usepackage{amsmath}
\usepackage{amscd}
\usepackage{float,amsmath,amssymb,mathrsfs,bm,multirow,graphics}
\usepackage[dvips]{graphicx}
\usepackage[percent]{overpic}
\usepackage{amsaddr}
\usepackage[numbers,sort&compress]{natbib}
\usepackage{tikz}
\usepackage{epstopdf}
\usepackage{subfigure}

\newcommand{\R}{{\Bbb R}}

\newcommand{\C}{{\Bbb C}}

\newtheorem{theorem}{Theorem}[section]
\newtheorem{proposition}[theorem]{Proposition}
\newtheorem{lemma}[theorem]{Lemma}

\newtheorem{remark}[theorem]{Remark}

\numberwithin{equation}{section}

\begin{document}

\title[Long-time asymptotics of the SK and KK equations]
{Long-time asymptotics of the Sawada-Kotera equation and Kaup-Kupershmidt equation on the line}

   \author{Deng-Shan Wang,~Xiaodong Zhu}
\address{Laboratory of Mathematics and Complex Systems (Ministry of Education), School of Mathematical Sciences, Beijing Normal University, Beijing 100875, China}
\email{xdzhu@mail.bnu.edu.cn}


\date{July 17, 2023.}



\begin{abstract}
Both Sawada-Kotera (SK) equation and Kaup-Kupershmidt (KK) equation are integrable systems with third-order Lax operator. Moreover, they are related with the same modified nonlinear equation (called modified SK-KK equation) by Miura transformations. This work first constructs the Riemann-Hilbert problem associated with the SK equation, KK equation and modified SK-KK equation by direct and inverse scattering transforms. Then the long-time asymptotics of these equations are studied based on Deift-Zhou steepest-descent method for Riemann-Hilbert problem. Finally, it is shown that the asymptotic solutions match very well with the results of direct numerical simulations.

\end{abstract}

\maketitle

\section{\bf Introduction}
In 1974, Sawada and Kotera \cite{Sawada-Kotera-1974} proposed the so-called Sawada-Kotera (SK) equation
\begin{equation}\label{SK}
	u_t+u_{x x x x x}+30\left(u u_{x x x}+u_x u_{x x}\right)+180 u^2 u_x=0,
\end{equation}
which is also named the Caudrey-Dodd-Gibbon equation given by Caudrey, Dodd and Gibbon \cite{Caudrey-Dodd-1976} independently. Subsequently, Kaup \cite{Kaup-1980} and Kupershmidt \cite{Kupershmidt-1984} gave the Kaup-Kupershmidt (KK) equation
\begin{equation}\label{KK}
	v_t+v_{xx x x x}+30(v v_{x x x}+\frac{5}{2} v_x v_{x x})+180 v^2 v_x=0.
\end{equation}
Both SK equation (\ref{SK}) and KK equation (\ref{KK}) are completely integrable systems with third-order Lax operator of the form $\psi_{xxx}+6Q\psi_x+6R\psi=k^3 \psi$ studied by Kaup \cite{Kaup-1980}, in which $Q=u$ and $R=0$ correspond to the SK equation, while $Q=u$ and $R=u_x/2$ correspond to the KK equation.
The Lax pair of the SK equation (\ref{SK}) in matrix form is
\begin{equation}\label{SK-lax-pair}
\left\{\begin{array}{l}
	\Phi_x=L \Phi, \\
	\Phi_t=Z \Phi,
\end{array}\right.
\end{equation}
where
\begin{equation}\label{SKlaxspace}
	L=\left(\begin{array}{ccc}
		0 & 1 & 0 \\
		0 & 0 & 1 \\
		k^3 & -6 u & 0
	\end{array}\right),\\	
\end{equation}
\begin{equation}\label{SKlaxtime}
Z=\left(\begin{array}{ccc}
	36 k^3 u & 6 u_{x x}-36 u^2 & 9 k^3-18 u_x \\
	18 k^3 u_x+9 k^6 & 6 u_{x x x}-18 k^3 u+36 u u_x & -12 u_{x x}-36 u^2 \\
	6 k^3 u_{x x}-36 u^2 k^3 & Z_{32} & -6 u_{x x x}-18 k^3 u-36 u u_x
\end{array}\right),
\end{equation}
with spectral parameter $k$ and $Z_{32}=36 u_x{ }^2+108 u u_{x x}+9 k^6+216 u^3+6 u_{x x x x}$.
\par
The Lax pair of the KK equation (\ref{KK}) in matrix form is
\begin{equation}\label{KK-lax-pair}
\left\{\begin{array}{l}
	\Phi_x=\tilde { L} \Phi, \\
	\Phi_t=\tilde Z \Phi,
\end{array}\right.
\end{equation}
where
\begin{equation}\label{KKlaxspace}
\tilde L=\left(\begin{array}{ccc}
	0 & 1 & 0 \\
	0 & 0 & 1 \\
	k^3-3 u_x & -6 u & 0
\end{array}\right),
\end{equation}
\begin{equation}\label{KKlaxtime}
\tilde Z=\left(\begin{array}{ccc}
	3 u_{x x x}+36 k^3 u+72 u u_x & -3 u_{x x}-36 u^2 & 9 k^3 \\
	\tilde Z_{21} & -18 k^3 u & -3 u_{x x}-36 u^2 \\
	\tilde Z_{31} &\tilde  Z_{32} & -72 u u_x-18 k^3 u-3 u_{x x x}
\end{array}\right),
\end{equation}
with $\tilde Z_{21}=9 k^3 u_x+9 k^6+3 u_{x x x x}+72 u_x^2+72 u u_{x x}$,  $\tilde Z_{31}=3 u_{x x x x x}+225 u_x u_{x x}$$+ 72 u u_{x x x}$$+108 u^2 u_x-36 u^2 k^3+6 k^3 u_{x x}$ and $\tilde Z_{32}=9 k^6-9 k^3 u_x+3 u_{x x x x}+72 u_x^2+$$90 u u_{x x}+216 u^3$.
\par
Notice that the space parts of the operator Lax pairs for the SK and KK equations have the similar form, i.e., $\mathscr{L}\phi=(\partial_{xxx}+6u\partial)\phi=k^3\phi$ for the SK equation, and $\tilde{\mathscr{L}}\phi=(\partial_{xxx}+6u\partial_x+3u_{x})\phi=k^3\phi$ the KK equation. Thus it is convenient to consider a more general Lax operator of the form
\begin{equation}\label{3-lax-operator-0}
\mathscr{L}\phi=(\partial_{xxx}+p\partial_x+q)
\end{equation}
which corresponds to the matrix form Lax pair
\begin{equation}\label{3-lax-matrix}
\Phi_x=L \Phi
\end{equation}
with
$$
L=\left(\begin{array}{ccc}
	0 & 1 & 0\\
	0 & 0 & 1\\
	k^3-q & -p & 0
\end{array}\right).
$$
It is obvious that $q=0, p=6u$ for the SK equation (\ref{SK}) and $q=3v_x, p=6v$ for the KK equation (\ref{KK}).
\par
In addition, both the SK equation (\ref{SK}) and the KK equation (\ref{KK}) are related with the modified SK-KK equation (mSK-KK)
\begin{equation}\label{msk-equation}
	w_t+w_{x x x x x}-(5 w_x w_{x x}+5 w w_x^2+5 w^2 w_{x x}-w^5)_x=0
\end{equation}
through the Miura transformations
\begin{equation}\label{miura-SK}
	u=\frac{1}{6}(w_x-w^2)\quad {\rm and} \quad v=\frac{1}{3}(w_x-\frac{w^2}{2}).
\end{equation}
It is noted that there doesn't exist singularity in the spectral problem of the mSK-KK equation (\ref{msk-equation}), thus it is practicable to study the long-time asymptotics of the equations (\ref{SK})-(\ref{KK}) by the examining the asymptotic behavior of the mSK-KK equation (\ref{msk-equation}). \par
In what follows, the direct and inverse scattering transforms \cite{GGKM-1967}-\cite{Shchesnovich-Yang} are performed to derive the Riemman-Hilbert problem of the SK and KK equations (\ref{SK})-(\ref{KK}) and the mSK-KK equation (\ref{msk-equation}).

\section{\bf The Riemman-Hilbert problem}

Introduce the gauge transformation
\begin{equation}\label{gauge-matrix}
\Phi=G\Psi \quad {\rm with}\quad G=\left(
	\begin{array}{ccc}
		{\alpha} & {\alpha^2} & {1} \\
		{\alpha^2}{ k} & {\alpha}{ k} & {k} \\
		k^2 & k^2 & k^2 \\
	\end{array}
	\right),\quad \alpha=e^{\frac{2\pi i}{3}},
\end{equation}
then the spectral problem (\ref{3-lax-matrix}) becomes
\begin{equation}\label{lax-pair-new-x-part}
\Psi_x=\mathcal{L} \Psi,
\end{equation}
where $\mathcal{L}=G^{-1}LG=k\Lambda+Q(x,t;k)$ with
$$
Q(x,t;k)=\frac{Q_{(1)}(x,t;k)}{k}+\frac{Q_{(2)}(x,t;k)}{k^2},
$$
and
$$
\Lambda=\left(\begin{array}{ccc}
	\alpha & 0 & 0\\
	0 & \alpha^2 & 0\\
	0 & 0 & 1\\
\end{array}
\right),\quad \\
Q_{(1)}=-\frac{p}{3}\left(\begin{array}{ccc}
	\alpha^2 & \alpha & 1\\
	\alpha^2 & \alpha & 1\\
	\alpha^2 & \alpha & 1\\
\end{array}
\right),\quad
Q_{(2)}=-\frac{q}{3}\left(\begin{array}{ccc}
	\alpha & \alpha^2 & 1\\
	\alpha & \alpha^2 & 1\\
	\alpha & \alpha^2 & 1\\
\end{array}
\right).
$$
\par
Following the same procedure, the gauge transformation (\ref{gauge-matrix}) maps the temporal part of the Lax pairs (\ref{SK-lax-pair}) and (\ref{KK-lax-pair}) into
\begin{equation}\label{lax-pair-new-t-part}
\Psi_t=\mathcal{Z} \Psi \quad {\rm and} \quad \Psi_t=\tilde {\mathcal{Z}} \Psi,
\end{equation}
respectively, where
$$
\mathcal{Z}=G^{-1}ZG=9k^5\Lambda^2+P(x,t;k) \quad {\rm and} \quad
\tilde {\mathcal{Z}}=G^{-1}\tilde ZG=9k^5\Lambda^2+\tilde P(x,t;k),
$$
where $P(x,t;k)$ and $\tilde P(x,t;k)\to0$ as $|x|\to \infty$.
\par
Thus the gauge transformation (\ref{gauge-matrix}) converts the Lax pairs (\ref{SK-lax-pair}) and (\ref{KK-lax-pair}) into
\begin{equation}\label{New-Lax-SK-KK}
\left\{\begin{array}{l}
	\Psi_x=(k\Lambda+Q) \Psi, \\
	\Psi_t=(k^5\Lambda^2+P)\Psi,
\end{array}\right.
\quad
\left\{\begin{array}{l}
	\Psi_x=(k\Lambda+Q) \Psi, \\
	\Psi_t=(k^5\Lambda^2+\tilde P)\Psi.
\end{array}\right.
\end{equation}
Furthermore, taking $\Psi=J e^{(k\Lambda x+k^5\Lambda^2 t)}$ yields
\begin{equation}\label{Lax-equation}
	\left\{\begin{array}{l}
	J_x-[k\Lambda,J]=Q J, \\
	J_t-[k^5\Lambda^2,J]=P J, \\
\end{array}\right. \quad
	\left\{\begin{array}{l}
	J_x-[k\Lambda,J]=Q J, \\
	J_t-[k^5\Lambda^2,J]=\tilde P J. \\
\end{array}\right.
\end{equation}
\par
In what follows, we only focus on the $x$-variable and take $t$-variable as a dump variable. Moreover, according to the equation $J_x-[k\Lambda,J]=Q J$, one can get the Volterra integral equation of the Jost functions $J_{+}(x, k)$ and $J_{-}(x, k)$ below
\begin{equation}\label{Jost-functions}
	\begin{aligned}
	& J_{+}(x, k)=I-\int_x^{\infty} e^{(x-y) \widehat{k\Lambda}}\left(Q(y, k) J_{+}(y, k)\right) d y, \\
	& J_{-}(x, k)=I+\int_{-\infty}^x e^{(x-y) \widehat{k\Lambda}}\left(Q(y, k) J_{-}(y, k)\right) d y,
\end{aligned}
\end{equation}
which indicates the singular set
$$
\Sigma:=\{k\in\mathbb{C}|{\rm Re}(\alpha^nk)={\rm Re}(\alpha^m k),\quad 0\le n<m<3\},
$$
then $\Sigma$ divided the complex plane into six regions, specifically
$$
\Omega_n:=\{k\in\mathbb{C}|\frac{(n-1)\pi}{3}<{\rm arg}(k)<\frac{n\pi}{3},n=1,\cdots,6\}.
$$
\par
The following way to construct the Riemann-Hilbert problem \cite{Boutet-de-Monvel-1}-\cite{Lenells-2018}
is standard, so we omit the process of proof, see \cite{Charlier-Lenells-2021} for details.

\subsection{Basic properties of the Jost functions}

\begin{proposition}
Suppose the initial potential functions $p_0(x),q_0(x)\in\mathcal{S}(x)$, then the matrix-valued Jost functions $J_+(x,k)$ and $J_-(x,k)$ have the following properties:

(1). $J_+(x,k)$ is well defined in the closure of  $(\omega^2S,\omega S,S)\setminus\{0\}$, and $J_-(x,k)$ is well defined in the closure of  $(-\omega^2S,-\omega S,-S)\setminus\{0\}$, where $S=\Omega_3\cup\Omega_4$. Moreover, the determinant of $J_{\pm}$ are always equal to $1$.

(2). $J_+(\cdot,k)$ and $J_-(\cdot,k)$ are smooth and rapidly decay in the closure of their domains (except for $\{0\}$).

(3). $J_+(x,\cdot)$ and $J_-(x,\cdot)$ are analytic in the interior of  their domains, but any order partial derivative of $k$ can be continuous to the closure of their domains (except for $\{0\}$).

(4). $J_+(x,k)$  and $J_-(x,k)$ satisfied the following symmetries:
$$
\begin{aligned}
	& J_+(x, k)=\mathcal{A} J_+(x, \omega k) \mathcal{A}^{-1}=\mathcal{B} {J_+^*(x, {k}^*)} \mathcal{B},\\
	& J_-(x, k)=\mathcal{A} J_-(x, \omega k) \mathcal{A}^{-1}=\mathcal{B} {J_-^*(x, {k}^*)} \mathcal{B},
\end{aligned}
$$
where $k$ is in their domains and $\mathcal{A}, \mathcal{B}$ are
$$
	\mathcal{A}:=\left(\begin{array}{ccc}
		0 & 0 & 1 \\
		1 & 0 & 0 \\
		0 & 1 & 0
	\end{array}\right) \quad \text { and } \quad \mathcal{B}:=\left(\begin{array}{lll}
		0 & 1 & 0 \\
		1 & 0 & 0 \\
		0 & 0 & 1
	\end{array}\right)
	$$

(5). When $p_0$ and $q_0$ are compact support, $J_+$ and $J_-$ are defined and analytic for $k$ on $\C\setminus\{0\}$ . 	
\end{proposition}

\subsection{The behavior of Jost functions for $k\to\infty$.}

Suppose the WKB expansion of the Jost functions $J_{\pm}$ to be
$$
J_{\pm}=I+\frac{J_{\pm}^{(1)}}{k}+\frac{J_{\pm}^{(2)}}{k^2}+\cdots
$$
Taking into account of the equation ({\ref{Lax-equation}}), one has
\begin{equation}\label{WkB-infty}
	\left\{\begin{array}{l}
	{\left[\Lambda, J_{\pm}^{(n+1)}\right]=(\partial_x J_{\pm}^{(n)})^{(o)}-\left(\mathrm{Q}_1 J_{\pm}^{(n-1)}\right)^{(o)}-\left(\mathrm{Q}_2 J_{\pm}^{(n-2)}\right)^{(o)},} \\
	(\partial_x J_{\pm}^{(n+1)})^{(d)}=\left(\mathrm{Q}_1 J_{\pm}^{(n)}\right)^{(d)}+\left(\mathrm{Q}_2 J_{\pm}^{(n-1)}\right)^{(d)}.
\end{array}\right.
\end{equation}
\par
Furthermore, we have
$$
J_+^{(1)}=\int_x^{\infty}\frac{p}{3}dy\left(\begin{array}{ccc}
	\alpha^2  & 0 & 0\\
	0 & \alpha& 0\\
	0 & 0 & 1\\
\end{array}
\right),
$$
\begin{equation}\label{J1-infty}
J_+^{(2)}=\int_x^{\infty}\frac{q+p(J_1)_{33}}{3}dy\left(\begin{array}{ccc}
	\alpha  & 0 & 0\\
	0 & \alpha^2 & 0\\
	0 & 0 & 1\\
\end{array}
\right)+\frac{p}{3(1-\alpha)}\left(\begin{array}{ccc}
	0  & 1 & -1\\
	-\alpha & 0 & \alpha\\
	\alpha^2 & -\alpha^2 & 0\\
\end{array}
\right).
\end{equation}

\begin{proposition}
Suppose $\{q_0,p_0\}\in\mathcal{S}(\R)$, there exist bounded smooth functions $f_{\pm}$, which rapidly decay as $x\to\infty$ and $x\to-\infty$, respectively. Let $m\ge 0$ be an integer and for each integer $n\ge 0$, then
$$
\left|\frac{\partial^n}{\partial_{k^n}}\left[ J_{\pm}-\left(I+\frac{J_{\pm}^{(1)}}{k}+\cdots+\frac{J_{\pm}^{(m)}}{k^m}\right)\right]\right|\le \frac{f_{\pm}(x)}{k^{m+1}},
$$
where $k$ is in the domain of $J_{\pm}$, respectively and large enough.
\end{proposition}

\subsection{The behavior of Jost functions for  $k\to0$.}

Since the kernel matrix function $Q(x;k)$ has double poles at $k=0$, this illustrates the asymptotics of $J_{\pm}$ as $k\to 0$.

\begin{proposition}
	Suppose $\{q_0,p_0\}\in\mathcal{S}(\R)$, there exist bounded smooth functions $f_{\pm}$, which are rapidly decay as $x\to\infty$ and $x\to-\infty$, respectively. Let $m\ge 0$ be an integer and for each integer $n\ge 0$, the Jost function $J_{\pm}$ has the following expansion:
$$
\left|\frac{\partial^n}{\partial_{k^n}}\left[ J_{\pm}(x,k)-\left(\frac{\mathcal{J}_{\pm}^{(-2)}}{k^2}+\frac{\mathcal{J}_{\pm}^{(-1)}}{k}+I+{\mathcal{J}_{\pm}^{(1)}}{k}+\cdots+{J_{\pm}^{(m)}}{k^m}\right)\right]\right|\le {g_{\pm}(x)}{k^{m+1}}
$$
where $k$ is small enough.
Furthermore, the leading term $\mathcal{J}_{\pm}^{(-2)}$  has the form:
$$
\mathcal{J}_{\pm}^{(-2)}(x)=a_{\pm}(x)\left(\begin{array}{ccc}
	\alpha & \alpha^2 & 1 \\
	\alpha & \alpha^2 & 1 \\
	\alpha & \alpha^2 & 1
\end{array}\right)
$$
where $a_{\pm}(x)$ is a real valued function and is dominated by $g_{\pm}(x)$ with rapid decay as $x\to \infty$ and $x\to-\infty$, respectively.
\end{proposition}

\subsection{The scattering matrix}

Define the scattering matrix as
\begin{equation}\label{Delta-Scaterting}
	\Delta(k)=I-\int_{\mathbb{R}} e^{-x \widehat{k\Lambda}}(Q J)(x, k) d x.
\end{equation}
\par
When the initial potential functions $p_0$ and $q_0$ are compact support, the scattering matrix $\Delta(k)$ satisfies
$$
J_+(x, k)=J_-(x, k) e^{x \widehat{k\Lambda}} \Delta(k), \quad k \in \mathbb{C}\setminus\{0\}.
$$
\begin{proposition}
	Suppose $\{q_0,p_0\}\in\mathcal{S}(\R)$, then the scattering function $\Delta(k)$ defined in (\ref{Delta-Scaterting}) has the following properties:

(a) The domain of $\Delta(k)$:
$$
\Delta(k) \in \left(\begin{array}{ccc}
	\omega^{2}\overline{S} & \mathbb{R}_{+} & \omega \mathbb{R}_{+} \\
	\mathbb{R}_{+} & \omega \overline{S} & \omega^{2} \mathbb{R}_{+} \\
	\omega \mathbb{R}_{+} & \omega^{2} \mathbb{R}_{+} & \overline{S}
\end{array}\right)\setminus\{0\}.
$$
Here $\overline{S}$ means the closure of $S$ and $\Delta(k)$ is continuous to the boundary of domain but is analytic in the interior of its domain.here $\overline{S}$ means the closure of $S$.

(b) The matrix-valued function $\Delta(k)$ has the following Laurents expansions as $k\to\infty$ and $k\to 0$, respectively.
$$
\Delta(k)=I-\sum_{j=1}^{N} \frac{\Delta_{j}}{k^{j}}+O(\frac{1}{k^{N+1}}),\quad k \rightarrow \infty,
$$
and
$$
\Delta(k)=\frac{\Delta^{(-2)}}{k^2}+\frac{\Delta^{(-1)}}{k}+\Delta^{(0)}+\Delta^{(1)} k+\cdots, \quad k \rightarrow 0.
$$
\par
(c) The matrix-valued function $\Delta(k)$ satisfies the symmetries:
$$
\Delta(k)=\mathcal{A} \Delta(\omega k) \mathcal{A}^{-1}=\mathcal{B} \Delta^*({k}^*) \mathcal{B}.
$$
\end{proposition}

\subsection{The cofactor Jost functions}

Define $M^{A}=(M^{-1})^T$, then the adjoint equation associated with the equation $J_x-[k\Lambda,J]=Q J$ is
\begin{equation}\label{cofactor-equation}
	\left(J^{A}\right)_{x}+\left[k\Lambda, J^{A}\right]=-Q^{T}J^{A}.
\end{equation}

In the same procedure, one can also get the cofactor Jost functions $J_{\pm}^A$ and cofactor scattering matrix $\Delta^A(k)$. Furthermore, the properties of $J_{\pm}^A$ and $\Delta^A(k)$ are similar.

\subsection{The eigenfunctions $M_n$ }

Define the eigenfunctions for the equation (\ref{Lax-equation}) in each $k\in \Omega_k\setminus\{0\}$ by the following Fredholm integral

\begin{equation}\label{M_n-eigenfunction}
	\left(M_{n}\right)_{i j}(x, k)=\delta_{i j}+\int_{\gamma_{i j}^{n}}\left(e^{\left(x-y\right) \widehat{k\Lambda}}\left(Q M_{n}\right)\left(y, k\right)\right)_{i j} d y, \quad i, j=1,2,3,
\end{equation}
where $\gamma_{ij}=(x,\infty)\ \text{or}\ (-\infty,x)$, which is determined by the exponential part. Notice that there are zeros of Fredholm determinants in the complex plane denoted by $\mathcal{Z}$, but in a proper assumptions, we can extend Fredholm solutions in (\ref{M_n-eigenfunction}) on $\mathcal{Z}$.

\begin{proposition}
	Suppose $\{q_0,p_0\}\in\mathcal{S}(\R)$, then the equation (\ref{M_n-eigenfunction}) uniquely defines six $3 \times 3$ matrix-valued solutions $\left\{M_{n}\right\}_{1}^{6}$ of (\ref{Lax-equation}) with the following properties:

(a) The eigenfunctions $M_{n}(x, k)$ are defined for $x \in \mathbb{R}$ and $k \in \bar{\Omega}_{n} \backslash{(\mathcal{Z}\cup\{0\})}$. Moreover, $M_{n}(x, k)$ is bounded except for $k\in{\mathcal{Z}\cup\{0\}}$ and smooth about $x$  and continuous to $k \in \bar{\Omega}_{n} \backslash{(\mathcal{Z}\cup\{0\})}$ but analytic in the interior of its domain.

(b) The eigenfunctions $M_{n}(x, k)$ satisfied the symmetries
$$
M_n(x, k)=\mathcal{A} M_n(x, \alpha k) \mathcal{A}^{-1}={\mathcal{B} M_n^*(x, k^*)} \mathcal{B},
$$
where $k\in\bar\Omega_k\setminus{\mathcal{Z}\cup\{0\}}$.

(c) The determinant of  eigenfunctions $M_{n}(x, k)$ identically equal to one for each $k\in\bar\Omega_k\setminus{\mathcal{Z}\cup\{0\}}$.
\end{proposition}

\subsection{The properties of $M_n$ as $k\to\infty$.}
\begin{proposition}
	Suppose $q_0,p_0\in\mathcal{S}(\R)$ and $q_0,p_0$ are not identically equal to zero. Given an integer $m\ge 1$ and $k$ is large enough in its domain,  $M_n$ can be approached by the expansion of $J_{+}$ as
\begin{equation}\label{M_n-infty}
	\left| M_{n}-\left(I+\frac{J_{+}^{(1)}}{k}+\cdots+\frac{J_{+}^{(m)}}{k^m}\right)\right|\le \frac{C}{k^{m+1}}.
\end{equation}
\end{proposition}

Now, assuming  $q_0,p_0\in\mathcal{S}(\R)$ are compact support, then one can get the relation between $M_n$ and $J_{\pm}$ by
$$
\begin{aligned}
	M_{n}(x, k) &=J_-(x, k) e^{x \widehat{\mathcal{L}(k})} S_{n}(k) \\
	&=J_+(x, k) e^{x \widehat{\mathcal{L}(k})} T_{n}(k), \quad x \in \mathbb{R}, k \in \bar{\Omega}_{n} \backslash \mathcal{Z},\quad n=1, 2, \ldots, 6.
\end{aligned}
$$
Combining the relationship between $J_+$ and $J_-$, the $S_n$ and $T_n$ can be linked by
$$
\Delta(k)=S_{n}(k) T_{n}^{-1}(k), \quad k \in \bar{\Omega}_{n} \backslash (\mathcal{Z}\cup\{0\}) .
$$
Since the Schwartz functions with compact support are dense in $\mathcal{S}(\R)$ with respect to the absolutely norm, one can asymptotically express the functions $M_n$, $J_{\pm}$ and $\Delta(k)$ by Schwartz initial potential functions.

\subsection{The jump matrix $v_n(x,k)$ }

\begin{lemma}
	Suppose $q_0,p_0\in\mathcal{S}(\R)$, then the matrix-valued functions $M_n(x,k)$ satisfies the boundary condition
$$
M_{+}(x, k)=M_{-}(x, k) v(x,  k), \quad k \in \Sigma \backslash(\mathcal{Z}\cup\{0\})
$$
where $v(x,  k)$ is jump matrix to be determined.
\par
In particular, when $q_0,p_0\in\mathcal{S}(\R)$ are compact support, there exists a matrix $\nu_1(k)$ such that
$$
M_1(x, k)=M_6(x, k) e^{x \widehat{k\Lambda}} \nu_1(k)
$$
and $M_n(x, k)=e^{x \widehat{\mathcal{L}(k)}} S_n(k)$ when $x$ is out of the compact support of $q_0,p_0$. Hence, we have
$$
\nu_1(k)=S_6(k)^{-1} S_1(k).
$$
\end{lemma}

In conclusion, this completes the jump function $v_n$.

 \begin{lemma}
 	Let $q_{0}, p_{0} \in \mathcal{S}(\mathbb{R})$, the eigenfunctions $M_{1}$ can be expressed in terms of the entries of $J_{\pm}, J_{\pm}^{A}, \Delta$, and $\Delta^{A}$ as follows:
$$
M_{1}=\left(\begin{array}{clc}
	J^{+}_{11} & \frac{(J^{-}_{31})^{A} (J^{+}_{23})^{A}-(J^{-}_{21})^{A} (J^{+}_{33})^{A}}{\delta_{11}} & \frac{J^{-}_{13}}{\delta_{33}^{A}} \\
	J^{+}_{21} & \frac{(J^{-}_{11})^{A} (J^{+}_{33})^{A}-(J^{-}_{31})^{A} (J^{+}_{13})^{A}}{\delta_{11}} & \frac{J^{-}_{23}}{\delta_{33}^{A}} \\
	J^{+}_{31} & \frac{(J^{-}_{21})^{A} (J^{+}_{13})^{A}-(J^{-}_{11})^{A} (J^{+}_{23})^{A}}{\delta_{11}} & \frac{J^{-}_{33}}{\delta_{33}^{A}}
\end{array}\right).
$$
Furthermore, for the $|k|$ small enough, the following property holds
$$
 C\left|M_{n}(x, k)-\sum_{l=-2}^{p} {M}_{n}^{(l)}(x) k^{l}\right| \leq|k|^{p+1}, \quad k \in \bar{\Omega}_{n}.
$$

 \end{lemma}
\par

\subsection{Construction of the Riemann-Hilbert problem}

Define the reflection coefficients $r_1(k)$ and $r_2(k)$ as
$$
\begin{cases}r_{1}(k)=\frac{\Delta_{12}(k)}{\Delta_{11}(k)}, & k \in(0, \infty), \\ r_{2}(k)=\frac{\Delta^{A}_{12}(k)}{\Delta^{A}_{11}(k)}, & k \in(-\infty, 0).\end{cases}
$$
\begin{proposition}
Suppose $q_{0}, p_{0} \in \mathcal{S}(\mathbb{R})$, then  $r_{1}:(0, \infty) \rightarrow \mathbb{C}$ and $r_{2}:(-\infty, 0) \rightarrow \mathbb{C}$ have the following properties:
$r_1(k)$ and $r_2(k)$ are smooth functions with rapidly decay as $|k|\to\infty $ on their domain and can be extended to $0$ as follows
$$
r_{1}(k)=r_{1}(0)+r_{1}^{\prime}(0) k+\frac{1}{2} r_{1}^{\prime \prime}(0) k^{2}+\cdots, \quad k \rightarrow 0,\quad k>0,
$$
and
$$
r_{2}(k)=r_{2}(0)+r_{2}^{\prime}(0) k+\frac{1}{2} r_{2}^{\prime \prime}(0) k^{2}+\cdots, \quad k \rightarrow 0,\quad k<0,
$$
where $r_1(0)=\omega,r_2(0)=1$.
\end{proposition}

\begin{remark}
The SK equation only has the term $q$ but without $p$ in the space Lax pair, which means the $Q_{2}$ is vanishing. In this case the behavior of $J_{\pm},\Delta,J^{A}_{\pm},\Delta^{A}$ as $k\to0$ only is a simple pole, moreover, the behavior of reflection coefficients $r_1$ and $r_2$ are also has  different value with $r_1(0)=\omega^2,r_2(0)=1$.
\end{remark}
Define the jump matrix $v_n(x,t;k)$ for $k\in\Sigma$ as
\begin{equation}\label{vn-jump-matrix}
\begin{aligned}
	&v_{1}=\left(\begin{array}{ccc}
		1 & -r_{1}(k) e^{-\theta_{21}} & 0 \\
		r_{1}^{*}(k) e^{\theta_{21}} & 1-\left|r_{1}(k)\right|^{2} & 0 \\
		0 & 0 & 1
	\end{array}\right),\\
	&v_{2}=\left(\begin{array}{ccc}
		1 & 0 & 0 \\
		0 & 1-r_{2}(\omega k) r_{2}^{*}(\omega k) & -r_{2}^{*}(\omega k) e^{-\theta_{32}} \\
		0 & r_{2}(\omega k) e^{\theta_{32}} & 1
	\end{array}\right),\\
	&v_{3}=\left(\begin{array}{ccc}
		1-r_{1}\left(\omega^{2} k\right) r_{1}^{*}\left(\omega^{2} k\right) & 0 & r_{1}^{*}\left(\omega^{2} k\right) e^{-\theta_{31}} \\
		0 & 1 & 0 \\
		-r_{1}\left(\omega^{2} k\right) e^{\theta_{31}} & 0 & 1
	\end{array}\right),\\
	&v_{4}=\left(\begin{array}{ccc}
		1-\left|r_{2}(k)\right|^{2} & -r_{2}^{*}(k) e^{-\theta_{21}} & 0 \\
		r_{2}(k) e^{\theta_{21}} & 1 & 0 \\
		0 & 0 & 1
	\end{array}\right),\\
	&v_{5}=\left(\begin{array}{ccc}
		1 & 0 & 0 \\
		0 & 1 & -r_{1}(\omega k) e^{-\theta_{32}} \\
		0 & r_{1}^{*}(\omega k) e^{\theta_{32}} & 1-r_{1}(\omega k) r_{1}^{*}(\omega k)
	\end{array}\right),\\
	&v_{6}=\left(\begin{array}{ccc}
		1 & 0 & r_{2}\left(\omega^{2} k\right) e^{-\theta_{31}} \\
		0 & 1 & 0 \\
		-r_{2}^{*}\left(\omega^{2} k\right) e^{\theta_{31}} & 0 & 1-r_{2}\left(\omega^{2} k\right) r_{2}^{*}\left(\omega^{2} k\right),
	\end{array}\right)
\end{aligned}
\end{equation}
where the terms $\theta_{i j}=\theta_{i j}(x, t, k)(1 \leq i \neq j \leq 3)$ are defined by $\theta_{i j}(x, t, k)=\left(l_i-l_j\right) x+$ $\left(z_i-z_j\right) t$ with $l_1(k)= \omega k,l_2=\omega^2k,l_3=k$ and $z_1(k)=9 \omega^2 k^5,z_2(k)=9 \omega k^5,z_1(k)=9  k^5$.

\subsection{Riemann-Hilbert problem} \label{Riemann-Hilbert-Problem}

Given $r_{1}(k)$ and $r_{2}(k)$, find a $3 \times 3$ matrix-valued function $M_n(x, t, k)$ with the following properties:

(a) $M_n(x, t, k): \mathbb{C} \backslash \Sigma \rightarrow \mathbb{C}^{3 \times 3}$ is analytic for $k\in \mathbb{C} \backslash \Sigma$.

(b) The limits of $ M_n(x, t, k)$ as $k$ approaches $\Sigma  $ from the left (+) and right (-) exist, are continuous on $\Sigma  $, and are related by
$$
M_{+}(x, t, k)=M_{-}(x, t, k) v(x, t, k), \quad k \in \Sigma,
$$
where $v$ is defined in terms of $r_{1}(k)$ and $r_{2}(k)$ by $(\ref{vn-jump-matrix})$.

(c) $M_n(x, t, k)=I+O\left(k^{-1}\right)$ as $k \rightarrow \infty,~ k \notin \Sigma $.

(d) $M_n(x, t, k)=\sum_{l=-2}^{p} {M}_{n}^{(l)}(x) k^{l}+O(k^{p+1})$ as $k \rightarrow 0$.

The reconstruction formula for the potential function $p(x, t)$ is
$$
p(x, t)=-3 \frac{\partial}{\partial x}\lim _{k \rightarrow \infty}k (M(x, t, k)_{33}-1).
$$
\par
Alternatively, the reconstruction formula for the solutions of the SK equation and KK equation can be expressed by
\begin{equation}\label{recover-formula}
u(x, t)= -\frac{1}{2}\frac{\partial}{\partial x}\lim _{k \rightarrow \infty}k (M(x, t, k)_{33}-1).
\end{equation}

\par	
\par	
\section{\bf Long-time asymptotics}

This section investigates the long-time asymptotics of the SK and KK equations (\ref{SK})-(\ref{KK}) by  Deift-Zhou steepest-descent method \cite{Deift-Zhou-1993}. Firstly, compute the critical points by standard way:
\begin{equation}
	\begin{split}
		\begin{aligned}
			\theta_{21}&=(\alpha^2-\alpha)kx+(\alpha-\alpha^2)9k^5t\\
			&=t[(\alpha^2-\alpha)k\xi+(\alpha-\alpha^2)9k^5]\\
			&:=t\Phi_{21}(k),\\
			\partial_{k}\Phi_{21}&=(\alpha^2-\alpha)\xi+(\alpha-\alpha^2)45k^4=0,\\
			k^4&=\frac{\xi}{45}.	
		\end{aligned}
	\end{split}	
\end{equation}
The same procedure shows that
\begin{equation}
	\begin{split}
		\begin{aligned}
       	\theta_{31}&=(1-\alpha)kx+(1-\alpha^2)9k^5t\\
    	&=t[(1-\alpha)k\xi+(1-\alpha^2)9k^5]\\
    	&:=t\Phi_{31}(k),\\
    	\partial_k\Phi_{31}&=(1-\alpha)\xi+(1-\alpha^2)9k^4=0,\\
    	k^4&=\frac{x}{45t}\alpha,
        \end{aligned}
	\end{split}
\end{equation}
and
\begin{align*}
	\theta_{32}&=(1-\alpha^2)kx+(1-\alpha)9k^5t\\
	&=t[(1-\alpha^2)k\xi+(1-\alpha)9k^5]\\
	&:=t\Phi_{32},\\
	\partial_k\Phi_{21}&=(1-\alpha^2)\xi+(1-\alpha)9k^4=0,\\
	k^4&=\frac{x}{45t}\alpha^2.
\end{align*}
\par
Denote $k_0=\sqrt[4]{\frac{|x|}{45t}}$ and in particular, for $x>0$ let
$$
k_0=\sqrt[4]{\frac{x}{45t}},
$$
where $\sqrt[4]{*}$ denotes the quartic root of non-negative real and $(*)^{\frac{1}{4}}$ is complex. Thus we have

for the $\Phi_{21}$, $k=({\frac{\xi}{45}})^{\frac{1}{4}}=\sqrt[4]{\frac{|\xi|}{45}}e^{\frac{n\pi i}{2}}$ , where $n=0,1,2,3$;

for the $\Phi_{31}$, $k=({\frac{|\xi|}{45}}\omega)^{\frac{1}{4}}=\sqrt[4]{\frac{|\xi|}{45}}e^{\frac{\pi i}{6}+\frac{n\pi i}{2}}$, where $n=0,1,2,3$;

for the $\Phi_{32}$, $k=({\frac{|\xi|}{45}}\omega^2)^{\frac{1}{4}}=\sqrt[4]{\frac{|\xi|}{45}}e^{\frac{\pi i}{3}+\frac{n\pi i}{2}}$ , where $n=0,1,2,3$.
\par
The critical points on $\Sigma$ are illustrated in Figure \ref{Critical-points}. In fact, the jump matrices $v(x,t,k)$ on each cut have the following symmetries:
$$
v(x,t,k)=\mathcal{A}v(x,t,\omega k)\mathcal{A}^{-1}=\mathcal{B}\overline {v(x,t,\bar k)}^{-1}\mathcal{B},\ \ k\in \Sigma.
$$
Thus the following transformations focus only on the points $\pm k_0$, furthermore, the other points can be obtained by the above symmetry.

\begin{figure}[!h]
	\centering
	\begin{overpic}[width=.65\textwidth]{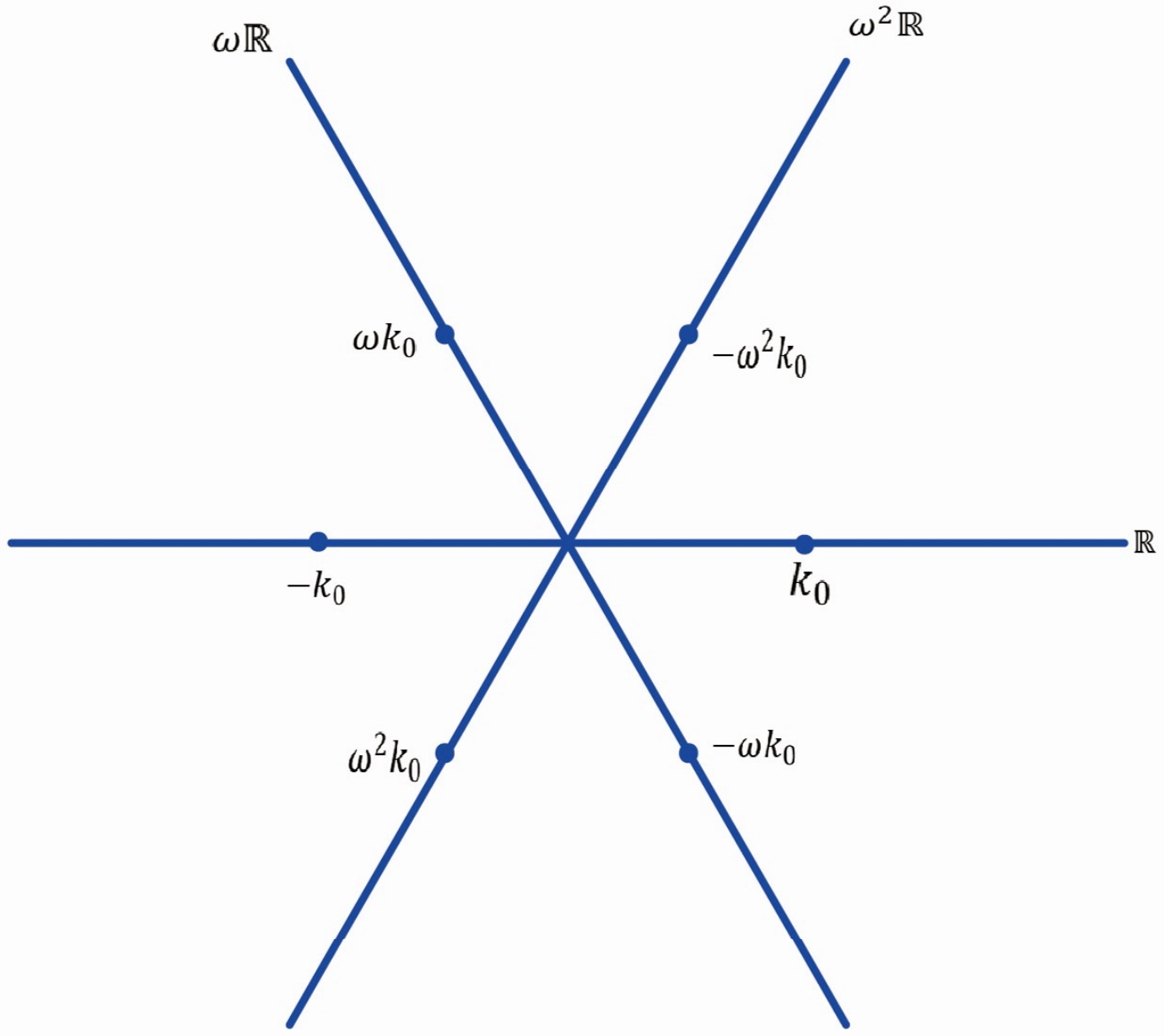}
	\end{overpic}
	\caption{{\protect\small
The jump contour $\Sigma$ and the critical points $\{\pm k_0,\pm \omega k_0,\pm \omega^2 k_0\}$.}}
    \label{Critical-points}
\end{figure}

\subsection{First transformation}

The jump matrices on $\R^+$ and $\R^-$ with direction away from the origin are
$$
v_{1}=\left(\begin{array}{ccc}
	1 & -r_{1}(k) e^{-\theta_{21}} & 0 \\
	r_{1}^{*}(k) e^{\theta_{21}} & 1-\left|r_{1}(k)\right|^{2} & 0 \\
	0 & 0 & 1
\end{array}\right),\\
$$
and
$$
v_{4}=\left(\begin{array}{ccc}
	1-\left|r_{2}(k)\right|^{2} & -r_{2}^{*}(k) e^{-\theta_{21}} & 0 \\
	r_{2}(k) e^{\theta_{21}} & 1 & 0 \\
	0 & 0 & 1
\end{array}\right).
$$
The sign diagram of ${\rm Re}(\theta_{21})$ is shown in Fig. \ref{sign-Figure}, where the shaded region denotes ${\rm Re}(\theta_{21})<0$, while the white region means ${\rm Re}(\theta_{21})>0$.
\begin{figure}
	\centering
	\begin{overpic}[width=.45\textwidth]{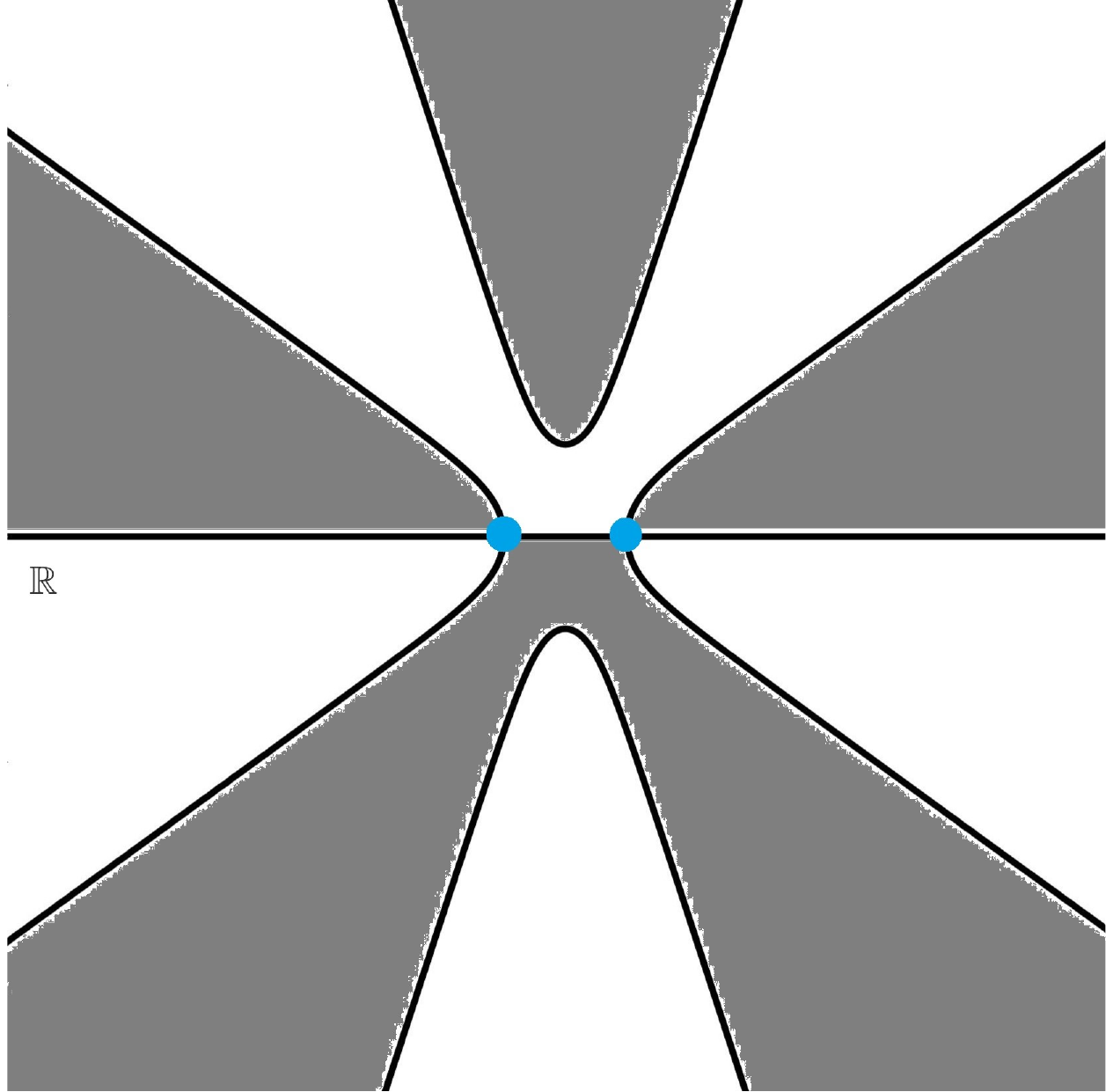}
		\put(58.6,30.5){\small $k_0$}
		\put(36.3,30.5){\small $-k_0$}
	\end{overpic}
	\caption{{\protect\small
		The sign diagram of ${\rm Re}(\theta_{21})$: the shaded region is ${\rm Re}(\theta_{21})<0$ and white is ${\rm Re}(\theta_{21})>0$.}}
	\label{sign-Figure}
\end{figure}
So decompose the jump matrices $v_1$ and $v_4$ on $0<|k|\le k_0$, respectively
$$
v_{1}=\left(\begin{array}{ccc}
	1 & -r_{1}(k) e^{-\theta_{21}} & 0 \\
	r_{1}^{*}(k) e^{\theta_{21}} & 1-\left|r_{1}(k)\right|^{2} & 0 \\
	0 & 0 & 1
\end{array}\right)=
\left(\begin{array}{ccc}
	1 & 0 & 0 \\
	r_{1}^{*}(k) e^{\theta_{21}} & 1& 0 \\
	0 & 0 & 1
\end{array}\right)
\left(\begin{array}{ccc}
	1 & -r_{1}(k) e^{-\theta_{21}} & 0 \\
	0 & 1& 0 \\
	0 & 0 & 1
\end{array}\right),
$$
and
$$
v_{4}=\left(\begin{array}{ccc}
	1-\left|r_{2}(k)\right|^{2} & -r_{2}^{*}(k) e^{-\theta_{21}} & 0 \\
	r_{2}(k) e^{\theta_{21}} & 1 & 0 \\
	0 & 0 & 1
\end{array}\right)=
\left(\begin{array}{ccc}
	1 & -r_{2}^{*}(k) e^{-\theta_{21}} & 0 \\
	0 & 1& 0 \\
	0 & 0 & 1
\end{array}\right)
\left(\begin{array}{ccc}
	1 & 0 & 0 \\
	r_{2}(k) e^{\theta_{21}} & 1& 0 \\
	0 & 0 & 1
\end{array}\right).
$$
\par
However, for $|k|\ge k_0$, a matrix-values function $\Theta$ should be introduced to open lenses naturally.
Since the reflect coefficients $r_1(k)$ and $r_2(k)$ rapidly decay, there exits $\mathcal{C}>0$ such that
$$
|r_j(k)|<1,\quad j=1,2,\quad k\in[\mathcal{C},\infty).
$$
Now, let $\delta_1(k),\delta_4(k)$ be the solution of the scalar RH problem
$$
\left\{\begin{aligned}
	\delta_{1+}(k) & =\delta_{1-}(k)\left(1-|r_1(k)|^2\right), & & k\in\Sigma_1^{(1)}, \\
	& =\delta_{1-}(z), & & k\in\C\setminus\Sigma_1^{(1)}, \\
	\delta_1(k) & \rightarrow 1 & & \text { as } k \rightarrow \infty,
\end{aligned}\right.
$$
and
$$
\left\{\begin{aligned}
	\delta_{4+}(k) & =\delta_{4-}(k)\left(1-|r_1(k)|^2\right), & & k\in\Sigma_4^{(1)}, \\
	& =\delta_{4-}(z), & & k\in\C\setminus\Sigma_4^{(1)}, \\
	\delta_4(k) & \rightarrow 1 & & \text { as } k \rightarrow \infty.
\end{aligned}\right.
$$
Hence, using the Plemelj formula, it follows that
$$
\delta_1(k)=\exp \left\{\frac{1}{2 \pi i} \int_{\left[k_0, \infty\right)} \frac{\ln \left(1-\left|r_1(s)\right|^2\right)}{s-k} d s\right\}, \quad k \in \mathbb{C} \backslash \Sigma_1^{(1)},
$$
and
$$
\delta_4(k)=\exp \left\{\frac{1}{2 \pi i} \int_{[-k_0,-\infty)} \frac{\ln \left(1-\left|r_2(s)\right|^2\right)}{s-k} d s\right\}, \quad k \in \mathbb{C} \backslash \Sigma_4^{(1)}.
$$

Let $\log_{\theta}(k)$ denote the logarithm of $k$ with branch cut along $\arg k=\theta$, i.e.,
$$
\log_0(k)=\ln|k|+\arg_0(k),  \arg_0(k)\in(0,2\pi);\\
\log_{\pi}(k)=\ln|k|+\arg_{\pi}(k),  \arg_{\pi}(k)\in(-\pi,\pi).
$$

\begin{proposition}
	The basic properties of $\delta$ functions are given below:
\par	
(1). If choosing the branch cut of $\ln$ with $\arg k\in(0,2\pi)$, the $\delta(k)$ can be rewritten as
$$
\delta_1( k)=e^{-i \nu_1 \log_0\left(k-k_0\right)} e^{-\chi_1( k)}
$$
where
$$
\nu_1=-\frac{1}{2 \pi} \ln \left(1-\left|r_1\left(k_0\right)\right|^2\right),
$$
and
$$
\chi_1( k)=\frac{1}{2 \pi i} \int_{k_0}^{\infty} \log_0(k-s) d \ln \left(1-\left|r_1(s)\right|^2\right).
$$
Moreover, we have
$$
\delta_4( k)=e^{-i \nu_4 \log_{\pi}\left(k+k_0\right)} e^{-\chi_4( k)}
$$
with
$$
\nu_4=-\frac{1}{2 \pi} \ln \left(1-\left|r_2\left(-k_0\right)\right|^2\right),
$$
and
$$
\chi_4( k)=\frac{1}{2 \pi i} \int_{-k_0}^{-\infty} \log_{\pi}(k-s) d \ln \left(1-\left|r_2(s)\right|^2\right).
$$
\par
(2). The $\delta(k)$ satisfies the conjugate symmetry and are bounded in $|k|>k_0$ such that
$$
\delta_1( k)={(\overline{\delta_1( \bar{k})})^{-1}}, \quad  k \in \mathbb{C} \backslash \Sigma_1^{(1)},\\
\delta_4( k)={(\overline{\delta_4( \bar{k})})^{-1}}, \quad  k \in \mathbb{C} \backslash \Sigma_4^{(1)},
$$
and
$$
|\delta_{1\pm}(k)|<\infty,\quad {\rm for}~~ k>k_0;\quad
|\delta_{4\pm}(k)|<\infty,\quad {\rm for}~~ k<-k_0.
$$

(3). As $k\to \pm k_0$ along with $k>k_0$ and $k<-k_0$, respectively, we have
$$
\begin{aligned}
	& \left|\chi_1( k)-\chi_1\left( k_0\right)\right| \leq C\left|k-k_0\right|\left(1+|\ln | k-k_0||\right), \\
	& \left|\partial_x\left(\chi_1( k)-\chi_1\left( k_0\right)\right)\right| \leq \frac{C}{t}\left(1+|\ln | k-k_0||\right),
\end{aligned}
$$
and
$$
\begin{aligned}
	& \left|\chi_4( k)-\chi_4\left( -k_0\right)\right| \leq C\left|k+k_0\right|\left(1+|\ln | k+k_0||\right), \\
	& \left|\partial_x\left(\chi_4( k)-\chi_4\left( k_0\right)\right)\right| \leq \frac{C}{t}\left(1+|\ln | k+k_0||\right).
\end{aligned}
$$

Moreover, we have
$$
\begin{gathered}
	\left|\partial_x \chi_1\left( k_0\right)\right|\leq \frac{C}{t} ,\quad
	\partial_x\left(\delta_1( k)^{\pm 1}\right)=\frac{\pm i \nu}{180 tk_0^3\left(k-k_0\right)} \delta_1( k)^{\pm 1}.
\end{gathered}
$$
\end{proposition}
\begin{proof}
	 We focus on the properties of $\delta_1$ and the properties of $\delta_2$ are similar.

(1). Recall that

$$
\delta_1(k)=\exp \left\{\frac{1}{2 \pi i} \int_{\left[k_0, \infty\right)} \frac{\ln \left(1-\left|r_1(s)\right|^2\right)}{s-k} d s\right\}, \quad k \in \mathbb{C} \backslash \Sigma_1^{(1)}.
$$

Using the technique of integrations by parts, rewrite the formula of $\delta_1$ as follows

\begin{align*}
	\delta_1(k)&=\exp \left\{\frac{1}{2 \pi i} \int_{\left[k_0, \infty\right)} \frac{\ln \left(1-\left|r_1(s)\right|^2\right)}{s-k} d s\right\}
	=\exp \left\{-\frac{1}{2 \pi i} \int_{\left[k_0, \infty\right)} \frac{\ln \left(1-\left|r_1(s)\right|^2\right)}{k-s} d s\right\}\\
	&=\exp \left\{\frac{1}{2 \pi i} \int_{\left[k_0, \infty\right)} {\ln \left(1-\left|r_1(s)\right|^2\right)} d\ (\ln(k-s)) \right\}\\
	&=\exp \left\{\frac{1}{2 \pi i} \left( {\ln \left(1-\left|r_1(s)\right|^2\right)} \ln(k-s)|_{k_0}^{\infty}-\int_{\left[k_0, \infty\right)}\ln(k-s) d\ {\ln \left(1-\left|r_1(s)\right|^2\right)}\right) \right\}\\
	&=\exp \left\{\frac{1}{2 \pi i} \left( -{\ln \left(1-\left|r_1(k_0)\right|^2\right)} \log_0(k-k_0)-\int_{\left[k_0, \infty\right)}\log_0(k-s) d\ {\ln \left(1-\left|r_1(s)\right|^2\right)}\right) \right\}\\
	&=e^{-i\nu_1\log_0(k-k_0)}e^{-\chi_1(k)},
\end{align*}
where we have chosen the branch cut $0<\arg(k-k_0)<2\pi$, since the jump contour is $k>k_0$.
\par
(2). By using the uniqueness of RH problem of $\delta_1(k)$, it is found that $\overline{\delta_1(\bar k)}^{-1}$ satisfies the same RH problem, so that $\delta_1(k)=\overline{\delta_1(\bar k)}^{-1}$.  When $k_0$ large enough, the reflection coefficients is strictly small than $1$ and inserting the symmetry of $\delta_1(k)$ yields
$$
\begin{aligned}
	\delta_{1+}(k)  &=\delta_{1-}(k)\left(1-|r_1(k)|^2\right),\\
	&=\overline{\delta_{1-}(\bar k)}^{-1}\left(1-|r_1(k)|^2\right),\\
	&=\overline{\delta_{1+}( k)}^{-1}\left(1-|r_1(k)|^2\right),\\
	|\delta_{1+}(k)|^2&=\left(1-|r_1(k)|^2\right),
\end{aligned}
$$
and the same procedure leads to
$$
|\delta_{1-}|^2\le\left(1-\sup_{k_0\in[c_0,M]}|r_1(k_0)|^2\right)^{-1}.
$$
In conclusion, $|\delta_{1\pm}|$ is bounded and hence, by the maximum principle $|\delta_1(k)|$ is bounded. Finally, using the fact of symmetry, it is concluded that $|\delta_1^{-1}(k)|$ is bounded.
\par
(3). The facts that $r_1(k)$ rapidly decays and the integral of $\log_0(k)$ is well defined near the zero indicate that for any $\epsilon>0$ one can estimate $\chi_1(k_0)$ as
$$
\begin{aligned}
	|\chi_1( k)-\chi_1( k_0)|&=\left|\frac{1}{2 \pi i} \int_{k_0}^{\infty} \log_0(k-s)-\log_0(k_0-s) d \ln \left(1-\left|r_1(s)\right|^2\right) \right|\\
	&\leq C \int_{k_0}^{k_0+\epsilon} \left|\log_0(k-s)-\log_0(k_0-s)\right| |d \ln \left(1-\left|r_1(s)\right|^2\right)|\\
	&+\int_{k_0+\epsilon}^{\infty} \left|\log_0(k-s)-\log_0(k_0-s)\right| |d \ln \left(1-\left|r_1(s)\right|^2\right)|,
\end{aligned}
$$
where for the first term, we have
$$
\begin{aligned}
	&\int_{k_0}^{k_0+\epsilon} \left|\log_0(k-s)-\log_0(k_0-s)\right| |d \ln \left(1-\left|r_1(s)\right|^2\right)|\\
	&=\int_{k_0}^{k_0+\epsilon} \left|\log_0(k-s)-\log_0(k_0-s)\right|\frac{\partial_s|r_1(s)|}{1-\left|r_1(s)\right|^2} ds\\
	&\le\|\frac{\partial_s|r_1(s)|}{1-\left|r_1(s)\right|^2}\|_{L^{\infty}(k_0,k_0+\epsilon)}\int_{k_0}^{k_0+\epsilon} \left|\log_0(k-s)-\log_0(k_0-s)\right| ds\\
	&\leq C |\left((s-k_0)\log_0(s-k_0)+(k_0-s)|_{k_0}^{k_0+\epsilon}\right)-\left((s-k)\log_0(s-k)+(k-s)|_{k_0}^{k_0+\epsilon}\right)|\\
	&=C|\epsilon log_0(\epsilon)-\epsilon-(k_0-k+\epsilon)\log_0(k_0-k+\epsilon)+(k_0-k+\epsilon)+(k_0-k)\log_0(k_0-k)-(k_0-k)|\\
	&\le|k_0-k|(1+|ln(k-k_0)|).
\end{aligned}
$$
and for the second term
$$
\begin{aligned}
	&\int_{k_0+\epsilon}^{\infty} \left|\log_0(k-s)-\log_0(k_0-s)\right| |d \ln \left(1-\left|r_1(s)\right|^2\right)\\
	&\le |k-k_0|\int_{k_0+\epsilon}^{\infty} |\frac{1}{\eta-s}|d \ln \left(1-\left|r_1(s)\right|^2\right) \le|k-k_0|(1+\ln(k_0-k)).
\end{aligned}
$$
Furthermore, for the critical point $k_0=\sqrt[4]{\frac{x}{45t}}$, by the estimates above and taking the derivative with respect to $x$, one can get the estimate for $\partial_x\chi(k)$.
\par
Using the symmetry properties, and considering $\delta_1$ and $\delta_4$, define $\delta_j~(j=2,3,5,6)$ as follows
$$
\begin{array}{ll}
	\delta_3( k)=\delta_1\left( \omega^2 k\right), & k \in \mathbb{C} \backslash \Sigma_3^{(1)}, \\
	\delta_5( k)=\delta_1( \omega k), & k \in \mathbb{C} \backslash \Sigma_5^{(1)},\\
	\delta_2( k)=\delta_4\left( \omega^2 k\right), & k \in \mathbb{C} \backslash \Sigma_2^{(1)}, \\
	\delta_6( k)=\delta_4( \omega k), & k \in \mathbb{C} \backslash \Sigma_6^{(1)},
\end{array}
$$
which satisfy the jump conditions
$$
\begin{array}{ll}
	\delta_{3+}( k)=\delta_{3-}( k)\left(1-\left|r_1\left(\omega^2 k\right)\right|^2\right), & k \in \Sigma_3^{(1)}, \\
	\delta_{5+}( k)=\delta_{5-}( k)\left(1-\left|r_1(\omega k)\right|^2\right), & k \in \Sigma_5^{(1)},\\
	\delta_{2+}( k)=\delta_{2-}( k)\left(1-\left|r_2\left(\omega^2 k\right)\right|^2\right), & k \in \Sigma_2^{(1)}, \\
	\delta_{6+}( k)=\delta_{6-}( k)\left(1-\left|r_2(\omega k)\right|^2\right), & k \in \Sigma_6^{(1)}.
\end{array}
$$
\end{proof}

\begin{remark}
	 The function $\ln(k)$ in $\delta_n(k)$ selects the branch cut $\frac{(n-1)\pi}{3}<\arg(k)<2\pi+\frac{(n-1)\pi}{3}$ for $n=1,2,3$ and the branch cut $-\frac{(7-n)\pi}{3}<\arg(k)<2\pi-\frac{(7-n)\pi}{3}$ for $n=4,5,6$.
\end{remark}

Thus define the matrix function $\Theta(k)$ to factorize the RH problem. In particular,
$$
M^{(1)}(x, t, k)=M(x, t, k) \Theta(k),
$$
where $\Theta(k)$ is
$$
\Theta(k)=\left(\begin{array}{ccc}
	\frac{\delta_1( k)\delta_6( k)}{\delta_3( k)\delta_4( k)} & 0 & 0 \\
	0 & \frac{\delta_5( k)\delta_4( k)}{\delta_1( k)\delta_2( k)} & 0 \\
	0 & 0 & \frac{\delta_3( k)\delta_2( k)}{\delta_5( k)\delta_6( k)}
\end{array}\right)=
\left(\begin{array}{ccc}
	\Theta_1(k) & 0 & 0 \\
	0 & \Theta_2(k) & 0 \\
	0 & 0 & \Theta_3(k)
\end{array}\right)
$$
with $\Theta( k)=I+O\left(k^{-1}\right) \text {as} k \rightarrow \infty$.
\par
The jump matrix is $v^{(1)}(x,t;k)=\Theta^{-1}_-v(x,t;k)\Theta_+$, and one can get for $|k|>k_0$ as follows
\begin{align*}
v_1^{(1)}&=\left(\begin{array}{ccc}
	\frac{\Theta_{1+}}{\Theta_{1-}} & -\frac{\Theta_{2+} }{\Theta_{1-}} r_1(k) e^{-t \Phi_{21}} & 0 \\
	\frac{\Theta_{1+} }{\Theta_{2-}} r_1^*(k) e^{t \Phi_{21}} & \frac{\Theta_2-}{\Theta_{2+}}\left(1-r_1(k) r_1^*(k)\right) & 0 \\
	0 & 0 & 1
\end{array}\right)
=\left(\begin{array}{ccc}
	\frac{\delta_{1+}}{\delta_{1-}} & -\frac{\tilde\delta_{v_1}}{\delta_{1-}\delta_{1+}} r_1(k) e^{-t \Phi_{21}} & 0 \\
	\frac{\delta_{1+}\delta_{1-} }{\tilde\delta_{v_1}} r_1^*(k) e^{t \Phi_{21}} & \frac{\delta_{1-}}{\delta_{1+}}\left(1-r_1(k) r_1^*(k)\right) & 0 \\
	0 & 0 & 1
\end{array}\right)\\
&=\left(\begin{array}{ccc}
	1-|r_1(k)|^2 & -\frac{\tilde\delta_{v_1}}{\delta_{1-}^2}\frac{r_1(k)}{1-|r_1(k)|^2} e^{-t \Phi_{21}} & 0 \\
	\frac{\delta_{1+}^2 }{\tilde\delta_{v_1}} \frac{r_1^*(k)}{1-|r_1(k)|^2} e^{t \Phi_{21}} & 1 & 0 \\
	0 & 0 & 1
\end{array}\right),
\end{align*}
and
$$
v_{4}^{(1)}=\left(\begin{array}{ccc}
	1 & -\frac{\delta_{4+}^2}{\tilde\delta_{v_4}}\frac{r_{2}^{*}(k)}{1-|r_2(k)|^2} e^{-t\Phi_{21}} & 0 \\
	\frac{\tilde\delta_{v_4}}{\delta_{4-}^2}\frac{r_{2}(k)}{1-|r_2(k)|^2} e^{t\Phi_{21}} & 1-\left|r_{2}(k)\right|^{2} & 0 \\
	0 & 0 & 1
\end{array}\right),
$$
where $\tilde\delta_{v_1}=\frac{\delta_3\delta^2_4\delta_5}{\delta_6\delta_2}$ and $\tilde\delta_{v_4}=\frac{\delta_1^2\delta_2\delta_6}{\delta_5\delta_3}$.

Furthermore, the other jump matrices can be obtained by the symmetry or direct calculations:

\begin{align*}
	&v_{3}^{(1)}=\left(\begin{array}{ccc}
	1 & 0 & \frac{\delta_{3+}^2}{\tilde\delta_{v_3}}\frac{r_{1}^{*}\left(\omega^{2} k\right)}{1-r_{1}\left(\omega^{2} k\right) r_{1}^{*}\left(\omega^{2} k\right)} e^{-\theta_{31}} \\
	0 & 1 & 0 \\
	-\frac{\tilde\delta_{v_3}}{\delta_{3-}^2}\frac{r_{1}\left(\omega^{2} k\right)}{1-r_{1}\left(\omega^{2} k\right) r_{1}^{*}\left(\omega^{2} k\right)} e^{\theta_{31}} & 0 & 1-r_{1}\left(\omega^{2} k\right) r_{1}^{*}\left(\omega^{2} k\right)
\end{array}\right),\\
&v_{5}^{(1)}=\left(\begin{array}{ccc}
	1 & 0 & 0 \\
	0 & 1-r_{1}(\omega k) r_{1}^{*}(\omega k) & -\frac{\tilde\delta_{v_5}}{\delta_{5-}^2}\frac{r_{1}(\omega k)}{1-r_{1}(\omega k) r_{1}^{*}(\omega k)} e^{-\theta_{32}} \\
	0 & \frac{\delta_{5+}^2}{\tilde\delta_{v_5}}\frac{r_{1}^{*}(\omega k)}{1-r_{1}(\omega k) r_{1}^{*}(\omega k)} e^{\theta_{32}} & 1
\end{array}\right),
\end{align*}

and

\begin{align*}
	&v_{2}^{(1)}=\left(\begin{array}{ccc}
	1 & 0 & 0 \\
	0 & 1 & -\frac{\delta_{2+}^2}{\tilde\delta_{v_2}}\frac{r_{2}^{*}(\omega k)}{1-r_{2}(\omega k) r_{2}^{*}(\omega k)} e^{-\theta_{32}} \\
	0 & \frac{\tilde\delta_{v_2}}{\delta_{2-}^2}\frac{r_{2}(\omega k)}{1-r_{2}(\omega k) r_{2}^{*}(\omega k)} e^{\theta_{32}} & 1-r_{2}(\omega k) r_{2}^{*}(\omega k)
\end{array}\right),\\
&v_{6}^{(1)}=\left(\begin{array}{ccc}
	1-r_{2}\left(\omega^{2} k\right) r_{2}^{*}\left(\omega^{2} k\right) & 0 & \frac{\tilde\delta_{v_6}}{\delta_{6-}^2}\frac{r_{2}\left(\omega^{2} k\right)}{1-r_{2}\left(\omega^{2} k\right) r_{2}^{*}\left(\omega^{2} k\right)} e^{-\theta_{31}} \\
	0 & 1 & 0 \\
	-\frac{\delta_{6+}^2}{\tilde\delta_{v_6}}\frac{r_{2}^{*}\left(\omega^{2} k\right)}{1-r_{2}\left(\omega^{2} k\right) r_{2}^{*}\left(\omega^{2} k\right)} e^{\theta_{31}} & 0 & 1
\end{array}\right).
\end{align*}

On the other hand, since on the $\Sigma_{j=7,8,\cdots,12}$ the $\delta_j$ has no jump, the jump matrix $v_{\{j=7,8,\cdots, 12\}}$ are just as follows
$$
v_7^{(1)}=\left(\begin{array}{ccc}
	\frac{\Theta_{1+}}{\Theta_{1-}} & -\frac{\Theta_{2+} }{\Theta_{1-}} r_1(k) e^{-t \Phi_{21}} & 0 \\
	\frac{\Theta_{1+} }{\Theta_{2-}} r_1^*(k) e^{t \Phi_{21}} & \frac{\Theta_2-}{\Theta_{2+}}\left(1-r_1(k) r_1^*(k)\right) & 0 \\
	0 & 0 & 1
\end{array}\right)\\
=\left(\begin{array}{ccc}
	1 & -\frac{\tilde\delta_{v1} }{\delta_{1}^2} r_1(k) e^{-t \Phi_{21}} & 0 \\
	\frac{\delta_{1}^2 }{\tilde\delta_{v1}} r_1^*(k) e^{t \Phi_{21}} & 1-r_1(k) r_1^*(k)& 0 \\
	0 & 0 & 1
\end{array}\right),\\
$$
and
$$
v_9^{(1)}=\left(\begin{array}{ccc}
	1-r_{1}\left(\omega^{2} k\right) r_{1}^{*}\left(\omega^{2} k\right) & 0 & \frac{\delta_{3}^2}{\tilde\delta_{v_3}}r_{1}^{*}\left(\omega^{2} k\right) e^{-\theta_{31}} \\
	0 & 1 & 0 \\
	-\frac{\tilde\delta_{v_3}}{\delta_{3}^2}r_{1}\left(\omega^{2} k\right) e^{\theta_{31}} & 0 & 1
\end{array}\right)\\
v_{11}^{(1)}=\left(\begin{array}{ccc}
	1 & 0 & 0 \\
	0 & 1 & -\frac{\tilde\delta_{v_5}}{\delta_{5}^2}{r_{1}(\omega k)} e^{-\theta_{32}} \\
	0 & \frac{\delta_{5}^2}{\tilde\delta_{v_5}}{r_{1}^{*}(\omega k)} e^{\theta_{32}} & 1-r_{1}(\omega k) r_{1}^{*}(\omega k)
\end{array}\right).
$$
Moreover, one has
$$
v_{8}^{(1)}=\left(\begin{array}{ccc}
	1 & 0 & 0 \\
	0 & 1-r_{2}(\omega k) r_{2}^{*}(\omega k) & -\frac{\delta_{2}^2}{\tilde\delta_{v_2}}{r_{2}^{*}(\omega k)} e^{-\theta_{32}} \\
	0 & \frac{\tilde\delta_{v_2}}{\delta_{2}^2}{r_{2}(\omega k)} e^{\theta_{32}} & 1
\end{array}\right),\\
v_{10}^{(1)}=\left(\begin{array}{ccc}
	1-\left|r_{2}(k)\right|^{2} & -\frac{\delta_{4}^2}{\tilde\delta_{v_4}}{r_{2}^{*}(k)} e^{-t\Phi_{21}} & 0 \\
	\frac{\tilde\delta_{v_4}}{\delta_{4}^2}{r_{2}(k)} e^{t\Phi_{21}} & 1 & 0 \\
	0 & 0 & 1
\end{array}\right),
$$
$$
v_{12}^{(1)}=\left(\begin{array}{ccc}
	1 & 0 & \frac{\tilde\delta_{v_6}}{\delta_{6}^2}{r_{2}\left(\omega^{2} k\right)} \\
	0 & 1 & 0 \\
	-\frac{\delta_{6}^2}{\tilde\delta_{v_6}}{r_{2}^{*}\left(\omega^{2} k\right)} e^{\theta_{31}} & 0 & 1-\left|r_{2}(k)\right|^{2}
\end{array}\right).
$$
\par
Next, suppose
$$
\rho_1(k)=\frac{r_1(k)}{1-r_1(k)r_1^{*}(k)},\quad \rho_2(k)=\frac{r_2(k)}{1-r_2(k)r_2^{*}(k)}.
$$
It is necessary to decompose the functions $r_j, r_j^*$ and $\rho_j(k)$ for analytic extension.

\begin{lemma}
	 The functions $r_j$ and $\rho_j~(j=1,2)$ have the following decompositions
$$
\begin{array}{ll}
	r_1(k)=r_{1, a}(x, t, k)+r_{1, r}(x, t, k), & k \in\left[0, k_0\right], \\
	r_2^*(k)=r_{2, a}^*(x, t, k)+r_{2, r}^*(x, t, k), & k \in[-k_0,0], \\
	\rho_1(k)=\rho_{1, a}(x, t, k)+\rho_{1, r}(x, t, k), & k \in\left[k_0, \infty\right),\\
	\rho_2^*(k)=\rho_{2, a}^*(x, t, k)+\rho_{2, r}^*(x, t, k), & k \in\left( -\infty,k_0\right].
\end{array}
$$
Fig. \ref{regionU.pdf} shows the sign diagram of ${\rm Re}(\Phi_{21})$, where $U_1\cup U_3\cup U_5=\{k: {\rm Re}(\Phi_{21})<0\}$ and $U_2\cup U_4\cup U_6=\{k: {\rm Re}(\Phi_{21})>0\}$.
\end{lemma}

\begin{figure}[!h]
	\centering
	\begin{overpic}[width=.45\textwidth]{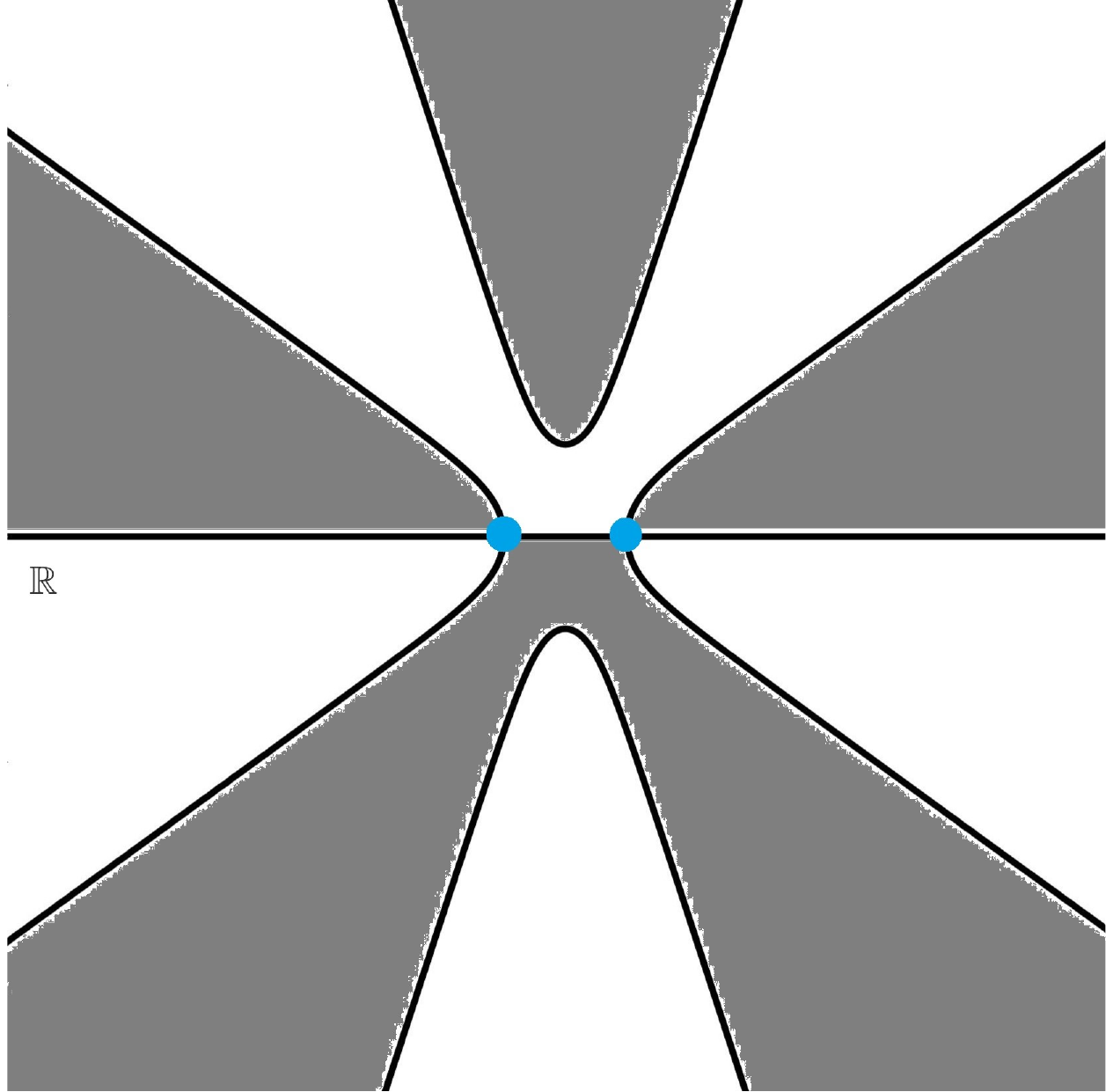}
		\put(63.6,37.9){\small $U_1$}
		\put(48.5,37.9){\small $U_2$}
		\put(33.5,37.9){\small $U_3$}
		\put(63.6,30.3){\small $U_6$}
		\put(48.5,30.3){\small $U_5$}
		\put(33.5,30.3){\small $U_4$}
	\end{overpic}
	\caption{{\protect\small
			The shaded region is ${\rm Re}(\Phi_{21})<0$ and white is ${\rm Re}(\Phi_{21})>0$.}}
	\label{regionU.pdf}
\end{figure}

Furthermore, the decomposition functions have the following properties:

(1). For $0<k_0<M$, $r_{1,a}$ and $r_{2,a}^*$ are defined and continuous for $k\in \bar U_2$  and analytic for $k\in U_2$, but the domain for $r_{1,a},r_{2,a}^*$ are restricted by $0\le\operatorname{Re}(k)\le k_0$ and $-k_0\le\operatorname{Re}(k)\le 0$, respectively. The function $\rho_{1,a}^*$ is defined and continuous for $k\in \bar U_1$  and analytic for $k\in U_1$, $\rho_{2,a}$ are defined and continuous for $k\in \bar U_3$  and analytic for $k\in U_3$.

(2). For $0<k_0<M$, the functions $r_{1,a},r_{2,a}^*$ and $\rho_{1,a}^*,\rho_{2,a}$ satisfy the following estimates:

$$
\left|\partial_x^l(r_{1, a}(x, t, k)-r_1(0))\right| \leq C|k|{e^{-t\operatorname{Re} \Phi_{21}( k)/2}},\quad k\in \bar U_2~ \text{and}~ {\rm Re}(k)\in[0,k_0];\\
$$
$$
\left|\partial_x^l(r_{1, a}(x, t, k)-r_1(k_0))\right| \leq C|k-k_0|{e^{-t\operatorname{Re} \Phi_{21}( k)/2}},\quad k\in \bar U_2\ \text{and}~ {\rm Re}(k)\in[0,k_0];\\
$$
and
$$
\left|\partial_x^l(r_{2, a}^*(x, t, k)-r_2^*(0))\right| \leq C|k|{e^{-t\operatorname{Re} \Phi_{21}( k)/2}},\quad k\in \bar U_2 ~\text{and}~\operatorname{Re}(k)\in[-k_0,0];\\
$$
$$
\left|\partial_x^l(r_{2, a}^*(x, t, k)-r_2^*(-k_0))\right| \leq C|k+k_0|{e^{-t\operatorname{Re} \Phi_{21}( k)/2}},\quad k\in \bar U_2\ \text{and}\ \operatorname{Re}(k)\in[-k_0,0].\\
$$
Furthermore,
$$
\left|\partial_x \rho_{1, a}(x, t, k)\right| \leq \frac{Ce^{-t\operatorname{Re} \Phi_{21}( k)/2}}{1+|k|}\quad k\in \bar U_6,\\
\left|\partial_x(\rho_{1, a}(x, t, k)-\rho_1(k_0))\right| \leq {C|k-k_0|e^{-t\operatorname{Re} \Phi_{21}( k)/2}} \quad k \in \bar U_6.
$$
and
$$
\left|\partial_x \rho_{2, a}(x, t, k)\right| \leq \frac{Ce^{t\operatorname{Re} \Phi_{21}( k)/2}}{1+|k|}\quad k\in \bar U_3,\\
\left|\partial_x(\rho_{2, a}(x, t, k)-\rho_2(-k_0))\right| \leq {C|k+k_0|e^{t\operatorname{Re} \Phi_{21}( k)/2}} \quad k \in \bar U_3.
$$

(3). For $0<k_0<M$ and $1\le p\le\infty$,  the functions  $r_{j,r} $ and $\rho_{j,r}$ satisfy

$$
\left\|\partial_x^l r_{1, r}(x, t, k)e^{-t\Phi_{21}}\right\|_{L^p{(0,k_0)}} \leq \frac{c}{t^{3/2}}\quad 0<k<k_0,\\
\left\|\partial_x^l r_{2, r}^*(x, t, k)e^{-t\Phi_{21}}\right\|_{L^p(-k_0,0)} \leq \frac{c}{t^{3/2}}\quad -k_0<k<0,
$$

and
$$
\left\|(1+|\cdot|)\partial_x^l\rho_{1, r}(x, t, \cdot)e^{-t\Phi_{21}}\right\|_{L^p[k_0,\infty)} \leq \frac{c}{t^{3/2}},\\
\left\|(1+|\cdot|)\partial_x^l\rho_{2, r}(x, t, \cdot)e^{t\Phi_{21}}\right\|_{L^p(-\infty,-k_0]} \leq \frac{c}{t^{3/2}}.
$$
Proof.

Recall that the reflection coefficient $r_1(k)$ is smooth for $k\in[0,\infty)$, so that we can take Taylor expansion about $r_1(k)$ near $k_0$ point:
$$
r_1(k)=r_1(k_0)+\partial_kr_1|_{k=k_0}(k-k_0)+\cdots+\frac{1}{n!}\int_{k_0}^k \partial_k^{(n+1)}r_1(\tau)(k-\tau)^nd\tau.
$$
Let $R(k)=\sum_{i=0}^n\partial_k^ir_1|_{k=k_0}(k-k_0)^i$  and $h(k):=r_1(k)-R(k)$, now we split the $h(k)$ into $h(k)=h_1(k)+h_2(k)$.

Set
$$
\beta(k)=(k-k_0)^q.
$$
Rewrite $\theta_{21}$ as
$$
\theta_{21}=2it\left(\frac{9\sqrt{3}k^5-\sqrt{3}k\xi}{2}\right):=2it\theta(\xi,k)
$$
For $k\in[0,k_0]$, As $k\to\theta(\xi,k)$ is one-to-one, we define
$$
\left\{
\begin{aligned}
	\frac{h}{\beta}(\theta)&:=\frac{h(\theta(k))}{\beta(\theta(k))}, &&-18\sqrt{3}k_0^5=\theta(k_0)<\theta<0,\\
	&:=0,&&\text{Otherwise}.
\end{aligned}
\right.
$$
And rewrite $h(k)$ as
$$
\begin{aligned}
	h(k)&=\frac{1}{n!}\int_{k_0}^k \partial_k^{(n+1)}r_1(\tau)(k-\tau)^nd\tau\\
	&=\frac{(k-k_0)^{n+1}}{n!}\int_{0}^1 \partial_k^{(n+1)}r_1(k_0+\mu(k-k_0))(1-\mu)^nd\mu
\end{aligned}
$$
then
$$
\frac{h}{\beta}(\theta)=(k-k_0)^{n-q+1}g(k_0,k),\quad \theta(k_0)<\theta<0,
$$
where
$$
g(k_0,k)=\frac{1}{n!}\int_{0}^1 \partial_k^{(n+1)}r_1(k_0+\mu(k-k_0))(1-\mu)^nd\mu.
$$
Since the reflect coefficient $r_1$ belong to Schwartz space, so that $|\partial_kg(k_0,k)|\leq\infty$, and as $\frac{dk}{d\theta}=\frac{2}{45\sqrt{3}(k^4-k_0^4)}$ and $0<k_0<M$, we can estimate that
$$
\begin{aligned}
	\int_{0}^{k_0}\left| \left(\frac{d}{d\theta}\right)^j\frac{h}{\beta}(k)\right|^2|\bar d\theta(k)|
	&\leq C\int_{0}^{k_0}\left| \left(\frac{d}{dk}\frac{1}{|k^4-k_0^4|}\right)^j\frac{h}{\beta}(k)\right|^2|k^4-k_0^4|\bar dk\\
	&= C\int_{0}^{k_0}\left| \left(\frac{1}{|k^4-k_0^4|}\frac{d}{dk}\right)^j(k-k_0)^{n-q+1}g(k_0,k)\right|^2|k^4-k_0^4|\bar dk\\
	&\leq C\int_{0}^{k_0}\left| (k-k_0)^{2n-2q-4j+3}\right|\bar dk<\infty,
\end{aligned}
$$
for the $0\leq j \le\frac{2n-2q+4}{4}$.
\par
Thus one can conclude that $\frac{h}{\beta}(\theta)\in H^j(\R)$ for $0\leq j \le\frac{2n-2q+4}{4}$, and one can define the Fourier inversion transform
$$
\frac{h}{\beta}(k)=\int_{-\infty}^{\infty}e^{is\theta(k)}\widehat{\left(\frac{h}{\beta}\right)}(s)\bar ds,\quad k\in [0,k_0],
$$
where
$$
\widehat{\left(\frac{h}{\beta}\right)}(s)=-\int_0^{k_0}e^{-is\theta(k)}\left(\frac{h}{\beta}\right)(k)\bar{d}\theta(k),\quad s\in \R.
$$
Because $\frac{h}{\beta}(\theta)\in H^j(\R)$ and by the Plancherel formula, we have that
$$
\int_{-\infty}^{\infty}(1+s^2)^j\left|\widehat{\left(\frac{h}{\beta}\right)}(s)\right|^2ds\le C.
$$
Now split the $h(k)$ as
$$
{h}(k)={\beta}(k)\int_{t}^{\infty}e^{is\theta(k)}\widehat{\left(\frac{h}{\beta}\right)}(s)\bar ds+{\beta}(k)\int_{-\infty}^{t}e^{is\theta(k)}\widehat{\left(\frac{h}{\beta}\right)}(s)\bar ds\\
:=h_1(k)+h_2(k),\quad k\in [0,k_0].
$$
Then for $k\in[0,k_0]$, one has

\begin{equation}
	\begin{split}
		\begin{aligned} |e^{-2it\theta(k)}h_1(k)|&=\left|e^{-2it\theta(k)}{\beta}(k)\int_{t}^{\infty}e^{is\theta(k)}\widehat{\left(\frac{h}{\beta}\right)}(s)\bar ds\right|\\
	&\leq \left|e^{-it\theta(k)}{\beta}(k)\right|\int_{t}^{\infty}\left|e^{i(s-t)\theta(k)}\widehat{\left(\frac{h}{\beta}\right)}(s)\right|\bar ds\\
	&=|\beta(k)|\int_{t}^{\infty}\left|\widehat{\left(\frac{h}{\beta}\right)}(s)\right|\bar ds\\
&\leq C\left(\int_{t}^{\infty}\left|(1+s^2)^j\widehat{\left(\frac{h}{\beta}\right)}(s)\right|^2 ds\right)^{\frac{1}{2}}\left(\int_{t}^{\infty}\left|(1+s^2)^{-j}\right|^2 ds\right)^{\frac{1}{2}}\\
	&\leq \frac{C}{t^{j-\frac{1}{2}}}.	
   \end{aligned}
	\end{split}
\end{equation}

On the other hand, for $k\in U_2$ and $\operatorname{Re}(k)\in[0,k_0]$, we have
$$
\begin{aligned}
	|e^{-2it\theta(k)}h_2(k)|&=\left|e^{-2it\theta(k)}{\beta}(k)\int_{-\infty}^{t}e^{is\theta(k)}\widehat{\left(\frac{h}{\beta}\right)}(s)\bar ds\right|\\
	&=\left|e^{-it\theta(k)}{\beta}(k)\int_{-\infty}^{t}e^{i(s-t)\theta(k)}\widehat{\left(\frac{h}{\beta}\right)}(s)\bar ds\right|\\
	&\le C\left|e^{-it\theta(k)}{\beta}(k)\right|\int_{-\infty}^{t}e^{(s-t)\operatorname{Re}(i\theta(k))}\left|\widehat{\left(\frac{h}{\beta}\right)}(s)\right| ds\\
	&\le Ce^{-t\operatorname{Re}(i\theta(k))}.
\end{aligned}
$$
Moreover, by directly computation one can get  $|e^{-2it\theta(k)}R(k)|\le Ce^{-2t\operatorname{Re}(i\theta(k))}$, for $k\in U_2$ and $\operatorname{Re}(k)\in[0,k_0]$.
\par
Finally, we can define $r_{1,a}(k):=h_2(k)+R(k)$ which can be analytically extent to the $U_2$ with $\operatorname{Re}(k)\in[0,k_0]$ and $r_{1,r}(k):=h_1(k)$ define on $0\le k\le k_0$.
\par
Take $n=2,q=1$ for the estimate above, then we can get that
$$
R(k)=r_1(k_0)+\partial_kr_1(k_0)(k-k_0)+\partial_k^2r_1(k_0)(k-k_0)^2,\\
h_2(k)=(k-k_0)\int_{-\infty}^{t}e^{is\theta(k)}\widehat{\left(\frac{h}{\beta}\right)}(s)\bar ds,
$$
moreover,
$$
|r_{1,a}(k)-r_1(k_0)|\le C|k-k_0|e^{-\operatorname{Re}\theta_{21}(k)}.
$$
By the directly computation, one can get
$$
|\partial_x(r_{1,a}(k)-r_1(0))|\leq C|k-k_0|e^{-\operatorname{Re}\theta_{21}(k)}.
$$
The estimate between $r_{1,a}(k)$ and $r_1(0)$ need us to make Taylor expansion near $k=0$, the procedure is similar, so we omitted it.

Again, by the Taylor's expansion, we have that
$$
(k+i)^{n+5}\rho_1^*(k)=\gamma_0+\gamma_1(k-k_0)+\cdots+\gamma_n(k-k_0)^n+\frac{1}{n!}
\int_{k_0}^k((\cdot+i)\rho_1(\cdot))^{(n+1)}(\mu)(k-\mu)^nd\mu,
$$
and define
$$
R(k)=\frac{\sum_{i=0}^{n}\gamma_i(k-k_0)^i}{(k+i)^{n+5}},h(k)=\rho_1^*(k)-R(k).
$$
Set
$$
\beta(k)=\frac{(k-k_0)^q}{(k+i)^{q+2}}.
$$
By the Fourier transform
$$
\frac{h}{\beta}=\int_{-\infty}^{\infty}e^{is\theta(k)}\widehat{\left(\frac{h}{\beta}\right)}(s)\bar ds,
$$
where
$$
\widehat{\frac{h}{\beta}}(s)=\int_{k_0}^{\infty}e^{-is\theta(k)}{\left(\frac{h}{\beta}\right)}(k)\bar d\theta(k).
$$
As the same procedure, we estimate that
$$
\frac{h}{\beta}=\frac{(k-k_0)^{n-q+1}}{(k+i)^{n-q+3}}g(k_0;k).
$$
Here
$$
g(k_0,k)=\frac{1}{n!}\int_{0}^1 ((\cdot+i)^{n+5}\rho_1(\cdot))^{(n+1)}(k_0+\mu(k-k_0))(1-\mu)^nd\mu
$$
and $|\partial_kg(k_0,k)|<\infty$.

Now, we obtain that
$$
\begin{aligned}
	\int_{k_0}^{\infty}\left| \left(\frac{d}{d\theta}\right)^j\frac{h}{\beta}(k)\right|^2|\bar d\theta(k)|
	&\leq C\int_{k_0}^{\infty}\left| \left(\frac{d}{dk}\frac{1}{|k^4-k_0^4|}\right)^j\frac{h}{\beta}(k)\right|^2|k^4-k_0^4|\bar dk\\
	&= C\int_{k_0}^{\infty}\left| \left(\frac{1}{|k^4-k_0^4|}\frac{d}{dk}\right)^j\frac{(k-k_0)^{n-q+1}}{(k+i)^{n-q+3}}g(k_0,k)\right|^2|k^4-k_0^4|\bar dk\\
	&\leq C\int_{k_0}^{\infty}\left| \frac{(k-k_0)^{2n-2q-4j+3}}{(k+i)^{2n-2q+6}}\right|\bar dk<\infty,
\end{aligned}
$$
for $j\le\frac{n-q+2}{2}$, and we can conclude that $\frac{h}{\beta}\in H^j(k_0,\infty]$.

By the Plancherel formula, we can get that
$$
\int_{-\infty}^{\infty}(1+s^2)^j\left|\widehat{\left(\frac{h}{\beta}\right)}(s)\right|^2ds\le C.
$$
As the before process, split $h$ into two part
$$
{h}(k)={\beta}(k)\int_{t}^{\infty}e^{is\theta(k)}\widehat{\left(\frac{h}{\beta}\right)}(s)\bar ds+{\beta}(k)\int_{-\infty}^{t}e^{is\theta(k)}\widehat{\left(\frac{h}{\beta}\right)}(s)\bar ds\\
:=h_1(k)+h_2(k),\quad k\in [k_0,\infty).
$$
For $k\ge k_0$, we have that
$$
|e^{-it\theta(k)}h_1(k)|\le\frac{C}{|k+i|^2t^{j-\frac{1}{2}}},
$$
on the other hand, for $k\in U_6$,
$$
|e^{-2it\theta(k)}h_2(k)|\le\frac{|k-k_0|^q e^{\operatorname{Re}(-it\theta(k))}}{|k+i|^{q+2}}.
$$
Let $\rho_{1,a}(x,t;k)=R(k)+h_2(k)$ and $\rho_{1,r}=h_1(k)$. Compare with the same procedure, we can obtain the estimate in the lemma.
\par
\begin{figure}[!h]
	\centering
	\includegraphics[scale=0.4]{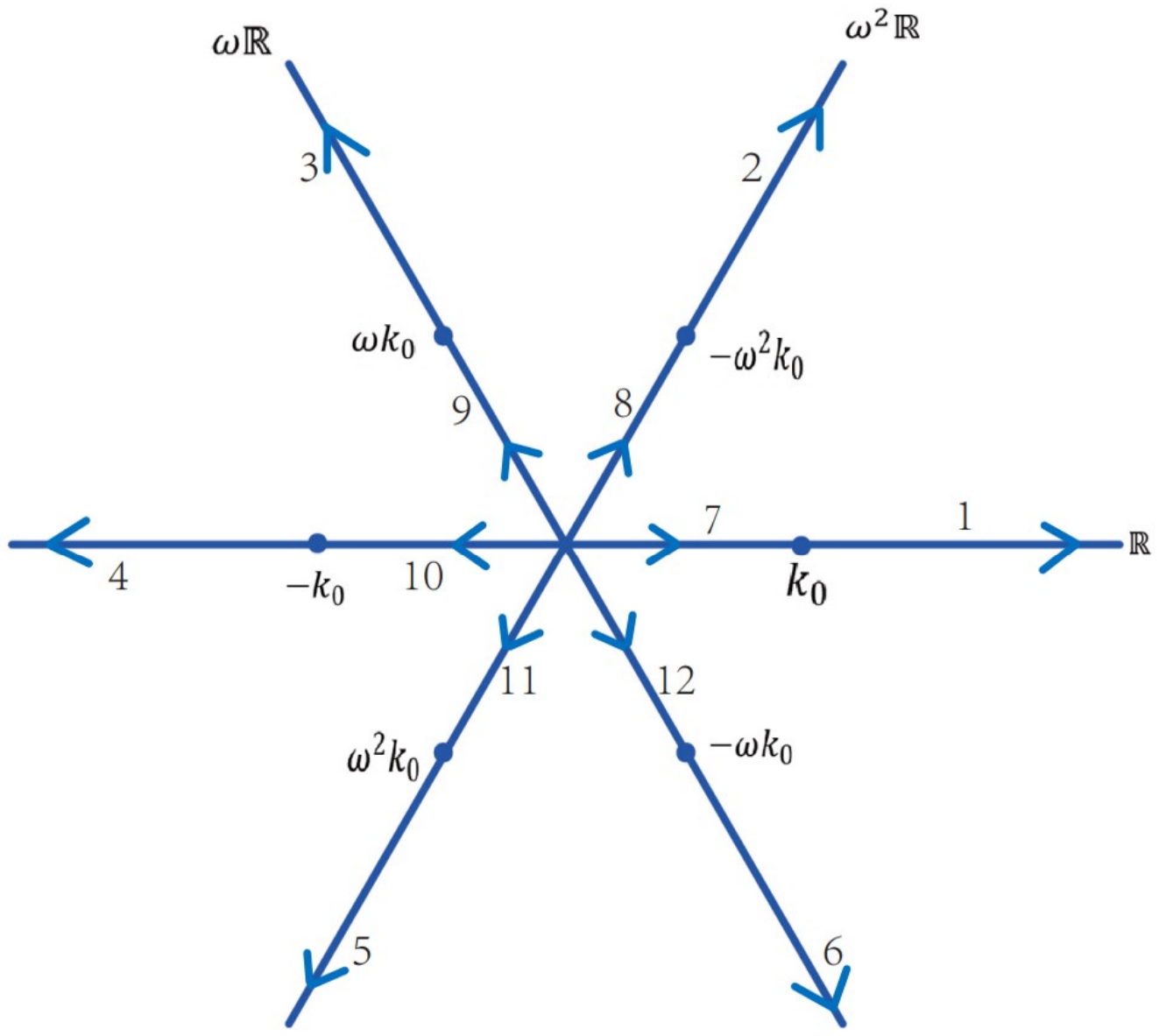}
	\caption{{\protect\small
			The first transformation of RH problem.}}
	\label{first-transform.pdf}
\end{figure}

\subsection{The second transform.}

Now, we focus on splitting the jump at the critical points $\pm k_0$ by recalling that
$$
v_1^{(1)}(x,t;k)=\left(\begin{array}{ccc}
	1-|r_1(k)|^2 & -\frac{\tilde\delta_{v_1}}{\delta_{1-}^2}\frac{r_1(k)}{1-|r_1(k)|^2} e^{-t \Phi_{21}} & 0 \\
	\frac{\delta_{1+}^2 }{\tilde\delta_{v_1}} \frac{r_1^*(k)}{1-|r_1(k)|^2} e^{t \Phi_{21}} & 1 & 0 \\
	0 & 0 & 1
\end{array}\right)=v^{(1)}_{1,lower}\ v_{1,r}^{(1)}\ v^{(1)}_{1,upper},
$$
where
$$
v^{(1)}_{1,lower}=\left(\begin{array}{ccc}
	1 & -\frac{\tilde\delta_{v_1}}{\delta_{1-}^2}\rho_{1,a} e^{-t \Phi_{21}} & 0 \\
	0 & 1 & 0 \\
	0 & 0 & 1
\end{array}\right),\quad
v^{(1)}_{1,upper}=\left(\begin{array}{ccc}
	1 & 0 & 0 \\
	\frac{\delta_{1+}^2 }{\tilde\delta_{v_1}} \rho^*_{1,a} e^{t \Phi_{21}} & 1 & 0 \\
	0 & 0 & 1
\end{array}\right),
$$
and
$$
v_{1,r}^{(1)}(x,t;k)=\left(\begin{array}{ccc}
	1-\frac{\delta_{+}^2}{\delta_{1-}^2}\rho_{1,r}(k)\rho^*_{1,r}(k) & -\frac{\tilde\delta_{v_1}}{\delta_{1-}^2}\rho_{1,r} e^{-t \Phi_{21}} & 0 \\
	\frac{\delta_{1+}^2 }{\tilde\delta_{v_1}} \rho_{1,r}^*  e^{t \Phi_{21}} & 1 & 0 \\
	0 & 0 & 1
\end{array}\right).
$$
Moreover, we have
$$
v_{7}^{(1)}=\left(\begin{array}{ccc}
	1 & -\frac{\tilde\delta_{v1} }{\delta_{1}^2} r_1(k) e^{-t \Phi_{21}} & 0 \\
	\frac{\delta_{1}^2 }{\tilde\delta_{v1}} r_1^*(k) e^{t \Phi_{21}} & 1-r_1(k) r_1^*(k)& 0 \\
	0 & 0 & 1
\end{array}\right)=v^{(1)}_{7,lower}\ v_{7,r}^{(1)}\ v^{(1)}_{7,upper},
$$
where
$$
v^{(1)}_{7,lower}=\left(\begin{array}{ccc}
	1 & 0 & 0 \\
	\frac{\delta_{1}^2 }{\tilde\delta_{v1}}r^*_{1,a} e^{t \Phi_{21}} & 1 & 0 \\
	0 & 0 & 1
\end{array}\right),\quad
v^{(1)}_{7,upper}=\left(\begin{array}{ccc}
	1 & -\frac{\tilde\delta_{v1} }{\delta_{1}^2}r_{1,a} e^{-t \Phi_{21}} & 0 \\
	0 & 1 & 0 \\
	0 & 0 & 1
\end{array}\right),
$$
and
$$
v_{7,r}^{(1)}=\left(\begin{array}{ccc}
	1 & -\frac{\tilde\delta_{v1} }{\delta_{1}^2} r_{1,r}(k) e^{-t\Phi_{21}} & 0 \\
	\frac{\delta_{1}^2 }{\tilde\delta_{v1}}r_{1,r}^{*}(k) e^{t\Phi_{21}} & 1-r_{1,r}(k)r^*_{1,r}(k) & 0 \\
	0 & 0 & 1
\end{array}\right).
$$
The same procedure yields
$$
v_{4}^{(1)}=\left(\begin{array}{ccc}
	1 & -\frac{\delta_{4+}^2}{\tilde\delta_{v_4}}\frac{r_{2}^{*}(k)}{1-|r_2(k)|^2} e^{-t\Phi_{21}} & 0 \\
	\frac{\tilde\delta_{v_4}}{\delta_{4-}^2}\frac{r_{2}(k)}{1-|r_2(k)|^2} e^{t\Phi_{21}} & 1-\left|r_{2}(k)\right|^{2} & 0 \\
	0 & 0 & 1
\end{array}\right)=v^{(1)}_{4,upper}\ v_{4,r}^{(1)}\ v^{(1)}_{4,lower},
$$
where
$$
v^{(1)}_{4,upper}=\left(\begin{array}{ccc}
	1 & 0 & 0 \\
	\frac{\tilde\delta_{v_4}}{\delta_{4-}^2}\rho_{2,a} e^{t\Phi_{21}} & 1 & 0 \\
	0 & 0 & 1
\end{array}\right),\quad
v^{(1)}_{4,lower}=\left(\begin{array}{ccc}
	1 & -\frac{\delta_{4+}^2}{\tilde\delta_{v_4}}\rho_{2,a}^{*} e^{-t\Phi_{21}} & 0 \\
	0 & 1 & 0 \\
	0 & 0 & 1
\end{array}\right),
$$
and
$$
v_{4,r}^{(1)}=\left(\begin{array}{ccc}
	1 & -\frac{\delta_{4+}^2}{\tilde\delta_{v_4}}\rho^*_{2,r} e^{-t\Phi_{21}} & 0 \\
	\frac{\tilde\delta_{v_4}}{\delta_{4-}^2}\rho_{2,r} e^{t\Phi_{21}} & 1-\frac{\delta_{4+}^2}{\delta_{4-}^2}\rho_{2,r}\rho^*_{2,r} & 0 \\
	0 & 0 & 1
\end{array}\right).
$$
On the other hand, one has
$$
v_{10}^{(1)}=\left(\begin{array}{ccc}
	1-\left|r_{2}(k)\right|^{2} & -\frac{\delta_{4}^2}{\tilde\delta_{v_4}}r_{2}^{*}(k) e^{-t\Phi_{21}} & 0 \\
	\frac{\tilde\delta_{v_4}}{\delta_{4}^2}r_{2}(k) e^{t\Phi_{21}} & 1 & 0 \\
	0 & 0 & 1
\end{array}\right)=v^{(1)}_{10,upper}\ v_{10,r}^{(1)}\ v^{(1)}_{10,lower},
$$
with
$$
v^{(1)}_{10,upper}=\left(\begin{array}{ccc}
	1 & -\frac{\delta_{4}^2}{\tilde\delta_{v_4}}r_{2,a}^* e^{-t \Phi_{21}} & 0 \\
	0 & 1 & 0 \\
	0 & 0 & 1
\end{array}\right),\quad
v^{(10)}_{1,lower}=\left(\begin{array}{ccc}
	1 & 0 & 0 \\
	\frac{\tilde\delta_{v_4}}{\delta_{4}^2}r_{2,a} e^{t \Phi_{21}} & 1 & 0 \\
	0 & 0 & 1
\end{array}\right),
$$
and
$$
v_{10,r}^{(1)}=\left(\begin{array}{ccc}
	1-r_{2,r}(k)r^*_{2,r}(k) & -\frac{\delta_{4}^2}{\tilde\delta_{v_4}}r_{2,r}^{*}(k) e^{-t\Phi_{21}} & 0 \\
	\frac{\tilde\delta_{v_4}}{\delta_{4}^2}r_{2,r}(k) e^{t\Phi_{21}} & 1 & 0 \\
	0 & 0 & 1
\end{array}\right).
$$

The $\Sigma_2$ is just shown in Fig. \ref{second-transform.pdf}.
\begin{figure}[!h]
	\centering
	\includegraphics[scale=0.75]{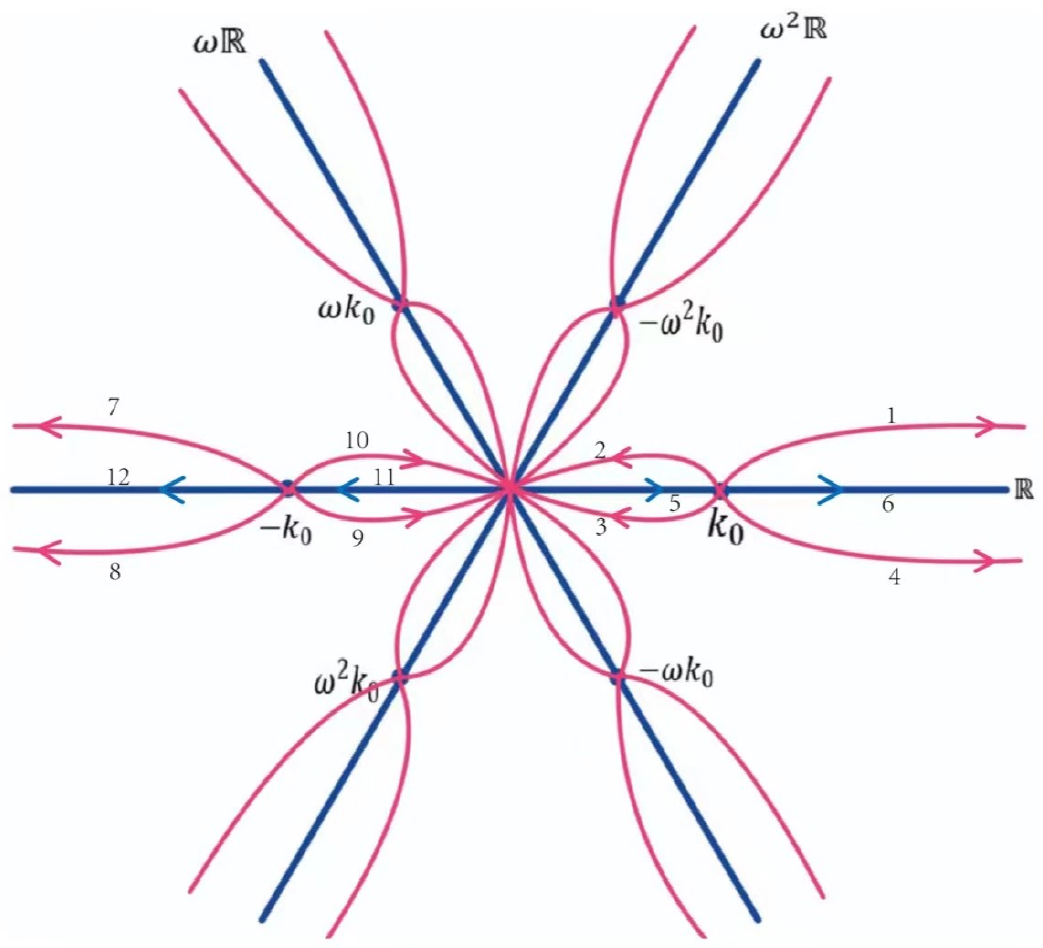}
	\caption{{\protect\small
			The second transformation of RH problem.}}
	\label{second-transform.pdf}
\end{figure}

The sets $V_j$ here are defined as in Fig.\ref{Vsets.pdf}.
\begin{figure}[!h]
	\centering
	\includegraphics[scale=0.50]{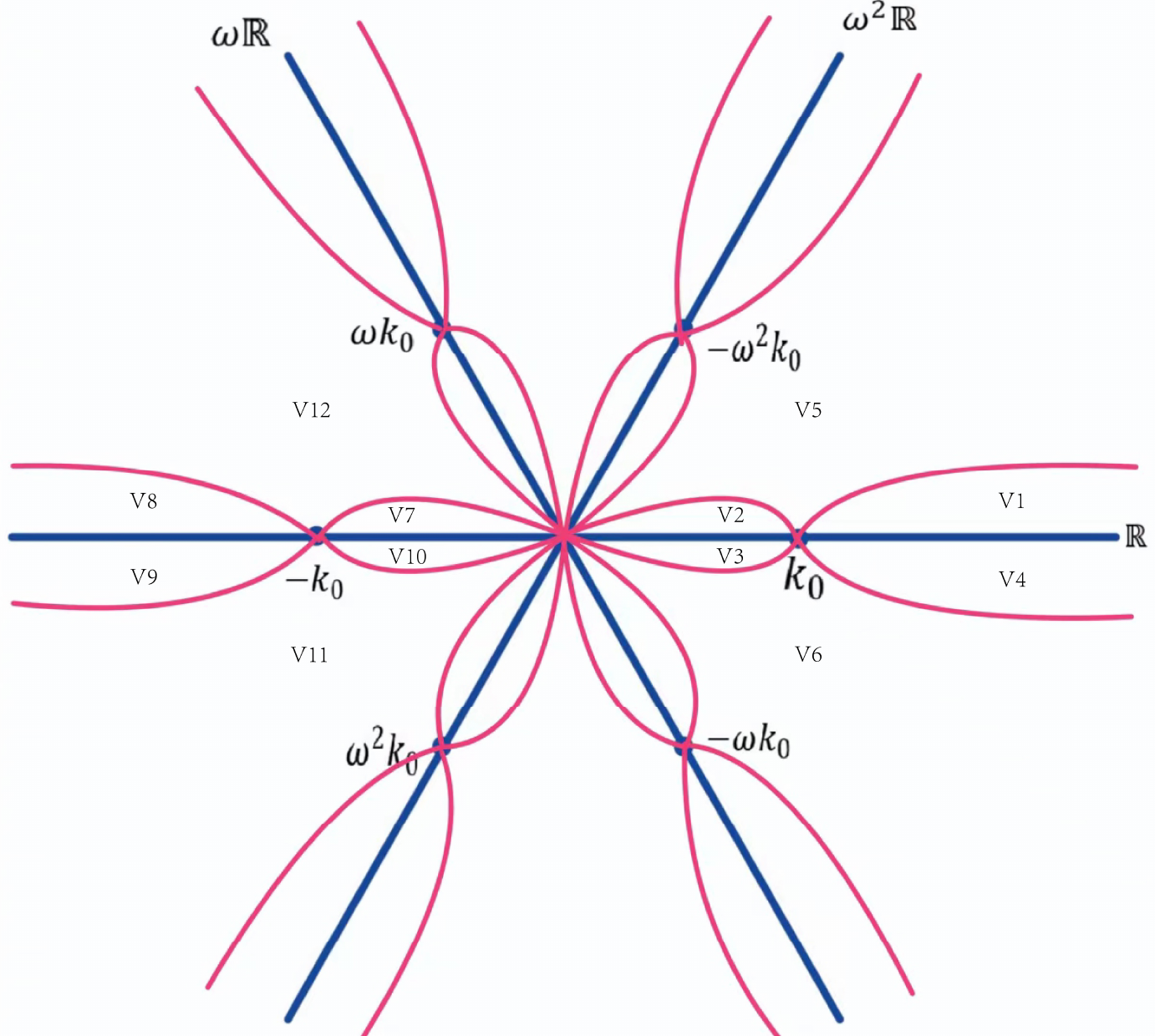}
	\caption{{\protect\small
			The sets $V_j$ in the complex plane.}}
	\label{Vsets.pdf}
\end{figure}

Now, introducing the intermediate function $G_j$ to factorize the RH problem for $M^{(1)}$ in to RH problem for $M^{(2)}$. In particular, define
$$
G_1(x, t, k)= \begin{cases}\left(v_{1,upper}^{(1) }\right)^{-1}, & k \in V_1, \\ \left(v_{7,upper}^{(1) }\right)^{-1}, & k \in V_2, \\ v_{7,lower}^{(1) }, & k \in V_3, \\ v_{1,lower}^{(1) }, & k \in V_4, \\ I, & k\in V_5\cup V_6,\end{cases}
$$
and
$$
G_2(x, t, k)= \begin{cases}v_{10,upper}^{(1) }, & k \in V_7, \\ v_{4,upper}^{(1) }, & k \in V_8, \\ \left(v_{4,lower}^{(1) }\right)^{-1}, & k \in V_9, \\ \left(v_{10,lower}^{(1) }\right)^{-1}, & k \in V_{10}, \\ I, & k\in V_{11}\cup V_{12}.\end{cases}
$$
Furthermore, the RH problems $M^{(1)}$ and $M^{(2)}$ can be related by
$$
M^{(2)}=M^{(1)}G(x,t;k),
$$
where the other $G_j$ can be defined by the symmetries.
\par
The jump matrix $v^{(2)}$ associated with the RH problems for $M^{(2)}$ is given by
$$
\begin{aligned}
	v^{(2)}_1&=v^{(1)}_{1,upper}=\left(\begin{array}{ccc}
		1 & 0 & 0 \\
		\frac{\delta_{1+}^2 }{\tilde\delta_{v_1}} \rho^*_{1,a} e^{t \Phi_{21}} & 1 & 0 \\
		0 & 0 & 1
	\end{array}\right),\\
	v^{(2)}_2&=\left(v^{(1)}_{7,upper}\right)^{-1}=\left(\begin{array}{ccc}
		1 & \frac{\tilde\delta_{v1} }{\delta_{1}^2}r_{1,a} e^{-t \Phi_{21}} & 0 \\
		0 & 1 & 0 \\
		0 & 0 & 1
	\end{array}\right),\\
	v^{(2)}_3&=\left(v^{(1)}_{7,lower}\right)^{-1}=\left(\begin{array}{ccc}
		1 & 0 & 0 \\
		-\frac{\delta_{1}^2 }{\tilde\delta_{v1}}r^*_{1,a} e^{t \Phi_{21}} & 1 & 0 \\
		0 & 0 & 1
	\end{array}\right),\\
	v^{(2)}_4&=v^{(1)}_{1,lower}=\left(\begin{array}{ccc}
		1 & -\frac{\tilde\delta_{v_1}}{\delta_{1-}^2}\rho_{1,a} e^{-t \Phi_{21}} & 0 \\
		0 & 1 & 0 \\
		0 & 0 & 1
	\end{array}\right),\\
	v^{(2)}_5&=v^{(1)}_{7,r}=\left(\begin{array}{ccc}
		1 & -\frac{\tilde\delta_{v1} }{\delta_{1}^2} r_{1,r}(k) e^{-t\Phi_{21}} & 0 \\
		\frac{\delta_{1}^2 }{\tilde\delta_{v1}}r_{1,r}^{*}(k) e^{t\Phi_{21}} & 1-r_{1,r}(k)r^*_{1,r}(k) & 0 \\
		0 & 0 & 1
	\end{array}\right),\\
	v^{(2)}_6&=v^{(1)}_{1,r}=\left(\begin{array}{ccc}
		1-\frac{\delta_{+}^2}{\delta_{1-}^2}\rho_{1,r}(k)\rho^*_{1,r}(k) & -\frac{\tilde\delta_{v_1}}{\delta_{1-}^2}\rho_{1,r} e^{-t \Phi_{21}} & 0 \\
		\frac{\delta_{1+}^2 }{\tilde\delta_{v_1}} \rho_{1,r}^*  e^{t \Phi_{21}} & 1 & 0 \\
		0 & 0 & 1
	\end{array}\right).\\
\end{aligned}
$$
Moreover,
$$
\begin{aligned}
	v^{(2)}_7&=v^{(1)}_{4,upper}=\left(\begin{array}{ccc}
		1 & 0 & 0 \\
		\frac{\tilde\delta_{v_4}}{\delta_{4-}^2}\rho_{2,a} e^{t\Phi_{21}} & 1 & 0 \\
		0 & 0 & 1
	\end{array}\right),\\
	v^{(2)}_8&=v^{(1)}_{4,lower}=\left(\begin{array}{ccc}
		1 & -\frac{\delta_{4+}^2}{\tilde\delta_{v_4}}\rho_{2,a}^{*} e^{-t\Phi_{21}} & 0 \\
		0 & 1 & 0 \\
		0 & 0 & 1
	\end{array}\right),\\
	v^{(2)}_9&=\left(v^{(1)}_{10,lower}\right)^{-1}=\left(\begin{array}{ccc}
		1 & 0 & 0 \\
		-\frac{\tilde\delta_{v_4}}{\delta_{4}^2}r_{2,a} e^{t \Phi_{21}} & 1 & 0 \\
		0 & 0 & 1
	\end{array}\right),\\
	v^{(2)}_{10}&=\left(v^{(1)}_{10,upper}\right)^{-1}=\left(\begin{array}{ccc}
		1 & \frac{\delta_{4}^2}{\tilde\delta_{v_4}}r_{2,a}^* e^{-t \Phi_{21}} & 0 \\
		0 & 1 & 0 \\
		0 & 0 & 1
	\end{array}\right),\\
	v^{(2)}_{11}&=v^{(1)}_{10,r}=\left(\begin{array}{ccc}
		1-r_{2,r}(k)r^*_{2,r}(k) & -\frac{\delta_{4}^2}{\tilde\delta_{v_4}}r_{2,r}^{*}(k) e^{-t\Phi_{21}} & 0 \\
		\frac{\tilde\delta_{v_4}}{\delta_{4}^2}r_{2,r}(k) e^{t\Phi_{21}} & 1 & 0 \\
		0 & 0 & 1
	\end{array}\right),\\
	v^{(2)}_{12}&=v^{(1)}_{4,r}=\left(\begin{array}{ccc}
		1 & -\frac{\delta_{4+}^2}{\tilde\delta_{v_4}}\rho^*_{2,r} e^{-t\Phi_{21}} & 0 \\
		\frac{\tilde\delta_{v_4}}{\delta_{4-}^2}\rho_{2,r} e^{t\Phi_{21}} & 1-\frac{\delta_{4+}^2}{\delta_{4-}^2}\rho_{2,r}\rho^*_{2,r} & 0 \\
		0 & 0 & 1
	\end{array}\right).\\
\end{aligned}
$$

In addition, the other jump matrices can be computed by the symmetry.

\begin{lemma}
The function $G_j$ is uniformly bounded for $k\in\C\setminus\Sigma_j$, and $G_j=I+O(\frac{1}{k})$ as $k\to\infty$.
\end{lemma}

\begin{proof}
	 We focus on the region $V_1$ and $V_2$. In the region $V_2$, $G_1=\left(v_{1,upper}^{(1)}\right)^{-1}$ so that it  suffices to show that  $\frac{\delta_{1+}^2 }{\tilde\delta_{v_1}} \rho^*_{1,a} e^{t \Phi_{21}}$ is bounded in $V_1$. Recall that $\delta_j$ is bounded in $\C\setminus\Sigma_j$ and $|\rho_{1,a} e^{t\Phi_{21}}|$ satisfied the lemma before, then $\frac{\delta_{1+}^2 }{\tilde\delta_{v_1}} \rho^*_{1,a} e^{t \Phi_{21}}$ is uniformly bounded. Nevertheless, the region $V_2$ is compact and $G_1$ is continuous, so $V_2$ is uniformly bounded.
\end{proof}

Again, recall the reconstruction formula that
\begin{align*}
u(x,t)&=-\frac{1}{2}\frac{\partial}{\partial x}\lim _{k \rightarrow \infty}k [M^{(1)}(x, t, k)-I]_{33}\\
&=-\frac{1}{2}\frac{\partial}{\partial x}\lim _{k \rightarrow \infty}k [M^{(1)}(x, t, k)G(x,t;k)-I]_{33}\\
&=-\frac{1}{2}\frac{\partial}{\partial x}\lim _{k \rightarrow \infty}k [M^{(2)}(x, t, k)-I]_{33}.
\end{align*}

\begin{lemma}
	For $0<k_0<M$ and any $\epsilon>0$, the jump functions $v^{(2)}$ converges uniformly to $I$ as $t\to\infty$ and $\partial_xv^{(2)}$ uniformly converges to the zero matrix expect for the points near the saddle points, i.e., $\left\{\pm k_0,\pm \omega k_0,\pm \omega^2k_0\right\}$. In particular, the jump matrix $v^{(2)}$ on $\Sigma_{5,6}$ have the following estimates:
$$
\|(1+|\cdot|)\partial_x^l(v^{(2)}-I)\|_{(L^1\cap L^{\infty})(\Sigma_{5,6}^{(2)})}\le Ct^{-\frac{3}{2}}.
$$
Moreover, using the fact of symmetric about jump matrix, we can get the similar estimate on the other $\Sigma_j^{(2)}$.
\end{lemma}

\begin{proof}
	We focus on the jump matrix on $\Sigma_{1,\cdots,6}^{(2)}$ and since for $k\in\Sigma_{1,2,3,4}^{(2)}$ the exponential part $\operatorname{Re}{t\Phi_{21}}$ is strictly less than $0$ for $\Sigma_{1,3}^{(2)}$ and strictly bigger than $0$ for $\Sigma_{2,4}^{(2)}$, except for the points near the saddle point $k_0$ (since $\operatorname{Re}t\Phi_{21}(k_0)=0$). Using the lemma about the properties of $r_{1,a},\rho_{1,a}$ and the bound of $\delta$ functions , we can conclude that $v^{(2)}_{1,2,3,4}$ ($\partial_x v_{1,2,3,4}^{(2)}$) converges to $I$ (resp. to the $0$ matrix) as $t\to\infty$.

Finally, we have
$$
(v_5^{(2)}-I)_{12}=-\frac{\tilde\delta_{v1} }{\delta_{1}^2} r_{1,r}(k) e^{-t\Phi_{21}},\quad
(v_6^{(2)}-I)_{12}=-\frac{\tilde\delta_{v_1}}{\delta_{1-}^2}\rho_{1,r} e^{-t \Phi_{21}},
$$
the lemma of $\delta$ and $r_{j,r},\rho_{j,r}$ implies that
$$
|(v_5^{(2)}-I)_{12}|\le Ct^{-\frac{3}{2}},\quad |(v_6^{(2)}-I)_{12}|\le Ct^{-\frac{3}{2}}.
$$
Moreover, the $(v_5^{(2)}-I)_{22}$ and $(v_6^{(2)}-I)_{22}$ can be written as a time of two part  $(v^{(2)}_{j}-I)_{12}$ and $(v^{(2)}_{j}-I)_{21}$, so the estimate is smaller than off-diagonal term.

By directly computation, we can conclude that
$$
\|(1+|\cdot|)\partial_x^l(v^{(2)}-I)\|_{(L^1\cap L^{\infty})(\Sigma_{5,6}^{(2)})}\le Ct^{-\frac{3}{2}}.
$$
\end{proof}

\subsection{The third transformation  }

In the transformations above, we have a new RH problem with the properties $v\to I$ as $t\to\infty$ for $k\in \{\R,\omega\R,\omega^2\R\},\ 0<k_0<M$ and other jump matrix on $\Sigma_j$ toward to $I$ as $t\to \infty$, but except for $k\in B(\pm k_0,{\epsilon})\cup B(\pm \omega k_0,{\epsilon})\cup B(\pm\omega^2 k_0,{\epsilon})$.

In order to factorize the RH problem for $M^{(2)}$ into model problem, we focus on the $\Sigma_A$ and $\Sigma_B$ where
$$
\Sigma_A=\Sigma_{\{1,2,3,4\}}^{(2)}\cap B_{\epsilon}(k_0),\ \Sigma_B=\Sigma_{\{7,8,9,10\}}^{(2)}\cap B_{\epsilon}(-k_0)
$$
Observe that the exponential part in the jump matrix on $\Sigma_A$ and $\Sigma_B$ are $\pm t\Phi_{21}$ and on the left contour $\Sigma_A$, we expand $t\Phi_{21}$ at $k_0$ just as follows
$$
\begin{aligned}
	t\Phi_{21}(k)&=t[(\alpha^2-\alpha)k\zeta+(\alpha-\alpha^2)9k^5]
	=9t(\alpha-\alpha^2)(k^5-5kk_0^4)\\
	&=9\sqrt{3}it[(k-k_0)^5+5k_0(k-k_0)^4+10k_0^2(k-k_0)^3+10k_0^3(k-k_0)^2-4k_0^5].
\end{aligned}
$$
Suppose $z_1={3^{\frac{5}{4}}2\sqrt{5t}k_0^{\frac{3}{2}}}(k-k_0)$, then rewrite $t\Phi_{21}$ as

$$
\begin{aligned}
	t\Phi_{21}(k)
	&=9\sqrt{3}ita^3[a^2z^5+5ak_0z^4+10k_0^2z^3]+\frac{iz^2}{2}+t\Phi_{21}(k_0)\\
	&=t\Phi_{21}^{0}(k_0,z)+\frac{iz_1^2}{2}+t\Phi_{21}(k_0),
\end{aligned}
$$
where $a=\frac{1}{{3^{\frac{5}{4}}2\sqrt{5t}k_0^{\frac{3}{2}}}}$.
\par
The other part of jump matrix functions on $\Sigma_A$ involves the $\delta$ function as
$$
\delta_1( k)=e^{-i \nu_1 \log_0\left(k-k_0\right)} e^{-\chi_1( k)},\quad k\in\C\setminus[k_0,\infty),
$$
where
$$
\nu_1=-\frac{1}{2 \pi} \ln \left(1-\left|r_1\left(k_0\right)\right|^2\right),
$$
and
$$
\chi_1( k)=\frac{1}{2 \pi i} \int_{k_0}^{\infty} \log_0(k-s) d \ln \left(1-\left|r_1(s)\right|^2\right) .
$$
Again, rewrite it as
$$
\begin{aligned}
	\frac{\delta_{1+}^{2}(k)}{\delta_{\tilde v_1}(k)}	
	&=e^{-2i \nu_1 \log_0\left(z\right)}\frac{a^{-2i\nu_1}e^{-2\chi_1(k_0)}}{\tilde\delta_{ v_1}(k_0)}
	\frac{e^{2\chi_1(k_0)-2\chi_1(k)}\tilde\delta_{v_1}(k_0)}{\tilde\delta_{ v_1}(k)}\\
	&:=e^{-2i \nu_1 \log_0\left(z\right)}\delta_A^0\delta_A^1,
\end{aligned}
$$
where $\delta_A^0=\frac{a^{-2i\nu}e^{-2\chi_1(k_0)}}{\delta_{\tilde v_1}(k_0)}$ and $\delta_A^1=\frac{e^{2\chi_1(k_0)-2\chi_1(k)}\delta_{\tilde v_1}(k_0)}{\delta_{\tilde v_1}(k)}$.

In the other hand, on the jump contour $\Sigma_B$, expand $t\Phi_{21}$ at $-k_0$ just as follows
$$
\begin{aligned}
	t\Phi_{21}(k)&=t[(\alpha^2-\alpha)k\zeta+(\alpha-\alpha^2)9k^5]
	=9t(\alpha-\alpha^2)(k^5-5kk_0^4)\\
	&=9\sqrt{3}it[(k+k_0)^5-5k_0(k+k_0)^4+10k_0^2(k+k_0)^3-10k_0^3(k+k_0)^2+4k_0^5].
\end{aligned}
$$
Suppose $z_2={3^{\frac{5}{4}}2\sqrt{5t}k_0^{\frac{3}{2}}}(k+k_0)$, then we can rewrite $t\Phi_{21}$ as
$$
\begin{aligned}
	t\Phi_{21}(k)
	&=9\sqrt{3}ita^3[a^2z^5_2-5ak_0z^4_2+10k_0^2z^3_2]-\frac{iz^2_2}{2}+t\Phi_{21}(-k_0)\\
	&=t\Phi_{21}^{0}(-k_0,z)-\frac{iz_2^2}{2}+t\Phi_{21}(-k_0).
\end{aligned}
$$
Moreover, the $\delta$ functions on the $\Sigma_B$ involves $\delta_4$, and
$$
\delta_4( k)=e^{-i \nu_4 \log_{\pi}\left(k+k_0\right)} e^{-\chi_4( k)},\quad k\in\C\setminus(-\infty,-k_0],
$$
with
$$
\nu_4=-\frac{1}{2 \pi} \ln \left(1-\left|r_2\left(-k_0\right)\right|^2\right),
$$
and
$$
\chi_4( k)=\frac{1}{2 \pi i} \int_{-k_0}^{-\infty} \log_{\pi}(k-s) d \ln \left(1-\left|r_2(s)\right|^2\right).
$$
In addition, we have
$$
\begin{aligned}
	\frac{\tilde\delta_{v_4}}{\delta_{4}^2}&=e^{2i \nu_4 \log_{\pi}\left(z_2\right)} \frac{\tilde\delta_{v_4}(-k_0)}{a^{-2i\nu_4}e^{-2\chi_4( -k_0)}}\frac{\tilde\delta_{ v_4}(k)}{e^{2\chi_4(-k_0)-2\chi_4(k)}\tilde\delta_{ v_4}(-k_0)}\\
	&:=e^{2i \nu_4 \log_{\pi}\left(z_2\right)}\left(\delta_B^0\right)^{-1}\left(\delta_B^1\right)^{-1},
\end{aligned}
$$
here $\delta_B^0=\frac{a^{-2i\nu_4}e^{-2\chi_4(-k_0)}}{\delta_{\tilde v_4}(-k_0)}$ and $\delta_B^1=\frac{e^{2\chi_4(-k_0)-2\chi_4(k)}\delta_{ v_4}(-k_0)}{\delta_{ \tilde v_4}(k)}$.

Now, define $H(\pm k_0,t)$ and transform the RH problem for $M^{(2)}$ of the form
$$
M^{(3,\epsilon)}=M^{(2)}(x,t;k)H(\pm k_0,t),\quad k\in B_{\epsilon}(\pm k_0),
$$
where
$$
H(k_0,t)=\left(\begin{array}{ccc}
	\left(\delta_A^{0}\right)^{-\frac{1}{2}} e^{-\frac{t}{2} \Phi_{21}\left( k_0\right)} & 0 & 0 \\
	0 & \left(\delta_A^{0}\right)^{\frac{1}{2}} e^{\frac{t}{2} \Phi_{21}\left( k_0\right)} & 0 \\
	0 & 0 & 1
\end{array}\right),
$$
and
$$
H(-k_0,t)=\left(\begin{array}{ccc}
	\left(\delta_B^{0}\right)^{\frac{1}{2}} e^{-\frac{t}{2} \Phi_{21}\left( -k_0\right)} & 0 & 0 \\
	0 & \left(\delta_B^{0}\right)^{-\frac{1}{2}} e^{\frac{t}{2} \Phi_{21}\left( -k_0\right)} & 0 \\
	0 & 0 & 1
\end{array}\right),
$$
so that we can derive the jump matrix on $\Sigma_A$ as follows

$$
\begin{aligned}
	v^{(3,\epsilon)}_1&=\left(\begin{array}{ccc}
		1 & 0 & 0 \\
		e^{-2i \nu_1 \log_0\left(z\right)}\delta_A^1 \rho^*_{1,a} e^{t \Phi_{21}^0(k_0,z)+\frac{iz^2}{2}} & 1 & 0 \\
		0 & 0 & 1
	\end{array}\right),\\
	v^{(3,\epsilon)}_2&=\left(\begin{array}{ccc}
		1 & e^{2i \nu_1 \log_0\left(z\right)}(\delta_A^1)^{-1} r_{1,a} e^{-t \Phi_{21}^0(k_0,z)-\frac{iz^2}{2}} & 0 \\
		0 & 1 & 0 \\
		0 & 0 & 1
	\end{array}\right),\\
	v^{(3,\epsilon)}_3&=\left(\begin{array}{ccc}
		1 & 0 & 0 \\
		-e^{-2i \nu_1 \log_0\left(z\right)}\delta_A^1r^*_{1,a} e^{t \Phi_{21}^0(k_0,z)+\frac{iz^2}{2}} & 1 & 0 \\
		0 & 0 & 1
	\end{array}\right),\\
	v^{(3,\epsilon)}_4&=\left(\begin{array}{ccc}
		1 & -e^{2i \nu_1 \log_0\left(z\right)}(\delta_A^1)^{-1}\rho_{1,a} e^{-t \Phi_{21}^0(k_0,z)-\frac{iz^2}{2}} & 0 \\
		0 & 1 & 0 \\
		0 & 0 & 1
	\end{array}\right).\\
\end{aligned}
$$
Moreover,

$$
\begin{aligned}
	v^{(3,\epsilon)}_7&=\left(\begin{array}{ccc}
		1 & 0 & 0 \\
		e^{2i \nu_4 \log_{\pi}\left(z\right)}\left(\delta_B^1\right)^{-1}\rho_{2,a} e^{t\Phi_{21}^0(-k_0,z)-\frac{iz^2}{2}} & 1 & 0 \\
		0 & 0 & 1
	\end{array}\right),\\
	v^{(3,\epsilon)}_8&=\left(\begin{array}{ccc}
		1 & -e^{-2i \nu_4 \log_{\pi}\left(z\right)}\delta_B^1\rho_{2,a}^{*} e^{-t\Phi_{21}^0(-k_0,z)+\frac{iz^2}{2}} & 0 \\
		0 & 1 & 0 \\
		0 & 0 & 1
	\end{array}\right),\\
	v^{(3,\epsilon)}_9&=\left(\begin{array}{ccc}
		1 & 0 & 0 \\
		-e^{2i \nu_4 \log_{\pi}\left(z\right)}\left(\delta_B^1\right)^{-1} r_{2,a} e^{t\Phi_{21}^0(-k_0,z)-\frac{iz^2}{2}} & 1 & 0 \\
		0 & 0 & 1
	\end{array}\right),\\
	v^{(3,\epsilon)}_{10}&=\left(\begin{array}{ccc}
		1 & e^{-2i \nu_4 \log_{\pi}\left(z\right)}\delta_B^1r_{2,a}^* e^{-t\Phi_{21}^0(-k_0,z)+\frac{iz^2}{2}} & 0 \\
		0 & 1 & 0 \\
		0 & 0 & 1
	\end{array}\right).\\
\end{aligned}
$$

When $z$ is fixed, we observe that $r_{j,a} \to r_j(k_0)$,$\rho_{j,a} \to \frac{r_j(k_0)}{1-|r_j(k_0)|^2}$ ,$\delta_A^1,\delta_B^1\to 1$ and $e^{\pm t\Phi_{21}^0(\pm k_0,z)}\to1$ as $t\to\infty$ so that the $v^{3,\epsilon} \to v^X_{A,B}$ as $t\to\infty$ where $v^{X}_{A,B}$ are jump matrix of model problem of $M^{X}_{A,B}$.

\par

\subsection{The model problem $M_{A,B}^X$}

Take $X_1=\{z\in\C:z=re^{\frac{\pi i}{4}},0\le r\le\infty\}$,$X_2=\{z\in\C:z=re^{\frac{3\pi i}{4}},0\le r\le\infty\}$ and $X_3=\{z\in\C:z=re^{\frac{5\pi i}{4}},0\le r\le\infty\}$, $X_4=\{z\in\C:z=re^{\frac{7\pi i}{4}},0\le r\le\infty\}$. Denote $X=\cup_{j=1}^4X_j$ and the function $\nu_A(y)=-\frac{1}{2 \pi} \ln \left(1-\left|y\right|^2\right)$ from $B(0,1)$ to $(0,\infty)$.  In what follows, define the model problem $M^X_{A,B}$ naturally.

\begin{figure}[!h]
	\centering
	\includegraphics[scale=0.5]{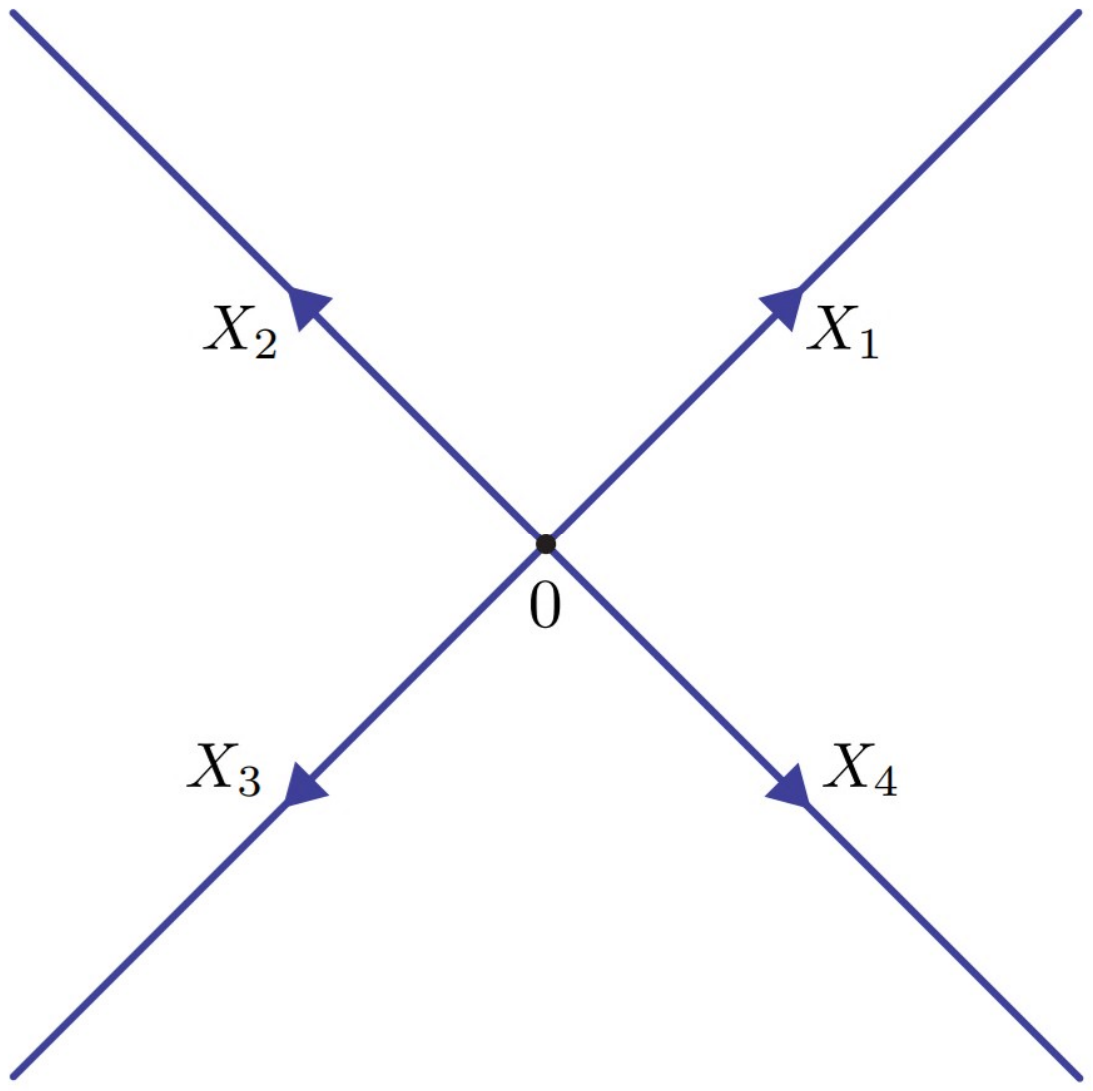}
	\caption{{\protect\small
			The contour.}}
	\label{model-X.pdf}
\end{figure}

\begin{proposition}
	The $3 \times 3$ matrix-valued function $M^X_A$ satisfies the following properties:	
	
 (1). $M^X_A(\cdot\ ,y):\C\setminus X\to\C^{3\times3}$  is analytic for $z \in\C\setminus X$.

 (2). $M^X_A(z,y)$ continuous to $X\setminus\{0\}$ and satisfy the jump condition below:
$$
(M^X_A(z,y))_+=(M^X_A(z,y))_-v^{X_A}_j(z,y),\quad z\in\C\setminus \{0\}.
$$
where the jump matrix $v^X_A(z,y)$ is defined as following:
$$
\begin{aligned}
	& \left(\begin{array}{ccc}
		1 & 0 & 0 \\
		{\frac{\bar{y}}{1-|y|^2}} z^{-2 i \nu_1(y)} e^{\frac{i z^2}{2}} & 1 & 0 \\
		0 & 0 & 1
	\end{array}\right) \quad \text { if } z \in X_1, \quad\left(\begin{array}{ccc}
		1 & y z^{2 i \nu_1(y)} e^{-\frac{i z^2}{2}} & 0 \\
		0 & 1 & 0 \\
		0 & 0 & 1
	\end{array}\right) \text { if } z \in X_2, \\
	& \left(\begin{array}{ccc}
		1 & 0 & 0 \\
		-\bar y z^{-2i \nu_1 } e^{\frac{iz^2}{2}} & 1 & 0 \\
		0 & 0 & 1
	\end{array}\right) \text { if } z \in X_3, \quad\left(\begin{array}{ccc}
		1 & -\frac{y}{1-|y|^2} z^{2i \nu_1 }  e^{-\frac{iz^2}{2}} & 0 \\
		0 & 1 & 0 \\
		0 & 0 & 1
	\end{array}\right) \text { if } z \in X_4,
\end{aligned}
$$
and $z^{2i\nu_1(y)}=e^{2i\nu_1(y){log_0(z)}}$ for choosing the branch cut running along $\R_+$.

(3). $M^X_A(z ,y)\to I$  as $z\to\infty$.

(4). $M^X_A(z ,y)\to O(1)$ as $z\to 0$.

For $|y|<1$, the RH problem $M^X_A$ satisfies the following expansion:
$$
M^X_A(y,z)=I+\frac{\left(M^X_A(y)\right)_1}{z}+O\left(\frac{1}{z^2}\right),
$$
where
$$
\left(M_A^X(y)\right)_1=\left(\begin{array}{ccc}
	0 & \beta_{12}^A & 0 \\
	\beta_{21}^A & 0 & 0 \\
	0 & 0 & 0
\end{array}\right), \quad y \in B(0,1),
$$
and
$$
\beta_{12}^A=\frac{\sqrt{2 \pi} e^{-\frac{\pi i}{4}} e^{-\frac{5 \pi \nu_1}{2}}}{\bar{y} \Gamma(-i \nu_1)}, \quad \beta_{21}^A=\frac{\sqrt{2 \pi} e^{\frac{\pi i}{4}} e^{\frac{3 \pi \nu_1}{2}}}{y \Gamma(i \nu_1)}.
$$
\end{proposition}

\begin{proposition}
	On the other hand, the $3 \times 3$ matrix valued function $M^X_B$ satisfies the following properties:

(1). $M^X_B(\cdot,y):\C\setminus X\to\C^{3\times3}$  is analytic for $z \in\C\setminus X$.

(2). $M^X_B(z,y)$ continuous to $X\setminus\{0\}$ and satisfy the jump condition below:
$$
(M^X_B(z,y))_+=(M^X_B(z,y))_-v^{X_B}_j(z,y),\quad z\in\C\setminus \{0\}.
$$
where the jump matrix $v^X_B(z,y)$ is defined as following:
$$
\begin{aligned}
	& \left(\begin{array}{ccc}
		1 & \bar{y} z^{-2 i \nu_4(y)} e^{\frac{i z^2}{2}} & 0 \\
		0 & 1 & 0 \\
		0 & 0 & 1
	\end{array}\right) \quad \text { if } z \in X_1, \quad\left(\begin{array}{ccc}
		1 & 0 & 0 \\
		\frac{y}{1-|y|^2} z^{2 i \nu_4(y)} e^{-\frac{i z^2}{2}} & 1 & 0 \\
		0 & 0 & 1
	\end{array}\right) \text { if } z \in X_2, \\
	& \left(\begin{array}{ccc}
		1 & -\frac{\bar y}{1-|y|^2} z^{-2i \nu_4 } e^{\frac{iz^2}{2}} & 0 \\
		0 & 1 & 0 \\
		0 & 0 & 1
	\end{array}\right) \text { if } z \in X_3, \quad\left(\begin{array}{ccc}
		1 & 0 & 0 \\
		-y z^{2i \nu_4 }  e^{-\frac{iz^2}{2}} & 1 & 0 \\
		0 & 0 & 1
	\end{array}\right) \text { if } z \in X_4,
\end{aligned}
$$
and $z^{2i\nu_4(y)}=e^{2i\nu_4(y)}e^{\log_{\pi}(z)}$ for choosing the branch cut running along $\R_-$.

(3). $M^X_B(z ,y)\to I$  as $z\to\infty$.

(4). $M^X_B(z ,y)\to O(1)$ as $z\to 0$.

For $|y|<1$, the RH problem $M^X_B$ satisfies the following expansion:
$$
M^X_B(y,z)=I+\frac{\left(M^X_B(y)\right)_1}{z}+O\left(\frac{1}{z^2}\right),
$$
where
$$
\left(M_B^X(y)\right)_1=\left(\begin{array}{ccc}
	0 & \beta_{12}^B & 0 \\
	\beta_{21}^B & 0 & 0 \\
	0 & 0 & 0
\end{array}\right), \quad y \in B(0,1),
$$
and
$$
\beta_{12}^B=\frac{\sqrt{2\pi}e^{\frac{\pi i}{4}}e^{-\frac{\pi\nu_4}{2}}}{ y\Gamma(i\nu_4)},\quad \beta_{21}^B=\frac{\sqrt{2\pi}e^{-\frac{\pi i}{4}}e^{-\frac{\pi\nu_4}{2}}}{\bar y\Gamma(-i\nu_4)}.
$$
\end{proposition}

\begin{remark}
	 The model problem $X_A$ and $X_B$ has a relationship like mirror reflection, since the  orientation of original RH problem 1 are opposite and the $\delta$ function.
\end{remark}

Furthermore, factorize Model problem to the Parabolic problem so that we can get the leading term in long time.

At first, add jump contour $\R$ with jump function $I$ and inverse the orientation of $X_2$ and $X_3$, as shown in Fig. \ref{model-X-R.pdf}.
\begin{figure}[!h]
	\centering
	\includegraphics[scale=0.4]{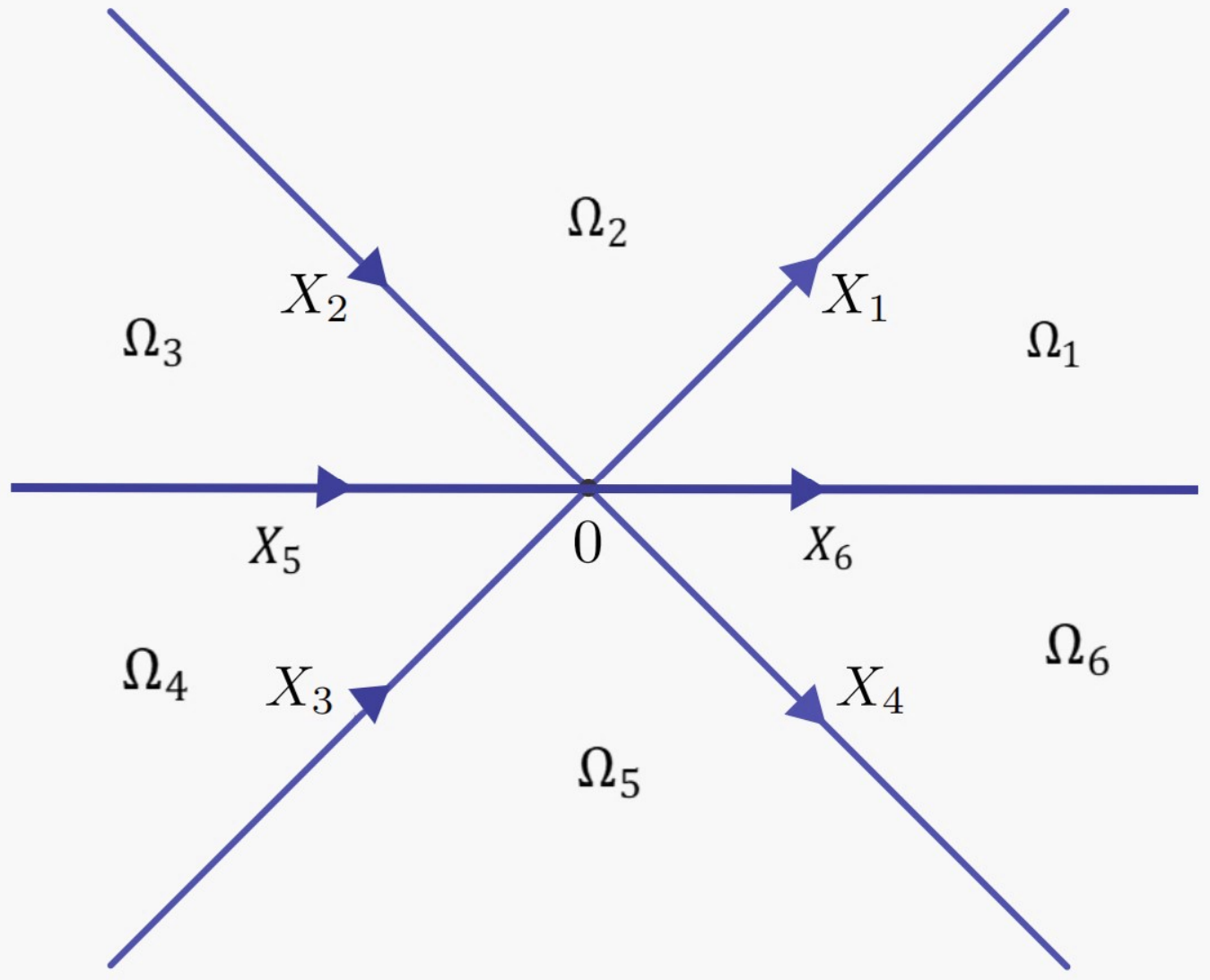}
    \caption{{\protect\small
			The new contour when adding jump contour $\R$ with jump function $I$.}}
	\label{model-X-R.pdf}
\end{figure}

Denote
$$
z^{i\nu_1\tilde\sigma_3}=
\left(\begin{array}{ccc}
	z^{i\nu_1\sigma_3} & 0 \\
	0 & 1\\
\end{array}\right)=
\left(\begin{array}{ccc}
	z^{i\nu_1} & 0 &0\\
	0 & z^{-i\nu_1} &0\\
	0 & 0 & 1
\end{array}\right),\quad
e^{\frac{iz^2}{4}\tilde\sigma_3}=
\left(\begin{array}{ccc}
	e^{\frac{iz^2}{4}\sigma_3} & 0 \\
	0 & 1\\
\end{array}\right)=
\left(\begin{array}{ccc}
	e^{\frac{iz^2}{4}} & 0 &0\\
	0 & e^{-\frac{iz^2}{4}}  &0\\
	0 & 0 & 1
\end{array}\right),
$$
and introduce the transform matrix $\mathcal{H}^A$ as follows
$$
\mathcal{H}^A(y, z)= \begin{cases}
	\left(\begin{array}{ccc}
		1 & 0 & 0 \\
		{\frac{\bar{y}}{1-|y|^2}} z^{-2 i \nu_1(y)} e^{\frac{i z^2}{2}} & 1 & 0 \\
		0 & 0 & 1
	\end{array}\right)z^{i\nu_1\tilde\sigma_3} , & z \in \Omega_1, \\
	
	z^{i\nu_1\tilde\sigma_3}, & z \in \Omega_2, \\
	
	\left(\begin{array}{ccc}
		1 & -y z^{2 i \nu_1(y)} e^{-\frac{i z^2}{2}} & 0 \\
		0 & 1 & 0 \\
		0 & 0 & 1
	\end{array}\right)z^{i\nu_1\tilde\sigma_3}, & z \in \Omega_3, \\
	
	\left(\begin{array}{ccc}
		1 & 0 & 0 \\
		-\bar y z^{-2i \nu_1 } e^{\frac{iz^2}{2}} & 1 & 0 \\
		0 & 0 & 1
	\end{array}\right)z^{i\nu_1\tilde\sigma_3}, & z \in \Omega_4, \\
	z^{i\nu_1\tilde\sigma_3}, & z\in \Omega_5,\\
	
	\left(\begin{array}{ccc}
		1 & \frac{y}{1-|y|^2} z^{2i \nu_1 }  e^{-\frac{iz^2}{2}} & 0 \\
		0 & 1 & 0 \\
		0 & 0 & 1
	\end{array}\right)z^{i\nu_1\tilde\sigma_3}, & z\in \Omega_6.\\
	
\end{cases}
$$
Further, take
$$
P^A(y,z)=M^X_A(y,z)\mathcal{H}^A(y,z).,\\
$$
By the directly computation, one has
$$
P^A_+(y,z)=P^A_-(y,z)v^{P^A}_j(y,z),\quad
v^{P^A}_j(y,z)=(\mathcal{H}^A_-(y,z))^{-1}v^{X_A}_j(y,z)\mathcal{H}^A_+(y,z).
$$
More specially, we have
$$
v^{P^A}_j(y,z)=\begin{cases}
	I,\quad z\in X_1,\\
	I,\quad z\in X_2,\\
	I,\quad z\in X_3,\\
	I,\quad z\in X_4,\\
	e^{-\frac{iz^2}{4}\operatorname{ad} \tilde\sigma_3}\left(\begin{array}{ccc}
		1& -y & 0\\
		\bar y & 1-|y|^2 &0\\
		0 & 0 & 1
	\end{array}
	\right),\quad z\in X_5,\\
	e^{-\frac{iz^2}{4}\operatorname{ad} \tilde\sigma_3}\left(\begin{array}{ccc}
		1& -y & 0\\
		\bar y & 1-|y|^2 &0\\
		0 & 0 & 1
	\end{array}
	\right),\quad z\in X_6.
\end{cases}
$$
There exists jump  for $z^{2iv_1(y)}=e^{2i\nu_1(y)\log_0(z)}$ in $z\in\R_+$, i.e.,
when $z\in X_6$
$$
\begin{aligned}
	v^{P^A}_6&=(\mathcal{H}^A_-(y,z))^{-1}v^{X_A}_6(y,z)\mathcal{H}^A_+(y,z)
	=(\mathcal{H}^A_{6,-}(y,z))^{-1}\mathcal{H}^A_{1,+}(y,z)\\
	&=z_-^{-i\nu_1\tilde\sigma_3}
	\left(\begin{array}{ccc}
		1 & -\frac{y}{1-|y|^2} z_-^{2i \nu_1 }  e^{-\frac{iz^2}{2}} & 0 \\
		0 & 1 & 0 \\
		0 & 0 & 1
	\end{array}\right)\left(\begin{array}{ccc}
		1 & 0 & 0 \\
		{\frac{\bar{y}}{1-|y|^2}} z_+^{-2 i \nu_1(y)} e^{\frac{i z^2}{2}} & 1 & 0 \\
		0 & 0 & 1
	\end{array}\right)z_+^{i\nu_1\tilde\sigma_3}\\
	&=z_-^{-i\nu_1\tilde\sigma_3}
	\left(\begin{array}{ccc}
		1-\frac{|y|^2}{(1-|y|^2)^2}z_-^{2i\nu_1}z_+^{-2i\nu_1} & -\frac{y}{1-|y|^2} z_-^{2i \nu_1 }  e^{-\frac{iz^2}{2}} & 0 \\
		\frac{\bar y}{1-|y|^2} z_+^{-2i \nu_1 }  e^{\frac{iz^2}{2}} & 1 & 0 \\
		0 & 0 & 1
	\end{array}\right)z_+^{i\nu_1\tilde\sigma_3}\\
	&=e^{-\frac{iz^2}{4}\operatorname{ad} \tilde\sigma_3}\left(\begin{array}{ccc}
		(1-|y|^2)z_-^{-i\nu_1}z_+^{i \nu_1 } & -\frac{y}{1-|y|^2} z_-^{i \nu_1 }z_+^{-i \nu_1 }   & 0 \\
		\frac{\bar y}{1-|y|^2} z_-^{i\nu_1}z_+^{-i \nu_1 }  &  z_-^{i\nu_1}z_+^{-i \nu_1 } & 0 \\
		0 & 0 & 1
	\end{array}\right)\\
	&=e^{-\frac{iz^2}{4}\operatorname{ad} \tilde\sigma_3}\left(\begin{array}{ccc}
		1 & -y   & 0 \\
		\bar y  &  1-|y|^2 & 0 \\
		0 & 0 & 1
	\end{array}\right).\\
\end{aligned}
$$
Since
$$
z_-^{2i\nu_1}z_+^{-2i\nu_1}=e^{2i\nu_1(\log_0(z_-)-\log_0(z_+))}=e^{2i\nu_1(\ln|z|+2\pi i-\ln|z|)}
=e^{-4\pi\nu_1}=(1-|y|^2)^2,
$$
we get a RH problem with jump contour on $\R$, in particular
$$
\begin{cases}
	P^A_+(y,z)=P^A_-(y,z)e^{-\frac{iz^2}{4}\operatorname{ad} \tilde\sigma_3}V^A(y,z),&z\in\R,\\
	P^A\to z^{i\nu_1\tilde\sigma_3} &\text{as}\quad z \to \infty,
\end{cases}
$$
where
$$
V^A=\left(\begin{array}{ccc}
	1 & -y   & 0 \\
	\bar y  &  1-|y|^2 & 0 \\
	0 & 0 & 1
\end{array}\right).
$$

Since $\mathcal H^A=\left(I+O\left(\frac{1}{z}\right)\right)z^{i\nu_1\tilde \sigma_3}$ and as the procedure in \cite{Deift-Zhou-1993}, one has
$$
\Psi=P^Ae^{-\frac{iz^2\tilde\sigma_3}{4}}=\hat{\Psi}z^{i\nu\tilde\sigma_3}e^{-\frac{iz^2\tilde\sigma_3}{4}},
$$
and
$$
P^A=\left(I+\frac{\left(M^X_A(y)\right)_1}{z}+O\left(\frac{1}{z^2}\right)\right)z^{i\nu_1\tilde\sigma_3},
$$
then we have
$$
\Psi_+(z)=\Psi_-(z)V^A(y),
$$
and further by direct computation
$$
\left(\partial_z\Psi+\frac{iz\tilde \sigma_3}{2}\Psi\right)_+=\left(\partial_z\Psi+\frac{iz\tilde \sigma_3}{2}\Psi\right)_-V^A(y).
$$
It is claimed that $\left(\partial_z\Psi+\frac{iz\tilde \sigma_3}{2}\Psi\right)\Psi^{-1}$ has no jump along the $\R$, so it is an entire function on the $\C$, then we have
$$
\left(\partial_z\Psi+\frac{iz\tilde \sigma_3}{2}\Psi\right)\Psi^{-1}=\frac{i}{2}\left[\tilde\sigma_3,\left(M^X_A\right)_1\right]+O\left(\frac{1}{z}\right).
$$
Let
$$
\frac{i}{2}\left[\tilde\sigma_3,\left(M^X_A\right)_1\right]:=\left(\begin{array}{ccc}
	0 & \tilde\beta_{12} & 0 \\
	\tilde\beta_{21} & 0 & 0 \\
	0 & 0 & 0
\end{array}\right),
$$
where $\beta_{12}=-i\tilde\beta_{12}$ and $\beta_{21}=i\tilde\beta_{21}$.

Rewrite the above ordinary differential equation as
$$
\frac{\partial \Psi_{11}}{\partial z}+\frac{1}{2} i z \Psi_{11}=\tilde\beta_{12} \Psi_{21},
$$
then we have
$$
\frac{\partial^2 \Psi_{11}}{\partial z^2}+\left(\frac{1}{2} i+\frac{1}{4} z^2-\tilde\beta_{12} \tilde\beta_{21}\right) \Psi_{11}=0,
$$
and
$$
\frac{\partial^2 \Psi_{22}}{\partial z^2}+\left(-\frac{1}{2} i+\frac{1}{4} z^2-\tilde\beta_{12} \tilde\beta_{21}\right) \Psi_{22}=0.
$$
Now focusing on the upper half plane and denoting $z=e^{\frac{3}{4} \pi i} \xi, \Psi_{11}^{+}(z)=\Psi_{11}^{+}\left(e^{\frac{3}{4} \pi i} \xi\right)=g(\xi)$, yields
$$
\frac{d^2 g(\xi)}{d \xi^2}+\left(\frac{1}{2}-\frac{1}{4} \xi^2+a\right) g(\xi)=0,
$$
which satisfies the Weber equation and by the standard analysis
$$
g(\xi)=c_1 D_a(\xi)+c_2 D_a(-\xi),
$$
where
$$
D_a(\xi)=\left\{\begin{array}{l}
	\xi^a e^{-\frac{1}{4} \xi^2}\left(1+O\left(\frac{1}{\xi^2}\right)\right), \quad|\arg \xi|<\frac{3}{4} \pi, \\
	\xi^a e^{-\frac{1}{4} \xi^2}\left(1+O\left(\frac{1}{\xi^2}\right)\right)-\frac{\sqrt{2 \pi}}{\Gamma(-a)} e^{a \pi i} \xi^{-a-1} e^{\frac{1}{4} \xi^2}\left(1+O\left(\frac{1}{\xi^2}\right)\right), \quad \frac{\pi}{4}<\arg \xi<\frac{5 \pi}{4}, \\
	\xi^a e^{-\frac{1}{4} \xi^2}\left(1+O\left(\frac{1}{\xi^2}\right)\right)-\frac{\sqrt{2 \pi}}{\Gamma(-a)} e^{-a \pi i} \xi^{-a-1} e^{\frac{1}{4} \xi^2}\left(1+O\left(\frac{1}{\xi^2}\right)\right), \quad-\frac{5 \pi}{4}<\arg \xi<-\frac{\pi}{4},
\end{array}\right.
$$
where $D_a(z)$ denotes the parabolic cylinder function.

Combining the behavior of $P^A$ as $z\to\infty$ above leads to
$$
a=i\nu_1,\quad c_1=(e^{\frac{3\pi i}{4}})^{i \nu_1}=e^{-\frac{3\pi \nu}{4}},\quad c_2=0,
$$
and in the other word
$$
\Psi_{11}^{+}(z)=e^{-\frac{3}{4} \pi \nu} D_a\left(e^{-\frac{3}{4} \pi i} z\right), \quad a=i \nu.
$$
For the $\operatorname{Im} z<0$, as the same procedure, let $z=e^{-\frac{1}{4} \pi i} \xi, \quad \Psi_{11}^{-}(z)=\Psi_{11}^{-}\left(e^{-\frac{1}{4} \pi i} \xi\right)=g(\xi)$ and one obtains that
$$
a=i\nu_1,\quad c_1=(e^{\frac{7\pi i}{4}})^{i \nu_1}=e^{-\frac{7\pi \nu_1}{4}},\quad c_2=0.
$$
On the other hand, when $\operatorname{Im} z>0$, let $z=e^{-\frac{1}{4} \pi i} \xi, \quad \Psi_{22}^{-}(z)=\Psi_{22}^{-}\left(e^{-\frac{1}{4} \pi i} \xi\right)=g(\xi)$ and we have
$$
a=-i\nu_1,\quad c_1=(e^{\frac{\pi i}{4}})^{-i \nu_1}=e^{\frac{\pi \nu_1}{4}},\quad c_2=0.
$$
When $\operatorname{Im} z<0$, let $z=e^{\frac{3}{4} \pi i} \xi, \quad \Psi_{22}^{-}(z)=\Psi_{22}^{-}\left(e^{\frac{3}{4} \pi i} \xi\right)=g(\xi)$ and we have
$$
a=-i\nu_1,\quad c_1=(e^{\frac{5\pi i}{4}})^{-i \nu_1}=e^{\frac{5\pi \nu_1}{4}},\quad c_2=0.
$$
Notice that $\operatorname{Im} z<0$, choose $e^{\frac{7}{4} \pi i}$ and $e^{\frac{5\pi i}{4}}$ since the branch cut from $0$ to $2\pi$.

In conclusion, we have
$$
\Psi_{11}(q, z)= \begin{cases}e^{\frac{-3 \pi \nu_1}{4}} D_{i \nu_1}\left(e^{-\frac{3 \pi i}{4}} z\right), & \operatorname{Im} z>0, \\ e^{\frac{-7 \pi \nu_1}{4}} D_{i \nu_1}\left(e^{\frac{\pi i}{4}} z\right), & \operatorname{Im} z<0,\end{cases}
$$
and
$$
\Psi_{22}(q, z)= \begin{cases}e^{\frac{\pi \nu_1}{4}} D_{-i \nu_1}\left(e^{-\frac{\pi i}{4}} z\right), & \operatorname{Im} z>0, \\ e^{\frac{5 \nu_1}{4}} D_{-i \nu_1}\left(e^{\frac{3 \pi i}{4}} z\right), & \operatorname{Im} z<0.\end{cases}
$$
Moreover, we have
$$
\begin{aligned}
	& \Psi_{21}^{+}(z)=e^{-\frac{3}{4} \pi \nu_1}\left(\tilde\beta_{12}\right)^{-1}\left[\partial_z D_a\left(e^{-\frac{3 \pi i}{4}} z\right)+\frac{i z}{2} D_a\left(e^{-\frac{3 \pi i}{4}} z\right)\right], \\
	& \Psi_{12}^{+}(z)=e^{\frac{1}{4} \pi \nu_1}\left(\tilde\beta_{21}\right)^{-1}\left[\partial_z D_{-a}\left(e^{-\frac{\pi i}{4}} k\right)-\frac{i z}{2} D_{-a}\left(e^{-\frac{\pi i}{4}} z\right)\right],
\end{aligned}
$$
and
$$
\begin{aligned}
	& \Psi_{21}^{-}(z)=e^{-\frac{7}{4} \pi \nu_1}\left(\tilde\beta_{12}\right)^{-1}\left[\partial_z D_a\left(e^{\frac{\pi i}{4}} z\right)+\frac{i k}{2} D_a\left(e^{\frac{\pi i}{4}} z\right)\right], \\
	& \Psi_{12}^{-}(z)=e^{\frac{5}{4} \pi \nu_1}\left(\tilde\beta_{21}\right)^{-1}\left[\partial_z D_{-a}\left(e^{\frac{3 \pi i}{4}} z\right)-\frac{i k}{2} D_{-a}\left(e^{\frac{3 \pi i}{4}} z\right)\right].
\end{aligned}
$$
Since $\left(\Psi_-\right)^{-1}\Psi_+=V^A(y)$ and we can obtain that
$$
\bar y=\Psi^-_{11}\Psi^+_{21}-\Psi^-_{21}\Psi^+_{11}=\frac{e^{-\frac{5\pi\nu_1}{2}}}{\tilde\beta_{12}}W\left(D_{i\nu_1}(e^{\frac{\pi i}{4}}z),D_{i\nu_1}(e^{-\frac{3\pi i}{4}}z)\right)=\frac{\sqrt{2\pi}e^{\frac{\pi i}{4}}e^{-\frac{5\pi\nu_1}{2}}}{\tilde\beta_{12}\Gamma(-i\nu_1)},
$$
and
$$
-y=\Psi^-_{22}\Psi^+_{12}-\Psi^-_{12}\Psi^+_{22}=\frac{e^{\frac{3\pi\nu_1}{2}}}{\tilde\beta_{21}}W\left(D_{-i\nu_1}(e^{\frac{3\pi i}{4}}z),D_{-i\nu_1}(e^{-\frac{\pi i}{4}}z)\right)=\frac{\sqrt{2\pi}e^{\frac{3\pi i}{4}}e^{\frac{3\pi\nu_1}{2}}}{\tilde\beta_{21}\Gamma(i\nu_1)}.
$$
Now, we have
$$
\beta_{12}^A=-i\tilde\beta_{12}=\frac{\sqrt{2\pi}e^{-\frac{\pi i}{4}}e^{-\frac{5\pi\nu_1}{2}}}{\bar y\Gamma(-i\nu_1)},\quad \beta_{21}^A=i\tilde\beta_{21}=\frac{\sqrt{2\pi}e^{\frac{\pi i}{4}}e^{\frac{3\pi\nu_1}{2}}}{ y\Gamma(i\nu_1)}.
$$

Again, introduce the matrix
$$
\mathcal{H}^B(y, z)= \begin{cases}
	\left(\begin{array}{ccc}
		1 & \bar{y} z^{-2 i \nu_4(y)} e^{\frac{i z^2}{2}} & 0 \\
		0 & 1 & 0 \\
		0 & 0 & 1
	\end{array}\right)z^{-i\nu_4\tilde\sigma_3} , & z \in \Omega_1, \\
	
	z^{-i\nu_4\tilde\sigma_3}, & z \in \Omega_2, \\
	
	\left(\begin{array}{ccc}
		1 & 0 & 0 \\
		-\frac{y}{1-|y|^2} z^{2 i \nu_4(y)} e^{-\frac{i z^2}{2}} & 1 & 0 \\
		0 & 0 & 1
	\end{array}\right)z^{-i\nu_4\tilde\sigma_3}, & z \in \Omega_3, \\
	
	\left(\begin{array}{ccc}
		1 & -\frac{\bar y}{1-|y|^2} z^{-2i \nu_4 } e^{\frac{iz^2}{2}} & 0 \\
		0 & 1 & 0 \\
		0 & 0 & 1
	\end{array}\right)z^{-i\nu_4\tilde\sigma_3}, & z \in \Omega_4, \\
	z^{-i\nu_4\tilde\sigma_3}, & z\in \Omega_5,\\
	
	\left(\begin{array}{ccc}
		1 & 0 & 0 \\
		y z^{2i \nu_4 }  e^{-\frac{iz^2}{2}} & 1 & 0 \\
		0 & 0 & 1
	\end{array}\right)z^{-i\nu_4\tilde\sigma_3}, & z\in \Omega_6.\\
	
\end{cases}
$$
and take the transformation
$$
P^B(y,z)=M^X_B(y,z)\mathcal{H}^B(y,z).\\
$$
The direct calculation yields
$$
P^B_+(y,z)=P^B_-(y,z)v^{P^B}_j(y,z),\quad
v^{P^B}_j(y,z)=(\mathcal{H}^B_-(y,z))^{-1}v^{X_B}_j(y,z)\mathcal{H}^B_+(y,z).
$$
In particular, we have
$$
v^{P^B}_j(y,z)=\begin{cases}
	I,\quad z\in X_1,\\
	I,\quad z\in X_2,\\
	I,\quad z\in X_3,\\
	I,\quad z\in X_4,\\
	e^{\frac{iz^2}{4}\operatorname{ad} \tilde\sigma_3}\left(\begin{array}{ccc}
		1& \bar y & 0\\
		- y & 1-|y|^2 &0\\
		0 & 0 & 1
	\end{array}
	\right),\quad z\in X_5,\\
	e^{\frac{iz^2}{4}\operatorname{ad} \tilde\sigma_3}\left(\begin{array}{ccc}
		1& \bar y & 0\\
		- y & 1-|y|^2 &0\\
		0 & 0 & 1
	\end{array}
	\right),\quad z\in X_6.
\end{cases}
$$
Like the situation in $P^A$, there exits branch cut on $\R_-$ for $z^{2i\nu_4}=e^{2i\nu_4\log_{\pi}z}$, so that the jump matrix on $X_5$ is just as follows

\begin{align*}
	V_5^B(y,z)&=z_-^{i\nu_4\tilde\sigma_3}	
	\left(\begin{array}{ccc}
		1 & \frac{\bar y}{1-|y|^2} z_-^{-2i \nu_4 } e^{\frac{iz^2}{2}} & 0 \\
		0 & 1 & 0 \\
		0 & 0 & 1
	\end{array}\right)\left(\begin{array}{ccc}
		1 & 0 & 0 \\
		-\frac{y}{1-|y|^2} z_+^{2 i \nu_4} e^{-\frac{i z^2}{2}} & 1 & 0 \\
		0 & 0 & 1
	\end{array}\right)	
	z_+^{-i\nu_4\tilde\sigma_3},\\	
	&=z_-^{i\nu_4\tilde\sigma_3}	
	\left(\begin{array}{ccc}
		1-\frac{|y|^2}{(1-|y|^2)^2}z_-^{-2i \nu_4 }z_+^{2 i \nu_4} & \frac{\bar y}{1-|y|^2} z_-^{-2i \nu_4 } e^{\frac{iz^2}{2}} & 0 \\
		-\frac{y}{1-|y|^2} z_+^{2 i \nu_4} e^{-\frac{i z^2}{2}} & 1 & 0 \\
		0 & 0 & 1
	\end{array}\right)	
	z_+^{-i\nu_4\tilde\sigma_3},\\	
	&=e^{\frac{iz^2}{4}\operatorname{ad} \tilde\sigma_3}
	\left(\begin{array}{ccc}
		(1-|y|^2)z_-^{i \nu_4 }z_+^{-i \nu_4 } & \frac{\bar y}{1-|y|^2} z_-^{-i \nu_4 }z_+^{i \nu_4 }  & 0 \\
		-\frac{y}{1-|y|^2} z_-^{-i \nu_4 }z_+^{i \nu_4 }  & z_-^{-i \nu_4 }z_+^{i \nu_4 } & 0 \\
		0 & 0 & 1
	\end{array}\right),\\
	&=e^{\frac{iz^2}{4}\operatorname{ad} \tilde\sigma_3}\left(\begin{array}{ccc}
		1 & {\bar y}  & 0 \\
		-{y}  & 1-|y|^2 & 0 \\
		0 & 0 & 1
	\end{array}\right),
\end{align*}
where
$$
z_-^{-2i \nu_4 }z_+^{2 i \nu_4}=e^{2i\nu_4(-\log_{\pi}(z_-)+\log_{\pi}(z_+))}=e^{-4\pi\nu_4}=(1+|r_2(-k_0)|^2)^2.
$$
\par
Thus the RH problem is
$$
\begin{cases}
	P^B_+(y,z)=P^B_-(y,z)e^{\frac{iz^2}{4}\operatorname{ad} \tilde\sigma_3}V^B(y,z),&z\in\R,\\
	P^B\to z^{-i\nu_4\tilde\sigma_3} &\text{as}\quad z \to \infty,
\end{cases}
$$
where
$$
V^B=\left(\begin{array}{ccc}
	1 & \bar y   & 0 \\
	- y  &  1-|y|^2 & 0 \\
	0 & 0 & 1
\end{array}\right).
$$

Since $\mathcal H^B=\left(I+O\left(\frac{1}{z}\right)\right)z^{-i\nu_1\tilde \sigma_3}$ and by the procedure in \cite{Deift-Zhou-1993}, we have
$$
\Psi=P^Be^{\frac{iz^2\tilde\sigma_3}{4}}=\hat{\Psi}z^{-i\nu_4\tilde\sigma_3}e^{\frac{iz^2\tilde\sigma_3}{4}},
$$
and
$$
P^B=\left(I+\frac{\left(M^X_B(y)\right)_1}{z}+O\left(\frac{1}{z^2}\right)\right)z^{-i\nu_1\tilde\sigma_3},
$$
then we have
$$
\Psi_+(z)=\Psi_-(z)V^B(y),
$$
and by direct computation
$$
\left(\partial_z\Psi-\frac{iz\tilde \sigma_3}{2}\Psi\right)_+=\left(\partial_z\Psi-\frac{iz\tilde \sigma_3}{2}\Psi\right)_-V^B(y).
$$
We claim that $\left(\partial_z\Psi-\frac{iz\tilde \sigma_3}{2}\Psi\right)\Psi^{-1}$ has no jump along the $\R$, so it is an entire function on the $\C$, then we have
$$
\begin{aligned}
	\left(\partial_z\Psi-\frac{iz\tilde \sigma_3}{2}\Psi\right)\Psi^{-1}&=(\partial_z\hat\Psi)\hat\Psi^{-1}+\hat\Psi(-i\nu_4\tilde\sigma_3z^{-1})\hat\Psi^{-1}+\hat\Psi\left(\frac{iz}{2}\tilde\sigma_3\right)\hat\Psi^{-1}-\left(\frac{iz}{2}\tilde\sigma_3\hat\Psi\right)\hat\Psi^{-1}\\
	&=-\frac{iz}{2}\left[\tilde\sigma_3,\hat\Psi\right]\hat\Psi^{-1}+O\left(\frac{1}{z}\right)\\
	&=-\frac{i}{2}\left[\tilde\sigma_3,\left(M^X_B\right)_1\right]+O\left(\frac{1}{z}\right).
\end{aligned}
$$
Let
$$
-\frac{i}{2}\left[\tilde\sigma_3,\left(M^X_B\right)_1\right]:=\left(\begin{array}{ccc}
	0 & \tilde\beta_{12} & 0 \\
	\tilde\beta_{21} & 0 & 0 \\
	0 & 0 & 0
\end{array}\right),
$$
where $\beta_{12}^B=i\tilde\beta_{12}$ and $\beta_{21}^B=-i\tilde\beta_{21}$.

Rewrite the above ODE as
$$
\frac{\partial \Psi_{11}}{\partial z}-\frac{1}{2} i z \Psi_{11}=\tilde\beta_{12} \Psi_{21}.
$$
Then we have
$$
\frac{\partial^2 \Psi_{11}}{\partial z^2}+\left(-\frac{1}{2} i+\frac{1}{4} z^2-\tilde\beta_{12} \tilde\beta_{21}\right) \Psi_{11}=0,
$$
and
$$
\frac{\partial^2 \Psi_{22}}{\partial z^2}+\left(\frac{1}{2} i+\frac{1}{4} z^2-\tilde\beta_{12} \tilde\beta_{21}\right) \Psi_{22}=0.
$$
Now focusing on the upper half plane and denoting $z=e^{\frac{1}{4} \pi i} \xi, \Psi_{11}^{+}(z)=\Psi_{11}^{+}\left(e^{\frac{1}{4} \pi i} \xi\right)=g(\xi)$, yield
$$
\frac{d^2 g(\xi)}{d \xi^2}+\left(\frac{1}{2}-\frac{1}{4} \xi^2+a\right) g(\xi)=0,
$$
which satisfies the Weber equation and by the standard analysis
$$
g(\xi)=c_1 D_a(\xi)+c_2 D_a(-\xi).
$$

Combine the behavior of $P^B$ as $z\to\infty$ before, we have
$$
a=-i\nu_4,\quad c_1=(e^{\frac{\pi i}{4}})^{-i \nu_4}=e^{\frac{\pi \nu}{4}},\quad c_2=0,
$$
in other word,
$$
\Psi_{11}^{+}(z)=e^{\frac{1}{4} \pi \nu} D_a\left(e^{\frac{1}{4} \pi i} z\right), \quad a=-i \nu_4.
$$
For the $\operatorname{Im} z<0$, as the same procedure, let $z=e^{-\frac{3}{4} \pi i} \xi, \quad \Psi_{11}^{-}(z)=\Psi_{11}^{-}\left(e^{-\frac{3}{4} \pi i} \xi\right)=g(\xi)$ and get that
$$
a=i\nu_1,\quad c_1=(e^{-\frac{3\pi i}{4}})^{-i \nu_4}=e^{-\frac{3\pi \nu_4}{4}},\quad c_2=0.
$$

On the other hand, when $\operatorname{Im} z>0$, let $z=e^{\frac{3}{4} \pi i} \xi, \Psi_{22}^{+}(z)=\Psi_{22}^{+}\left(e^{\frac{3}{4} \pi i} \xi\right)=g(\xi)$ and we have
$$
a=i\nu_4,\quad c_1=(e^{\frac{3\pi i}{4}})^{i \nu_4}=e^{-\frac{3\pi \nu_1}{4}},\quad c_2=0.
$$
When $\operatorname{Im} z<0$, let $z=e^{\frac{1}{4} \pi i} \xi, \quad \Psi_{22}^{-}(z)=\Psi_{22}^{-}\left(e^{\frac{1}{4} \pi i} \xi\right)=g(\xi)$ and we have
$$
a=i\nu_4,\quad c_1=(e^{-\frac{\pi i}{4}})^{i \nu_4}=e^{\frac{\pi \nu_1}{4}},\quad c_2=0.
$$
Notice that $\operatorname{Im} z<0$, we choose $e^{-\frac{3}{4} \pi i}$ and $e^{-\frac{\pi i}{4}}$ since the branch cut from $-\pi$ to $\pi$.

In conclusion, we have
$$
\Psi_{11}(q, z)= \begin{cases}e^{\frac{\pi \nu_4}{4}} D_{-i \nu_4}\left(e^{-\frac{\pi i}{4}} z\right), & \operatorname{Im} z>0, \\ e^{-\frac{3 \nu_4}{4}} D_{-i \nu_4}\left(e^{\frac{3 \pi i}{4}} z\right), & \operatorname{Im} z<0,\end{cases}
$$
and
$$
\Psi_{22}(q, z)= \begin{cases}e^{\frac{-3 \pi \nu_1}{4}} D_{i \nu_4}\left(e^{-\frac{3 \pi i}{4}} z\right), & \operatorname{Im} z>0, \\ e^{\frac{ \pi \nu_1}{4}} D_{i \nu_4}\left(e^{\frac{\pi i}{4}} z\right), & \operatorname{Im} z<0.\end{cases}
$$
Moreover, we further have
$$
\begin{aligned}
	& \Psi_{12}^{+}(z)=e^{-\frac{3}{4} \pi \nu_4}\left(\tilde\beta_{12}\right)^{-1}\left[\partial_z D_{i\nu_4}\left(e^{-\frac{3 \pi i}{4}} z\right)+\frac{i z}{2} D_{i\nu_4}\left(e^{-\frac{3 \pi i}{4}} z\right)\right], \\
	& \Psi_{21}^{+}(z)=e^{\frac{1}{4} \pi \nu_4}\left(\tilde\beta_{21}\right)^{-1}\left[\partial_z D_{-i\nu_4}\left(e^{-\frac{\pi i}{4}} k\right)-\frac{i z}{2} D_{-i\nu_4}\left(e^{-\frac{\pi i}{4}} z\right)\right],
\end{aligned}
$$
and
$$
\begin{aligned}
	& \Psi_{12}^{-}(z)=e^{\frac{1}{4} \pi \nu_1}\left(\tilde\beta_{21}\right)^{-1}\left[\partial_z D_{i\nu_4}\left(e^{\frac{\pi i}{4}} z\right)+\frac{i k}{2} D_{i\nu_4}\left(e^{\frac{\pi i}{4}} z\right)\right], \\
	& \Psi_{21}^{-}(z)=e^{-\frac{3}{4} \pi \nu_1}\left(\tilde\beta_{21}\right)^{-1}\left[\partial_z D_{-i\nu_4}\left(e^{\frac{3 \pi i}{4}} z\right)-\frac{i k}{2} D_{-i\nu_4}\left(e^{\frac{3 \pi i}{4}} z\right)\right].
\end{aligned}
$$
Since $\left(\Psi_-\right)^{-1}\Psi_+=V^B(y)$ and we can obtain that
$$
- y=\Psi^-_{11}\Psi^+_{21}-\Psi^-_{21}\Psi^+_{11}=\frac{e^{-\frac{\pi\nu_4}{2}}}{\tilde\beta_{12}}W\left(D_{-i\nu_4}(e^{\frac{3\pi i}{4}}z),D_{-i\nu_4}(e^{-\frac{\pi i}{4}}z)\right)=\frac{\sqrt{2\pi}e^{\frac{3\pi i}{4}}e^{-\frac{\pi\nu_4}{2}}}{\tilde\beta_{12}\Gamma(i\nu_4)},
$$
and
$$
\bar y=\Psi^-_{22}\Psi^+_{12}-\Psi^-_{12}\Psi^+_{22}=\frac{e^{-\frac{\pi\nu_4}{2}}}{\tilde\beta_{21}}W\left(D_{i\nu_4}(e^{\frac{\pi i}{4}}z),D_{i\nu_4}(e^{-\frac{3\pi i}{4}}z)\right)=\frac{\sqrt{2\pi}e^{\frac{\pi i}{4}}e^{-\frac{\pi\nu_4}{2}}}{\tilde\beta_{21}\Gamma(-i\nu_4)}.
$$
Now, we have
$$
\beta_{12}^B=i\tilde\beta_{12}=\frac{\sqrt{2\pi}e^{\frac{\pi i}{4}}e^{-\frac{\pi\nu_4}{2}}}{ y\Gamma(i\nu_4)},\quad \beta_{21}^B=-i\tilde\beta_{21}=\frac{\sqrt{2\pi}e^{-\frac{\pi i}{4}}e^{-\frac{\pi\nu_4}{2}}}{\bar y\Gamma(-i\nu_4)}.
$$

Next an explicit proof of our observation to illustrate the relation between $M^{(3,\epsilon)}$ RH problem and the model problem $M^X_{A,B}$ is given.

\begin{lemma}
	 The matrix function $H(\pm k_0,t)$ is uniformly bounded:
$$
\sup_{t\ge t_0}|\partial_x^lH(\pm k_0,t)|\le C, \quad 0<k_0<M
$$
for $l=0,1$.

The function $\delta_{A,B}^0,\delta_{A,B}^1$ and $e^{ t\Phi_{21}^0(\pm k_0)}$ satisfy the following properties:
$$
|\delta_A^0|=e^{2\pi\nu},\quad |\delta_B^0|=1, \quad 0<k_0<M\ \text{and}\ t \ge t_0.\\
|\partial_x\delta_A^0|\le \frac{C\ln t}{t},\quad |\delta_B^0|\le \frac{C\ln t}{t}, \quad 0<k_0<M\ \text{and}\ t \ge t_0.
$$
Moreover, we have
$$
|\delta_{A}^1(k)-1|\le C|k-k_0|(1+|\ln(|k- k_0|)|),\\
|\partial_x\delta_{A}^1(k)|\le\frac{C}{t}|\ln(|k-k_0|)|,
$$
and
$$
|\delta_{B}^1(k)-1|\le C|k+k_0|(1+|\ln(|k+ k_0|)|),\\
|\partial_x\delta_{B}^1(k)|\le\frac{C}{t}|\ln(|k+k_0|)|,
$$
$$
|\partial_x(e^{t\Phi_{21}^0(\pm k_0,z)}-1)|\le C\frac{k_0^2z^3}{t^{\frac{3}{2}}}e^{t\operatorname{Re} \Phi_{21}}.
$$
\end{lemma}

\begin{proof}
	Recall that $\delta_A^0=\frac{a^{-2i\nu_1}e^{-2\chi_1(k_0)}}{\tilde \delta_{v_1}(k_0)}$ and by the directly computation
$$
|a^{-2i\nu}|=|({3^{\frac{5}{4}}2\sqrt{5t}k_0^{\frac{3}{2}}})^{2i\nu_1}|=|e^{2i\nu_1 ln(a)}|=1,
$$
since the coefficients $\nu$ and $a$ are real valued, and
$$
|\tilde \delta_{v_1}(k_0)|=\left|\frac{\delta_3(k_0)\delta^2_4(k_0)\delta_5(k_0)}{\delta_6(k_0)\delta_2(k_0)}\right|=\left|\frac{\delta_1(\alpha^2k_0)\delta^2_4(k_0)\delta_1(\alpha k_0)}{\delta_4(\alpha k_0)\delta_2(\alpha^2 k_0)}\right|=1,
$$
where we using the fact that $\delta_{1,4}( k)={(\overline{\delta_{1,4}( \bar{k})})^{-1}}$ and the symmetric between $\delta_{1}$ (resp. $\delta_{4}$) and $\delta_{3,5}$ (resp. $\delta_{2,6}$).
\par
Furthermore, the real part of $\chi_1$ is that
$$
\operatorname{Re}\chi_1(k_0)=\frac{1}{2 \pi} \int_{k_0}^{\infty} \pi d \ln \left(1-\left|r_1(s)\right|^2\right)=-\frac{1}{2} \ln \left(1-\left|r_1\left(k_0\right)\right|^2\right)=\pi \nu_1,
$$
since we choose the branch cut from $0$ to $2\pi$.

So that
$$
|\delta_A^0|=\left|\frac{a^{-2i\nu}e^{-2\chi_1(k_0)}}{\delta_{\tilde v_1}(k_0)}\right|=e^{-2\pi\nu_1}.
$$
By the same procedure, we can get that
$$
\operatorname{Re}\chi_1(k_0)=\frac{1}{2 \pi} \int_{k_0}^{\infty} 0 d \ln \left(1-\left|r_1(s)\right|^2\right)=0,
$$
and
$$
|\delta_B^0|=\left|\frac{a^{-2i\nu}e^{-2\chi_4(k_0)}}{\tilde\delta_{ v_4}(k_0)}\right|=1.
$$

Moreover, using the before formula, we can get that

$$
\begin{aligned}
	\left|\partial_x \delta_A^0(\zeta, t)\right| & =\left|\delta_A^0(\zeta, t) \partial_x \ln \delta_A^0(\zeta, t)\right|=e^{-2 \pi \nu_1}\left|\partial_x \ln \delta_A^0(\zeta, t)\right| \\
	& \leq C\left(\left|\ln t \partial_x \nu_1\right|+\left|\partial_x \chi_1\left(k_0\right)\right|+\left|\partial_x \ln \tilde\delta_{v_1}\left( k_0\right) \right|\right),
\end{aligned}
$$
since $k_0=\sqrt[4]{\frac{x}{45t}}$ we can get that $\partial_x=\frac{1}{4k_0^3t}\partial_{k_0}$. Rewrite it as below
$$
|\partial_x\nu_1| \le C \frac{1}{t} \frac{\partial_k|r_1(k)|^2|_{k=k_0}}{1-|r_1(k_0)|^2}\le C\frac{1}{t},\\
\left|\partial_x \chi_1\left(k_0\right)\right|\le \frac{C}{t},\ \left|\partial_x \ln \tilde\delta_{v_1}\left( k_0\right) \right|\le \frac{C}{t} \left|\partial_{k_0} \ln \tilde\delta_{v_1}\left( k_0\right) \right|
$$
since the estimate about $\chi_1$ is above and $\tilde\delta_{v_1}$ is analytic near $k_0$.

Recall that $\delta_A^1=\frac{e^{2\chi_1(k_0)-2\chi_1(k)}\delta_{\tilde v_1}(k_0)}{\delta_{\tilde v_1}(k)}$, by the directly computation

$$
|e^{2\chi_1(k_0)-2\chi_1(k)}-1|\le C|\chi_1(k_0)-\chi_1(k)|\le C|k-k_0|(1+|\ln(|k-k_0|)|),
$$
and
$$
\partial_x \delta_A^1(k)=\delta_A^1( k) \partial_x \log \delta_A^1( k).
$$
Since the $\tilde \delta_{v_1}$ is analytic near $k_0$ and combine the estimate before, we can get that
$$
\left|\partial_x \delta_{A}^1(\zeta, k)\right| \leq C\left(\left|\partial_x\left(\chi_1( k)-\chi_1\left(k_0\right)\right)\right|+\frac{1}{t}\left|\partial_{k_0} \log \tilde\delta_{v_1}\right|\right)\le \frac{C |\ln(|k-k_0|)|}{t}.
$$

Finally, we have ${t\Phi_{21}^0}(k_0)=9\sqrt{3}ita^3[a^2z^5-5ak_0z^4+10k_0^2z^3]$ and taking the Taylor expansion yields
$$
|e^{t\Phi_{21}^0(\pm k_0,z)}-1|\le C\frac{k_0^2z^3}{t^{\frac{3}{2}}}e^{t\operatorname{Re} \Phi_{21}}.
$$
\end{proof}
In conclude, $M^{(2)}(x,t;k)=M^{(3,\epsilon)}(x,t;k)H(\pm k_0,t)^{-1} \to M^X_{A,B}(y,z)H(\pm k_0,t)^{-1}$ as $t\to\infty$ for $k\in\Sigma_{A,B}$. But on the boundary of $\partial B(\pm k_0,\epsilon)$ the RH problem $M^{X}_{A,B}H(\pm k_0,t)^{-1}$ not converge to the $I$ (as $t\to\infty$ the $z\to\infty$  )  this suggests that we need to introduce new RH problem defined by
$$
M^{(\pm k_0)}(x,t;k)=H(\pm k_0,t)M^X_{A,B}(y,z)H(\pm k_0,t)^{-1},\quad k\in B(\pm k_0,\epsilon),
$$
and the $H(\pm k_0,t)$ on the right hand does not change the jump matrix.

\begin{lemma}
	$M^{(\pm k_0)}$ is analytic for $k\in B(\pm k_0,\epsilon)\setminus\Sigma_{A,B}$ and satisfies the jump condition $M^{(\pm k_0)}_+=M^{(\pm k_0)}_-V^{(\pm k_0)}$ on $\Sigma_{A,B}$, respectively. Moreover, for $0<k_0<M$ and $t$ large enough we have the following estimate
$$
\|\partial_x^l(V^{(2)}-V^{(\pm k_0)})\|_{L^1(\Sigma_{A,B})}\le C\frac{\ln t}{t},
$$
and
$$
\|\partial_x^l(V^{(2)}-V^{(\pm k_0)})\|_{L^{\infty}(\Sigma_{A,B})}\le C\frac{\ln t}{t^{\frac{1}{2}}}.
$$
Furthermore,
$$
	\left\|\partial_x^l\left(M^{(\pm k_0)}(x, t, \cdot)^{-1}-I\right)\right\|_{L^{\infty}\left(\partial B\left(\pm k_0,\epsilon\right)\right)}  =O\left(t^{-1 / 2}\right), \\
$$
$$
\frac{1}{2 \pi i} \int_{\partial B_{\left(\pm k_0,\epsilon\right)}}\left(M^{(\pm k_0)}(x, t, k)^{-1}-I\right) d k  =-\frac{H(\pm k_0, t) \left(M^X_{A,B}(y)\right)^{(1)} H(\pm k_0, t)^{-1}}{a(t)}+O\left(t^{-1}\right).
$$
\end{lemma}

\begin{proof}
	Recall that
$$
M^{(k_0)}(x,t;k)=H(k_0,t)M^X_A(y,z)H(k_0,t)^{-1},\quad k\in B(k_0,\epsilon),
$$
where
$$
V^{(k_0)}(x,t;k)=H(k_0,t)V^{X_B}(y,z)H(k_0,t)^{-1},
$$
and
$$
V^{(2)}(x,t;k)=H(k_0,t)V^{(3,\epsilon)}(x,t;k)H(k_0,t)^{-1}.
$$
So that we get that
$$
V^{(2)}-V^{(k_0)}=H(k_0,t)\left(V^{(3,\epsilon)}-V^{X_B}\right)H(k_0,t).
$$
Since $H(k_0,t)^{\pm1}$ is bounded and this suggests us to show that
$$
\begin{aligned}
	& \left\|\partial_x^l\left[V^{(3,\epsilon)}(x, t; \cdot)-V^{X_A}(x, t, z(k_0, \cdot))\right]\right\|_{L^1\left(\mathcal{X}_j^\epsilon\right)} \leq C t^{-1} \ln t, \\
	& \left\|\partial_x^l\left[V^{(3,\epsilon)}(x, t; \cdot)-V^{X_A}(x, t, z(k_0, \cdot))\right]\right\|_{L^{\infty}\left(\mathcal{X}_j^\epsilon\right)} \leq C t^{-1 / 2} \ln t.
\end{aligned}
$$
Again, we have
$$
v^{(3,\epsilon)}_1=\left(\begin{array}{ccc}
	1 & 0 & 0 \\
	e^{-2i \nu_1 \log_0\left(z\right)}\delta_A^1 \rho^*_{1,a} e^{t \Phi_{21}^0(k_0,z)+\frac{iz^2}{2}} & 1 & 0 \\
	0 & 0 & 1
\end{array}\right),\quad
v^{X_A}_1=\left(\begin{array}{ccc}
	1 & 0 & 0 \\
	{\frac{\bar{y}}{1-|y|^2}} z^{-2 i \nu_1(y)} e^{\frac{i z^2}{2}} & 1 & 0 \\
	0 & 0 & 1
\end{array}\right).
$$
It suffices to show that
$$
\begin{aligned}
	&\left|e^{-2i \nu_1 \log_0\left(z\right)}\delta_A^1 \rho^*_{1,a} e^{t \Phi_{21}^0(k_0,z)+\frac{iz^2}{2}} -{\frac{\bar{y}}{1-|y|^2}} z^{-2 i \nu_1(y)} e^{\frac{i z^2}{2}}\right|\\
	&=|e^{-2i \nu_1 \log_0\left(z\right)}|\left|\delta_A^1 \rho^*_{1,a} e^{t \Phi_{21}^0(k_0,z)}-{\frac{\bar{y}}{1-|y|^2}}\right||e^{\frac{iz^2}{2}}|\\
	&\le C\left|(\delta_A^1-1)\rho^*_{1,a} e^{t \Phi_{21}^0(k_0,z)}+(e^{t \Phi_{21}^0(k_0,z)}-1) \rho^*_{1,a}+(\rho^*_{1,a}(k)-\rho^*_{1,a}(k_0))\right||e^{\frac{iz^2}{2}}|\\
	&\le C\left(|k-k_0|ln(|k-k_0|)+|k-k_0| \right)e^{t\operatorname{Re}\Phi_{21}}|e^{\frac{iz^2}{2}}|\\
	&\le C |k-k_0|(1+\ln(|k-k_0|))e^{-ct|k-k_0|^2},
\end{aligned}
$$
which implies that
$$
\left\|\left(v^{(3,\epsilon)}-v^{X_A}\right)_{21}\right\|_{L^1\left(\Sigma_A\right)} \leq C \int_0^{\infty} s(1+|\ln s|) e^{-c t u^2} d s \leq C t^{-1} \ln t,
$$
and
$$
\left\|\left(v^{(3,\epsilon)}-v^{X_A}\right)_{21}\right\|_{L^{\infty}\left(\Sigma_A\right)} \leq C \sup _{s \geq 0} s(1+|\ln s|) e^{-c t s^2} \leq C t^{-1 / 2} \ln t.
$$
For the $\partial_x\left(v^{(3,\epsilon)}-v^{X_A}\right)_{21}$, it follows that
$$
\begin{aligned}
	\partial_x\left(v^{(3,\epsilon)}-v^{X_A}\right)_{21}=&\partial_x(e^{-2i\nu_1\log_0(z)})\left((\delta_A^1-1)\rho^*_{1,a} e^{t \Phi_{21}^0(k_0,z)}+(e^{t \Phi_{21}^0(k_0,z)}-1) \rho^*_{1,a}+(\rho^*_{1,a}(k)-\rho^*_{1,a}(k_0)\right)e^{\frac{iz^2}{2}}\\
	&+e^{-2i\nu_1\log_0(z)}\partial_x\left((\delta_A^1-1)\rho^*_{1,a} e^{t \Phi_{21}^0(k_0,z)}+(e^{t \Phi_{21}^0(k_0,z)}-1) \rho^*_{1,a}+(\rho^*_{1,a}(k)-\rho^*_{1,a}(k_0)\right)e^{\frac{iz^2}{2}}\\
	&+e^{-2i\nu_1\log_0(z)}\left((\delta_A^1-1)\rho^*_{1,a} e^{t \Phi_{21}^0(k_0,z)}+(e^{t \Phi_{21}^0(k_0,z)}-1) \rho^*_{1,a}+(\rho^*_{1,a}(k)-\rho^*_{1,a}(k_0)\right)\partial_xe^{\frac{iz^2}{2}}\\
	&:=I +II+III.
\end{aligned}
$$
For the first part, using the fact that $|\partial_xe^{-2i\nu_1log_0(z)}|\le \frac{C}{t(k-k_0)}$, we have that
$$
\|I\|_{L^1(\Sigma_A)}\le C t^{-1} \int_0^{\infty}(1+\ln s) e^{-c t s^2} d s \leq C t^{-3 / 2} \ln t,\\
\|I\|_{L^1(\Sigma_A)} \leq C t^{-1} \sup _{u \geq 0}(1+\ln s) e^{-c t s^2} \leq C t^{-1} \ln t.
$$
For the second term and last term, using the lemma before, we have the same estimate.
\par
Since
$$
z_1={3^{\frac{5}{4}}2\sqrt{5t}k_0^{\frac{3}{2}}}(k-k_0),
$$
so for the $k\in\partial B(k_0,\epsilon)$, $z_1\to\infty$ as $t\to\infty$ and combining the WKB expansion of $M^{X_A}$, we have
$$
M^{X_A}(y,z)=I+\frac{M^{X_A}_1(y)}{{3^{\frac{5}{4}}2\sqrt{5t}k_0^{\frac{3}{2}}}(k-k_0)}+O(\frac{1}{t}),
$$
as $t\to\infty$, and
$$
\left(M^{(k_0)}\right)^{-1}-I=-\frac{H(k_0, t) M_1^{X_A}(y) H(k_0, t)^{-1}}{{3^{\frac{5}{4}}2\sqrt{5t}k_0^{\frac{3}{2}}}(k-k_0)}+O\left(t^{-1}\right), \quad t \rightarrow \infty.
$$
\end{proof}

\subsection{Small norm RH problem}

Now, using the symmetric property of RH problem, we make extension of $M^{(\pm k_0)}$ as follows
$$
\tilde M^{(\pm k_0)}=\mathcal{A}M^{(\pm k_0)}(x,t;\alpha k)\mathcal{A}^{-1}.
$$
Denote $\tilde B^{(\pm k_0)}_{\epsilon}=B(\pm k_0,\epsilon)\cup B(\pm \alpha k_0,\epsilon)\cup B(\pm \alpha^2 k_0,\epsilon)$ and introduce the solution $\tilde M(x,t;k)$ as following:
$$
\tilde M(x,t;k):=\begin{cases}\begin{aligned}
		&M^{(2)}\left(\tilde M^{(k_0)}\right)^{-1}, &&k\in \tilde B^{(k_0)}_{\epsilon}\\
		&M^{(2)}\left(\tilde M^{(-k_0)}\right)^{-1},&&k\in \tilde B^{(-k_0)}_{\epsilon},\\
		&M^{(2)}, && \text{otherelse}.
	\end{aligned}
\end{cases}
$$
Moreover, the jump  contour denoted as $\tilde{\Sigma}:=\Sigma^{(2)}\cup\partial \tilde B^{(k_0)}_{\epsilon}\cup \partial \tilde B^{(-k_0)}_{\epsilon}$ (see Fig. \ref{Partial-B.pdf})
and the jump matrix as follows
$$
\tilde{V}:=\begin{cases}\begin{aligned}
		&V^{(2)}, && k \in \tilde{\Sigma}\setminus \overline {\left(\tilde B^{(\pm k_0)}_{\epsilon}\right)},\\
		&(\tilde M^{(k_0)})^{-1}, && k \in \partial \tilde B^{(k_0)}_{\epsilon},\\
		&(\tilde M^{(-k_0)})^{-1}, && k \in \partial \tilde B^{(-k_0)}_{\epsilon},\\
		&\tilde M^{(k_0)}_{-}V^{(2)}(\tilde M^{(k_0)}_+)^{-1}, && k \in \tilde B^{(k_0)}_{\epsilon}\cap \tilde{\Sigma},\\
		&\tilde M^{(-k_0)}_{-}V^{(2)}(\tilde M^{(-k_0)}_+)^{-1}, && k \in \tilde B^{(-k_0)}_{\epsilon}\cap \tilde{\Sigma}.
		
\end{aligned}\end{cases}
$$
In conclusion, construct a RH problem that satisfies $\tilde M_+=\tilde M_-\tilde V$ for $k\in\tilde \Sigma$ and is analytic in $\C\setminus \tilde\Sigma$.

Suppose $\tilde\Sigma_{A,B}:=\Sigma_{A,B}\cup\alpha\Sigma_{A,B}\cup\alpha^2\Sigma_{A,B}$ and denote
$$
\Sigma':=\tilde{\Sigma}\setminus\left(\Sigma\cup\tilde\Sigma_{A,B}\cup\partial \tilde{B}^{(\pm k_0)}_{\epsilon}\right).
$$

\begin{figure}[!h]
	\centering
	\begin{overpic}[width=.85\textwidth]{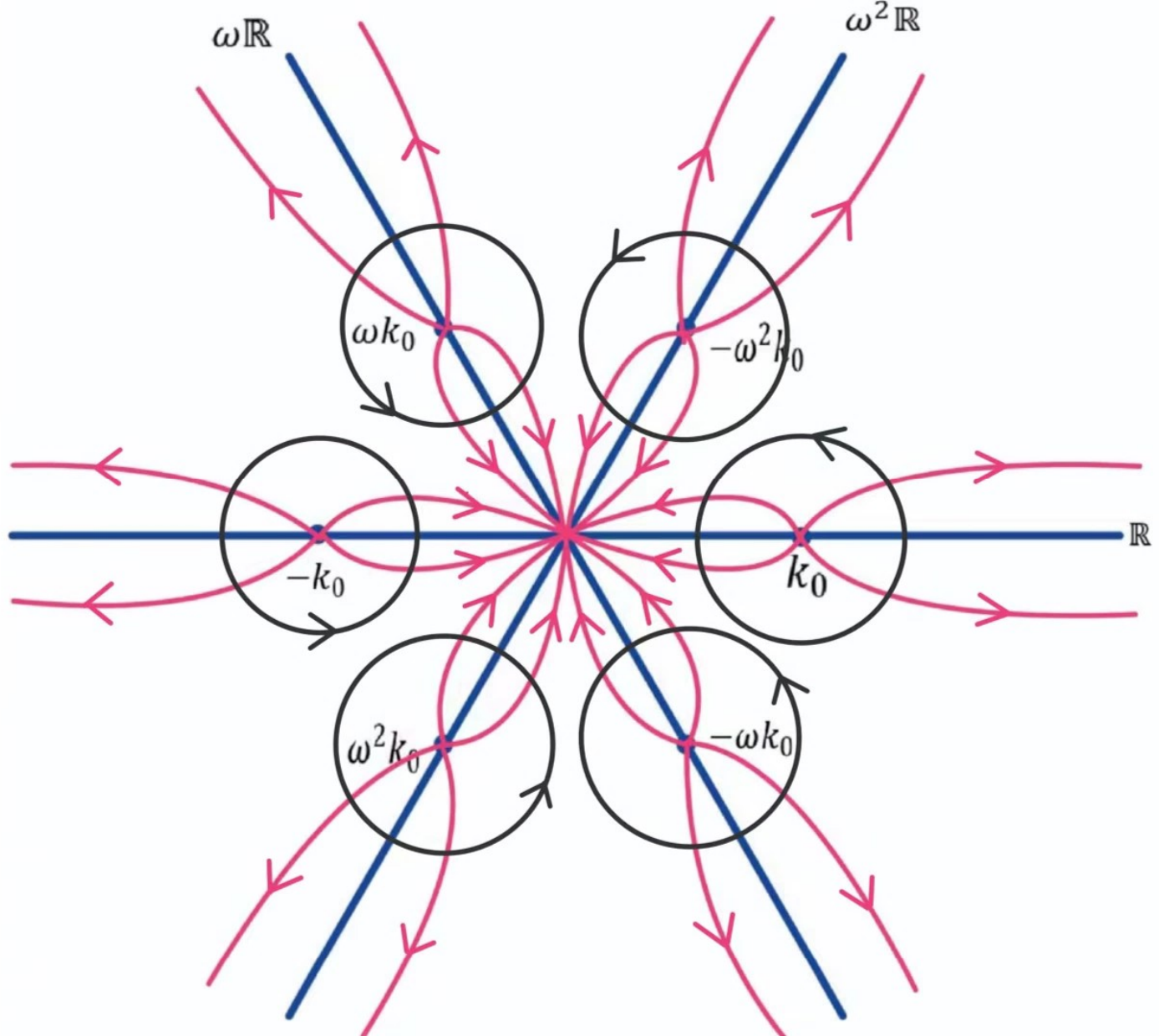}
	\end{overpic}
	\caption{{\protect\small
		The jump  contour $\tilde{\Sigma}:=\Sigma^{(2)}\cup\partial \tilde B^{(k_0)}_{\epsilon}\cup \partial \tilde B^{(-k_0)}_{\epsilon}$.}}
	\label{Partial-B.pdf}
\end{figure}

\begin{lemma}
	 Let $W=\tilde V-I$ and the estimate of jump matrix as followd is uniformly for $t$ large enough and $0<k_0<M$
$$
\begin{aligned}
	& \left\|(1+|\cdot|) \partial_x^l {W}\right\|_{\left(L^1 \cap L^{\infty}\right)(\Sigma)} \leq \frac{C}{k_0^3 t^{\frac{3}{2}}}, \\
	& \left\|(1+|\cdot|) \partial_x^l W\right\|_{\left(L^1 \cap L^{\infty}\right)\left(\Sigma^{\prime}\right)} \leq C e^{-c t}, \\
	& \left\|\partial_x^l W\right\|_{\left(L^1 \cap L^{\infty}\right)(\partial \tilde{B}^{(\pm k_0)})} \leq C t^{-1 / 2}, \\
	& \left\|\partial_x^l W\right\|_{L^1\left(\tilde{\Sigma}_{A,B}\right)} \leq C t^{-1} \ln t, \\
	& \left\|\partial_x^l W\right\|_{L^{\infty}\left(\tilde{\Sigma}_{A,B}\right)} \leq C t^{-1 / 2} \ln t.
\end{aligned}
$$
\end{lemma}

\begin{proof}
	(a) For the first inequality, recall the jump matrix on $\Sigma$ involves $r_{j,r},\rho_{j,r}$ and $(\tilde M^{(\pm k_0)})^{\pm1}$ is bounded, i.e.,
$$
\|(1+|\cdot|)\partial_x^l(v^{(2)}-I)\|_{(L^1\cap L^{\infty})(\Sigma_{5,6}^{(2)})}\le Ct^{-\frac{3}{2}}.
$$
So that on $\Sigma^{(2)}_{5,6}\cap B(k_0,\epsilon)$,
$$
W=\tilde V-I=m^{(k_0)}_-v^{(2)}\left(m^{(k_0)}_+\right)^{-1}-I=m^{(k_0)}_-\left(v^{(2)}-I\right)\left(m^{(k_0)}_+\right)^{-1}
$$
since the jump contour of $\tilde M^{(k_0)}$ is on $\tilde\Sigma_{A,B}$, so that $m^{(k_0)}$ is analytic on $\Sigma^{(2)}_{5,6}\cap B(k_0,\epsilon)$ and is bounded, we have
$$
\left\|(1+|\cdot|) \partial_x^l {W}\right\|_{\left(L^1 \cap L^{\infty}\right)(\Sigma)} \leq \frac{C}{k_0^3 t^{\frac{3}{2}}}. \\
$$

(b) For the second inequality, the $\Sigma'$ means $\Sigma^{(2)}\setminus\overline{B^{\pm k_0}_\epsilon}$, we would like to focus on $\Sigma^{(2)}\setminus B(k_0,\epsilon)$ and the $W$ on that contour only involves $(v^{(2)}_1)_{21}$ is not zero, i.e.,
$$
\frac{\delta_{1+}^2 }{\tilde\delta_{v_1}} \rho^*_{1,a} e^{t \Phi_{21}}.
$$
Because the $\partial_x^l\delta$ function is bounded, the estimate of $\rho^*_{1,a}$ is
$$
\left|\partial_x \rho_{1, a}^*(x, t, k)\right| \leq \frac{Ce^{t\operatorname{Re} \Phi_{21}( k)}}{1+|k|}.
$$
Moreover, the $\operatorname{Re}\Phi_{21}<-c$ for $|k-k_0|>\epsilon$ so that we have
$$
\left\|(1+|\cdot|) \partial_x^l W\right\|_{\left(L^1 \cap L^{\infty}\right)\left(\Sigma^{\prime}\right)} \leq C e^{-c t}. \\
$$

(c) The third estimate is the direct outcome of lemma above.

(d) For $k\in\tilde \Sigma_{A,B}$, in particular,
$$
W=\tilde M^{(k_0)}_{-}(V^{(2)}-V^{(k_0)})(\tilde M^{(k_0)}_+)^{-1},
$$
and combining the estimate above, it is found that $M^{(k_0)}$ is bounded uniformly for $0<k_0<m$.
\end{proof}

Now, introduce the Cauchy operator
$$
\left(\mathrm{C} f\right)(z)=\int_{\tilde\Sigma} \frac{f(\zeta)}{\zeta-z} \frac{\mathrm{d} \zeta}{2 \pi i}, \quad z \in \C\setminus \tilde \Sigma.
$$
If $(1+|z|)^{\frac{1}{3}}f(z)\in L^3(\tilde\Sigma)$, then $(Cf)(z)$ is analytic from $\C\setminus\tilde\Sigma$ to $\C$ with property that for any component $D$ in $\C\setminus \tilde \Sigma$, there are curves $\{C \}_{n=1}^\infty$ which surround each compact subset of $D$ satisfy
$$
\sup_{n\ge1}\int_{C_n}(1+|z|)|f(z)|^3|dz|<\infty,
$$
moreover, $C_{\pm}f$ exist a.e. $z\in\tilde\Sigma$ and $(1+|z|)^{\frac{1}{3}}C_{\pm}f(z)\in L^3(\tilde\Sigma)$.

On one hand, the $C_{\pm}$ are bounded operator from weighted $L^3(\tilde\Sigma)$ to itself (denote it as $\dot L^{3}(\tilde\Sigma)$) and satisfy $C_+-C_-=I$.

On the other hand, recall the estimate before
$$
\begin{cases}\begin{aligned}
		&\|(1+|\cdot|)\partial_x^lW\|_{L^1(\tilde\Sigma)}\le C t^{-\frac{1}{2}},\\
		&\|(1+|\cdot|)\partial_x^lW\|_{L^{\infty}(\tilde\Sigma)}\le C t^{-\frac{1}{2}}\ln t.\\
\end{aligned}\end{cases}
$$
By the Riesz in interpolation inequality, we have
$$
\|(1+|\cdot|)\partial_x^lW\|_{L^p(\tilde\Sigma)}\le C t^{-\frac{1}{2}}(\ln t)^{\frac{1}{p}},
$$
so that $W$ belong to the weighted $L^3(\tilde\Sigma)$ and $L^{\infty}(\tilde\Sigma)$.

Define
$$
C_{W} f=C_{+}\left(fW_{-}\right)+C_{-}\left(fW_{+}\right)
$$
and $C_W:\dot L^3(\tilde\Sigma)+ L^{\infty}(\tilde\Sigma)\to \dot L^3(\tilde\Sigma)$.

\begin{lemma}
	 For $t$ large enough and $0<k_0<M$, the operator $I-C_w$ is invertible and $(I-C_w)^{-1}$is a bounded linear operator from $\dot L^3(\tilde\Sigma)$ to itself.
\end{lemma}

\begin{proof}
	Since  $C_{\pm}$ are bounded operator from weighted $L^3(\tilde\Sigma)$ to itself then for any $f\in\dot L^3(\tilde\Sigma)$, we have
$$
\begin{aligned}
	C_Wf&=\mathrm{C}_{+}\left(fW_{-}\right)+\mathrm{C}_{-}\left(fW_{+}\right)\\
	&\le \|C_+\|_{\dot L^3(\tilde\Sigma)\to\dot L^3(\tilde\Sigma)}\|W\|_{L^\infty(\tilde\Sigma)}\|f\|_{\dot L^3(\tilde\Sigma)}+\|C_-\|_{\dot L^3(\tilde\Sigma)\to\dot L^3(\tilde\Sigma)}\|W\|_{L^\infty(\tilde\Sigma)}\|f\|_{\dot L^3(\tilde\Sigma)}\\
	&\le \left(\|C_+\|_{\dot L^3(\tilde\Sigma)\to\dot L^3(\tilde\Sigma)}+\|C_-\|_{\dot L^3(\tilde\Sigma)\to\dot L^3(\tilde\Sigma)}\right)\|W\|_{L^\infty(\tilde\Sigma)}\|f\|_{\dot L^3(\tilde\Sigma)},\\
\end{aligned}
$$
so that $\|C_W\|_{\dot L^3(\tilde\Sigma)\to\dot L^3(\tilde\Sigma)}\le \left(\|C_+\|_{\dot L^3(\tilde\Sigma)\to\dot L^3(\tilde\Sigma)}+\|C_-\|_{\dot L^3(\tilde\Sigma)\to\dot L^3(\tilde\Sigma)}\right)\|W\|_{L^\infty(\tilde\Sigma)}$ , and by the estimate before
$$
\|(1+|\cdot|)\partial_x^lW\|_{L^{\infty}(\tilde\Sigma)}\le C t^{-\frac{1}{2}}\ln t,\\
$$
as $t$ large enough to satisfy $\|W\|_{L^\infty(\tilde\Sigma)}<\frac{1}{\left(\|C_+\|_{\dot L^3(\tilde\Sigma)\to\dot L^3(\tilde\Sigma)}+\|C_-\|_{\dot L^3(\tilde\Sigma)\to\dot L^3(\tilde\Sigma)}\right)}$, then the operator $I-C_w$ is invertible.
\end{proof}

Let $\mu\in I+\dot L^3(\tilde \Sigma)$ satisfy the following equation
$$
\mu=I+C_W\mu,
$$
furthermore, one has $\mu=I+(I-C_W)^{-1}C_wI$.

\begin{lemma}
	For $t$ large enough and $0<k_0<M$, the RH problem $\tilde M$ has a unique solution as following:
$$
\tilde M(x,t;k)=I+C(\mu W)=I+\int_{\tilde\Sigma} \frac{\mu(x,t;\zeta)W(x,t;\zeta)}{\zeta-z} \frac{\mathrm{d} \zeta}{2 \pi i}, \quad z \in \C\setminus \tilde \Sigma.
$$
\end{lemma}

\begin{lemma}
	For $t$ large enough, we have
$$
\|\partial_x^l(\mu-I)\|_{L^p(\tilde \Sigma)}\le  {\frac{C(\ln t)^{\frac{1}{p}}}{t^{\frac{1}{2}}}}.
$$
\end{lemma}

\begin{proof}
	Denote $\|C_{\pm}\|_p:=\left(\|C_+\|_{L^p(\tilde\Sigma)\to L^p(\tilde\Sigma)}+\|C_-\|_{ L^p(\tilde\Sigma)\to L^p(\tilde\Sigma)}\right)$ and assume $t$ large enough to satisfy $\|W\|_{L^{\infty}(\tilde\Sigma)}<\|C_{\pm}\|_p^{-1}$.

When $l=0$, we have
$$
\begin{aligned}
	\|\mu-I\|_{L^P(\tilde\Sigma)}&=\|(I-C_W)^{-1}C_wI\|_{L^P(\tilde\Sigma)}
	\le \sum_{j=1}^\infty\|C_W\|^{j-1}_{L^P(\tilde\Sigma)\to L^P(\tilde\Sigma)}\|C_wI\|_{L^p(\tilde\Sigma)}\\
	&\le \sum_{j=1}^\infty\|C_{\pm}\|^{j}_{p}\|W\|^{j-1}_{L^\infty(\tilde \Sigma)}\|w\|_{L^p(\tilde\Sigma)}=\frac{\|C_{\pm}\|_p \|w\|_{L^p(\tilde\Sigma)}}{1-\|C_{\pm}\|_p\|w\|_{L^{\infty}(\tilde\Sigma)}}.
\end{aligned}
$$
So that combine the estimate of $\|w\|_{L^p(\tilde\Sigma)}$ we get the aimed estimate.

When $l=1$, we get
$$
\partial_x(\mu-I)=\partial_x\sum_{j=1}^\infty(C_W)^jI.
$$
Since the series is uniformly bounded and we can change the order of sum and derivative, then we have
$$
\begin{aligned}
	\|\partial_x(\mu-I)\|_{L^P(\tilde\Sigma)}&\leq \sum_{j=2}^{\infty}(j-1)\left\|{{C}}_W\right\|_{L^p(\tilde{\Sigma})\to L^p(\tilde{\Sigma}}^{j-2}\left\|\partial_x C_W\right\|_{L^p(\tilde{\Sigma})\to L^p(\tilde{\Sigma}}\left\|C_W I\right\|_{L^p(\tilde{\Sigma})}\\
	&+\sum_{j=1}^{\infty}\left\|C_W\right\|_{L^p(\tilde{\Sigma})\to L^p(\tilde{\Sigma}}^{j-1}\left\|\partial_x {C_W} I\right\|_{L^p(\tilde{\Sigma})}\\
	&\le C\sum_{j=2}^\infty j\|C_{\pm}\|_p^j\|W\|_{L^\infty(\tilde\Sigma)}^{j-2}\|\partial_xW\|_{L^\infty(\tilde\Sigma)}\|W\|_{L^p(\tilde\Sigma)}+\sum_{j=1}^\infty\|C_{\pm}\|_p^j\|W\|_{L^\infty(\tilde\Sigma)}^{j-1}\|\partial_xW\|_{L^p(\tilde\Sigma)}\\
	&\le\frac{C\left(\|\partial_xW\|_{L^\infty(\tilde\Sigma)}\|W\|_{L^p(\tilde\Sigma)}+ \|\partial_xW\|_{L^p(\tilde\Sigma)}\right)}{1-\|C_{\pm}\|_p\|W\|_{L^\infty(\tilde\Sigma)}}.
\end{aligned}
$$
\end{proof}

Now, we can get the following non-tangential limit as $k\to\infty$
$$
Q(x,t):=\lim_{k\to\infty}k(\tilde M(x,t;k)-I)=-\frac{1}{2\pi i}\int_{\tilde\Sigma}\mu (x,t;k)W(x,t;k)dk.
$$

\begin{lemma}
	 When $t\to\infty$, we have
$$
Q(x,t)=-\frac{1}{2\pi i}\int_{\partial\tilde B^{(k_0)}_\epsilon\cup\partial\tilde B^{(-k_0)}_\epsilon}W(x,t;k)dk+O\left(\frac{\ln t}{t}\right).
$$
\end{lemma}

\begin{proof}
	Decompose $Q(x,t)$ as follows
$$
Q(x,t)=-\frac{1}{2\pi i}\int_{\partial\tilde B^{(k_0)}_\epsilon\cup\partial\tilde B^{(-k_0)}_\epsilon}\mu (x,t;k)W(x,t;k)dk+Q_1(x,t)+Q_2(x,t),
$$
where
$$
Q_1(x,t):=-\frac{1}{2\pi i}\int_{\tilde\Sigma}(\mu (x,t;k)-I)W(x,t;k)dk,
$$
and
$$
Q_2(x,t):=-\frac{1}{2\pi i}\int_{\tilde\Sigma\setminus{(\partial\tilde B^{(k_0)}_\epsilon}\cup \partial\tilde B^{(k_0)}_\epsilon)}W(x,t;k)dk.
$$
For $Q_1(x,t)$, by the H$ \ddot o$lder inequality, we have that
$$
|Q_1(x,t)|\le C\|(\mu (x,t;\cdot)-I)W(x,t;\cdot)\|_{L^1(\tilde\Sigma)}\le C \|\mu (x,t;\cdot)-I\|_{L^p(\tilde\Sigma)}\|W(x,t;\cdot)\|_{L^q(\tilde\Sigma)}\le\frac{C\ln t}{t},
$$
where $\frac{1}{p}+\frac{1}{q}=1$.

For $Q_2(x,t)$, we have that
$$
|Q_2(x,t)|\le C\|W(x,t;\cdot)\|_{L^1(\tilde \Sigma \setminus(\partial\tilde B^{(k_0)}_\epsilon\cup\partial\tilde B^{(-k_0)}_\epsilon))}\le\frac{C\ln t}{t}.
$$

Now, suppose
$$
R(x,t;k_0):=-\frac{1}{2\pi i}\int_{\partial B(k_0,\epsilon)}W(x,t;k)dk=-\frac{1}{2\pi i}\int_{\partial B(k_0,\epsilon)}((M^{(k_0)})^{-1}-I) dk,
$$
and
$$
R(x,t;-k_0):=-\frac{1}{2\pi i}\int_{\partial B(k_0,\epsilon)}W(x,t;k)dk=-\frac{1}{2\pi i}\int_{\partial B(-k_0,\epsilon)}((M^{(-k_0)})^{-1}-I) dk,
$$
moreover, $R(x,t;\pm k_0)$ satisfy
$$
R(x,t;k_0)=\frac{H(k_0, t) M_1^{X_A}(y(k_0)) H(k_0, t)^{-1}}{{3^{\frac{5}{4}}2\sqrt{5t}k_0^{\frac{3}{2}}}}+O\left(t^{-1}\right),
$$
and
$$
R(x,t;-k_0)=\frac{H(-k_0, t) M_1^{X_B}(y(-k_0)) H(-k_0, t)^{-1}}{{3^{\frac{5}{4}}2\sqrt{5t}k_0^{\frac{3}{2}}}}+O\left(t^{-1}\right).
$$

By the symmetry of $\tilde M(x,t;k)$, we have
$$
\tilde M(x,t;k)=\mathcal{A}\tilde M(x,t;\alpha k)\mathcal{A}^{-1}, \quad k\in\C\setminus\tilde \Sigma.
$$
Then $\mu$ and $W$ also satisfy the above symmetric and we can find that
$$
\begin{aligned}
	-\frac{1}{2\pi i}\int_{\partial\tilde B^{(k_0)}_\epsilon\cup\partial\tilde B^{(-k_0)}_\epsilon}W(x,t;k)dk&=-\frac{1}{2\pi i}
	\left(\int_{\partial B(k_0,\epsilon)}+\int_{\partial B(\alpha k_0,\epsilon)}+\int_{\partial B(\alpha^2k_0,\epsilon)}\right)
	W(x,t;k)dk\\
	&-\frac{1}{2\pi i}
	\left(\int_{\partial B(-k_0,\epsilon)}+\int_{\partial B(-\alpha k_0,\epsilon)}+\int_{\partial B(-\alpha^2k_0,\epsilon)}\right)
	W(x,t;k)dk\\
	&=R(x,t;k_0)+\alpha \mathcal{A}^{-1}R(x,t;k_0)\mathcal{A}+\alpha^2 \mathcal{A}^{-2}R(x,t;k_0)\mathcal{A}^2\\
	&+R(x,t;-k_0)+\alpha \mathcal{A}^{-1}R(x,t;-k_0)\mathcal{A}+\alpha^2 \mathcal{A}^{-2}R(x,t;-k_0)\mathcal{A}^2.
\end{aligned}
$$
Therefore, we get that

\begin{align*}
	\partial_x\lim_{k\to\infty}k(\tilde M(x,t;k)-I)=\partial_x\left(\frac{\sum_{j=0}^2\alpha^j\mathcal{A}^{-j}H(k_0, t) M_1^{X_A}(y(k_0)) H(k_0, t)^{-1}\mathcal{A}^{j}}{{3^{\frac{5}{4}}2\sqrt{5t}k_0^{\frac{3}{2}}}}\right)\\
+\partial_x\left(\frac{\sum_{j=0}^2\alpha^j\mathcal{A}^{-j}H(-k_0, t) M_1^{X_B}(y(-k_0)) H(-k_0, t)^{-1}\mathcal{A}^{j}}{{3^{\frac{5}{4}}2\sqrt{5t}k_0^{\frac{3}{2}}}}\right).
\end{align*}

\end{proof}

\subsection{The long-time asymptotics of solutions for the KK equation and SK equation}

Since the SK equation (\ref{SK}) and KK equation (\ref{KK}) have the same reconstruction formula, the asymptotic solution $u(x,t)$ is
$$
u(x,t)=-\frac{1}{2}\partial_x\left(\lim_{k\to\infty}k(N_3(x,t;k)-1) \right),
$$
where $N(x,t;k)=\left(N_1,N_2,N_3\right)=(\alpha,\alpha^2,1)M(x,t;k)$.
\par
Recall that for $k\in\C\setminus{(\tilde B^{(k_0)}_{\epsilon}\cup\tilde B^{(-k_0)}_{\epsilon})}$, the $M(x,t;k)$ has the relationship with $\tilde M(x,t;k)$ as follows
$$
M=\tilde MG^{-1}\Delta^{-1}.
$$
Then it follows
$$
\begin{aligned}
	u(x,t)&=-\frac{1}{2}\partial_x\left(\lim_{k\to\infty}k[((\alpha,\alpha^2,1)\tilde MG^{-1}\Delta^{-1})_3-1] \right)\\
	&=-\frac{1}{2}\partial_x\left(\lim_{k\to\infty}k[((\alpha,\alpha^2,1)\tilde M)_3-1] \right)+O\left(\frac{\ln t}{t}\right).
\end{aligned}
$$
For the second equality, since the $G^{-1}\Delta^{-1}$ tends to $I$ as $k\to\infty$ and their derivatives are dominated by $\frac{\ln t}{t}$, it is concluded that
\begin{align*}
	u(x,t)=&-\frac{1}{2}\partial_x\left(\left(\begin{array}{lll}
		\alpha & \alpha^2 & 1
	\end{array}\right) \frac{\sum_{j=0}^2 \alpha^j \mathcal{A}^{-j} H(k_0, t) M_1^{X_A}(y(k_0)) H(k_0, t)^{-1} \mathcal{A}^j}{{3^{\frac{5}{4}}2\sqrt{5t}k_0^{\frac{3}{2}}}}\right)\\
	&-\frac{1}{2}\partial_x\left(\left(\begin{array}{lll}
		\alpha & \alpha^2 & 1
	\end{array}\right) \frac{\sum_{j=0}^2 \alpha^j \mathcal{A}^{-j} H(-k_0, t) M_1^{X_B}(y(-k_0)) H(-k_0, t)^{-1} \mathcal{A}^j}{{3^{\frac{5}{4}}2\sqrt{5t}k_0^{\frac{3}{2}}}}\right)+O\left(\frac{\ln t}{t}\right)\\
	=&-\frac{1}{2}\partial_x\left(\frac{\alpha^2\beta^A_{21}\delta^0_Ae^{t\Phi_{21}(k_0)}+\alpha\beta^A_{12}(\delta^0_A)^{-1}e^{-t\Phi_{21}(k_0)}
		+\alpha^2\beta^B_{21}(\delta^0_B)^{-1}e^{t\Phi_{21}(-k_0)}+\alpha\beta^B_{12}\delta^0_Be^{-t\Phi_{21}(-k_0)}}{{3^{\frac{5}{4}}2\sqrt{5t}k_0^{\frac{3}{2}}}}\right)+O\left(\frac{\ln t}{t}\right)\\
\end{align*}
\begin{align*}	=&-\frac{1}{3^{\frac{5}{4}}2\sqrt{5t}k_0^{\frac{3}{2}}}\left(\partial_x\operatorname{Re}\left(\alpha^2\beta^A_{21}\delta^0_Ae^{t\Phi_{21}(k_0)}\right)+\partial_x\operatorname{Re}\left(\alpha\beta^B_{12}\delta^0_Be^{-t\Phi_{21}(-k_0)}\right)\right)+O\left(\frac{\ln t}{t}\right)\\
	=&-\frac{1}{3^{\frac{5}{4}}2\sqrt{5t}k_0^{\frac{3}{2}}}(\sqrt{\nu_1}\partial_x\cos \left(\frac{19\pi }{12}-\left(\arg y_1+\arg\Gamma(i\nu_1)\right)-(36\sqrt{3}tk_0^5)+\nu_1\ln({{3^{\frac{7}{2}}20tk_0^{5}}})\right.\\
	&\left.+\nu_4\ln(4)+\frac{1}{\pi }\int_{-k_0}^{-\infty}\log_{\pi}\frac{|s-\omega k_0|}{|s- k_0|}d\ln(1-|r_2(s)|^2)+\frac{1}{\pi }\int_{k_0}^\infty\log_0\frac{|s-k_0|}{|s-\omega k_0|}d\ln(1-|r_1(s)|^2)\right)\\
	&+\sqrt{\nu_4}\partial_x\cos \left(\frac{11\pi }{12}-\left(\arg y_4+\arg\Gamma(i\nu_4)\right)-(36\sqrt{3}tk_0^5)+\nu_4\ln({{3^{\frac{7}{2}}20tk_0^{5}}})\right.\\
	&\left.+\nu_1\ln(4)+\frac{1}{\pi }\int_{k_0}^{\infty}\log_{0}\frac{|s+\omega k_0|}{|s+ k_0|}d\ln(1-|r_1(s)|^2)+\frac{1}{\pi }\int_{-k_0}^{-\infty}\log_{\pi}\frac{|s+ k_0|}{|s+\omega k_0|}d\ln(1-|r_2(s)|^2)\right))+O\left(\frac{\ln t}{t}\right)\\
\end{align*}
\begin{align*}
	=&-\frac{1}{3^{\frac{3}{4}}2\sqrt{5t}k_0^{\frac{1}{2}}}(\sqrt{\nu_1}\sin \left(\frac{19\pi }{12}-\left(\arg y_1+\arg\Gamma(i\nu_1)\right)-(36\sqrt{3}tk_0^5)+\nu_1\ln({{3^{\frac{7}{2}}20tk_0^{5}}})\right.\\
	&\left.+\nu_4\ln(4)+\frac{1}{\pi }\int_{-k_0}^{-\infty}\log_{\pi}\frac{|s-\omega k_0|}{|s- k_0|}d\ln(1-|r_2(s)|^2)+\frac{1}{\pi }\int_{k_0}^\infty\log_0\frac{|s-k_0|}{|s-\omega k_0|}d\ln(1-|r_1(s)|^2)\right)\\	
	&+\sqrt{\nu_4}\sin \left(\frac{11\pi }{12}-\left(\arg y_4+\arg\Gamma(i\nu_4)\right)-(36\sqrt{3}tk_0^5)+\nu_4\ln({{3^{\frac{7}{2}}20tk_0^{5}}})\right.\\
	&\left.+\nu_1\ln(4)+\frac{1}{\pi }\int_{k_0}^{\infty}\log_{0}\frac{|s+\omega k_0|}{|s+ k_0|}d\ln(1-|r_1(s)|^2)+\frac{1}{\pi }\int_{-k_0}^{-\infty}\log_{\pi}\frac{|s+ k_0|}{|s+\omega k_0|}d\ln(1-|r_2(s)|^2)\right))+O\left(\frac{\ln t}{t}\right),\\
\end{align*}
where we have used the fact that $\overline {\beta^{A,B}_{12}}=\beta^{A,B}_{21}$, $\overline{\delta^0_{A,B}}=(\delta^0_{A,B})^{-1}$ and $\frac{d}{dx}=\frac{d}{dk_0}\frac{1}{180tk_0^3}$.

\begin{theorem}
	 Suppose $u(x,t)$ is a Schwartz class solution of $SK$ equation (\ref{SK}) (or $KK$ equation (\ref{KK})) with initial data $u_0(x)$ in Schwartz space, then in the generic case and $\frac{x}{t}$ in a compact subsets of $(0,\infty)$, the solution has the following asymptotics as $t\to\infty$

\begin{equation}
	\begin{split}
			\begin{aligned}\label{longtime-solution}
				u(x,t)=&-\frac{1}{3^{\frac{3}{4}}2\sqrt{5t}k_0^{\frac{1}{2}}}(\sqrt{\nu_1}\sin \left(\frac{19\pi }{12}-\left(\arg y_1+\arg\Gamma(i\nu_1)\right)-(36\sqrt{3}tk_0^5)+\nu_1\ln({{3^{\frac{7}{2}}20tk_0^{5}}})\right.\\
		&\left.+\nu_4\ln(4)+\frac{1}{\pi }\int_{-k_0}^{-\infty}\log_{\pi}\frac{|s-\omega k_0|}{|s- k_0|}d\ln(1-|r_2(s)|^2)+\frac{1}{\pi }\int_{k_0}^\infty\log_0\frac{|s-k_0|}{|s-\omega k_0|}d\ln(1-|r_1(s)|^2)\right)\\	
		&+\sqrt{\nu_4}\sin \left(\frac{11\pi }{12}-\left(\arg y_4+\arg\Gamma(i\nu_4)\right)-(36\sqrt{3}tk_0^5)+\nu_4\ln({{3^{\frac{7}{2}}20tk_0^{5}}})\right.\\
		&\left.+\nu_1\ln(4)+\frac{1}{\pi }\int_{k_0}^{\infty}\log_{0}\frac{|s+\omega k_0|}{|s+ k_0|}d\ln(1-|r_1(s)|^2)+\frac{1}{\pi }\int_{-k_0}^{-\infty}\log_{\pi}\frac{|s+ k_0|}{|s+\omega k_0|}d\ln(1-|r_2(s)|^2)\right))\\&+O\left(\frac{\ln t}{t}\right).\\
			\end{aligned}
	\end{split}
\end{equation}
\end{theorem}

\begin{proof}
	To be specific, denote
$$
y_1=r_1(k_0),\ y_4=r_2(-k_0),\ \nu_1=-\frac{1}{2 \pi} \ln \left(1-\left|r_1\left(k_0\right)\right|^2\right),\ \nu_4=-\frac{1}{2 \pi} \ln \left(1-\left|r_2\left(-k_0\right)\right|^2\right),
$$

$$
\beta^A_{21}=\sqrt{\nu_1}\exp i\left(\frac{\pi}{4}-\arg y-\arg\Gamma(i\nu_1)\right),\quad e^{t\Phi_{21}(k_0)}=\exp-i(36\sqrt{3}tk_0^5).
$$

Recall
$$
\delta^0_A=\frac{a^{-2i\nu_1}e^{-2\chi_1(k_0)}}{\tilde\delta_{v_1}(k_0)},
$$
and
$$
a^{-2i\nu_1}=\exp\left(i\nu_1\ln({{3^{\frac{5}{2}}20tk_0^{3}}})\right),\quad
e^{-2\chi_1(k_0)}=\exp\left(-\frac{1}{\pi i}\int_{k_0}^\infty\log_0|s-k_0|d\ln(1-|r_1(s)|^2)\right),
$$
otherwise,

\begin{align*}
	(\tilde \delta_{v_1})^{-1}=\frac{\delta_2(k_0)\delta_6(k_0)}{\delta_3(k_0)\delta_5(k_0)\delta_4^2(k_0)}=\exp\left(i\nu_4\ln(4)-\frac{1}{\pi i}\int_{-k_0}^{-\infty}\log_{\pi}\frac{|s-\omega k_0|}{|s- k_0|}d\ln(1-|r_2(s)|^2)\right)\\
*\exp\left(i\nu_1\ln(3k_0^2)+\frac{1}{\pi i}\int_{k_0}^\infty\log_0|s-\omega k_0|d\ln(1-|r_1(s)|^2)\right),
\end{align*}
with
\begin{align*}
	&\delta_3(k_0)\delta_5(k_0)=\delta_1(\omega^2k_0)\delta_1(\omega k_0)=\exp\left(-i\nu_1\ln(3k_0^2)-\frac{1}{\pi i}\int_{k_0}^\infty\log_0|s-\omega k_0|d\ln(1-|r_1(s)|^2)\right),\\
	&\delta_2(k_0)\delta_6(k_0)=\delta_4(\omega^2k_0)\delta_4(\omega k_0)=\exp\left(-i\nu_4\ln(k_0^2)-\frac{1}{\pi i}\int_{-k_0}^{-\infty}\log_{\pi}|s-\omega k_0|d\ln(1-|r_2(s)|^2)\right),\\
	&\delta_4^2(k_0)=\exp\left(-i\nu_4\ln(4k_0^2)-\frac{1}{\pi i}\int_{-k_0}^{-\infty}\log_{\pi}|s- k_0|d\ln(1-|r_2(s)|^2)\right).
\end{align*}

In conclusion, we have

\begin{align*}
	\alpha^2\beta^A_{21}\delta^0_Ae^{t\Phi_{21}(k_0)}=\sqrt{\nu_1}\exp \left(\frac{4\pi i}{3}+i\left(\frac{\pi}{4}-\arg y-\arg\Gamma(i\nu_1)\right)-i(36\sqrt{3}tk_0^5)+i\nu_1\ln({{3^{\frac{7}{2}}20tk_0^{5}}})\right.\\
\left.+i\nu_4\ln(4)-\frac{1}{\pi i}\int_{-k_0}^{-\infty}\log_{\pi}\frac{|s-\omega k_0|}{|s- k_0|}d\ln(1-|r_2(s)|^2)-\frac{1}{\pi i}\int_{k_0}^\infty\log_0\frac{|s-k_0|}{|s-\omega k_0|}d\ln(1-|r_1(s)|^2)\right)\\
=\sqrt{\nu_1}\exp \left(\frac{19\pi i}{12}-i\left(\arg y+\arg\Gamma(i\nu_1)\right)-i(36\sqrt{3}tk_0^5)+i\nu_1\ln({{3^{\frac{7}{2}}20tk_0^{5}}})\right.\\
\left.+i\nu_4\ln(4)-\frac{1}{\pi i}\int_{-k_0}^{-\infty}\log_{\pi}\frac{|s-\omega k_0|}{|s- k_0|}d\ln(1-|r_2(s)|^2)-\frac{1}{\pi i}\int_{k_0}^\infty\log_0\frac{|s-k_0|}{|s-\omega k_0|}d\ln(1-|r_1(s)|^2)\right).
\end{align*}

On the other hand, we have

\begin{align*}
	\alpha\beta^B_{12}\delta^0_Be^{-t\Phi_{21}(-k_0)}=\sqrt{\nu_4}\exp \left(\frac{2\pi i}{3}+i\left(\frac{\pi}{4}-\arg y-\arg\Gamma(i\nu_4)\right)-i(36\sqrt{3}tk_0^5)+i\nu_4\ln({{3^{\frac{7}{2}}20tk_0^{5}}})\right.\\
\left.+i\nu_1\ln(4)-\frac{1}{\pi i}\int_{k_0}^{\infty}\log_{0}\frac{|s+\omega k_0|}{|s+ k_0|}d\ln(1-|r_1(s)|^2)-\frac{1}{\pi i}\int_{-k_0}^{-\infty}\log_{\pi}\frac{|s+k_0|}{|s+\omega k_0|}d\ln(1-|r_2(s)|^2)\right)\\
=\sqrt{\nu_4}\exp \left(\frac{11\pi i}{12}-i\left(\arg y+\arg\Gamma(i\nu_4)\right)-i(36\sqrt{3}tk_0^5)+i\nu_4\ln({{3^{\frac{7}{2}}20tk_0^{5}}})\right.\\
\left.+i\nu_1\ln(4)-\frac{1}{\pi i}\int_{k_0}^{\infty}\log_{0}\frac{|s+\omega k_0|}{|s+ k_0|}d\ln(1-|r_1(s)|^2)-\frac{1}{\pi i}\int_{-k_0}^{-\infty}\log_{\pi}\frac{|s+ k_0|}{|s+\omega k_0|}d\ln(1-|r_2(s)|^2)\right),
\end{align*}
where
$$
\beta^B_{12}=\sqrt{\nu_4}\exp i\left(\frac{\pi}{4}-\arg y-\arg\Gamma(i\nu_4)\right),\quad e^{-t\Phi_{21}(-k_0)}=\exp-i(36\sqrt{3}tk_0^5),
$$
again,
$$
\delta^0_B=\frac{a^{-2i\nu_4}e^{-2\chi_4(-k_0)}}{\tilde\delta_{v_4}(k_0)},
$$
and
$$
a^{-2i\nu_4}=\exp\left(i\nu_4\ln({{3^{\frac{5}{2}}20tk_0^{3}}})\right),\quad
e^{-2\chi_4(-k_0)}=\exp\left(-\frac{1}{\pi i}\int_{-k_0}^{-\infty}\log_{\pi}|s+k_0|d\ln(1-|r_2(s)|^2)\right),
$$
and
\begin{align*}
	(\tilde \delta_{v_4})^{-1}=\frac{\delta_3(-k_0)\delta_5(-k_0)}{\delta_2(-k_0)\delta_6(-k_0)\delta_1^2(-k_0)}=\exp\left(i\nu_1\ln(4)-\frac{1}{\pi i}\int_{k_0}^{\infty}\log_{\pi}\frac{|s+\omega k_0|}{|s+ k_0|}d\ln(1-|r_1(s)|^2)\right)\\
*\exp\left(i\nu_4\ln(3k_0^2)+\frac{1}{\pi i}\int_{-k_0}^{-\infty}\log_0|s+\omega k_0|d\ln(1-|r_2(s)|^2)\right),
\end{align*}
with
\begin{align*}
	&\delta_3(-k_0)\delta_5(-k_0)=\delta_1(-\omega^2k_0)\delta_1(-\omega k_0)=\exp\left(-i\nu_1\ln(k_0^2)-\frac{1}{\pi i}\int_{k_0}^\infty\log_0|s+\omega k_0|d\ln(1-|r_1(s)|^2)\right),\\
	&\delta_2(-k_0)\delta_6(-k_0)=\delta_4(-\omega^2k_0)\delta_4(-\omega k_0)=\exp\left(-i\nu_4\ln(3k_0^2)-\frac{1}{\pi i}\int_{-k_0}^{-\infty}\log_{\pi}|s+\omega k_0|d\ln(1-|r_2(s)|^2)\right),\\
	&\delta_1^2(-k_0)=\exp\left(-i\nu_1\ln(4k_0^2)-\frac{1}{\pi i}\int_{k_0}^{\infty}\log_{\pi}|s+ k_0|d\ln(1-|r_1(s)|^2)\right).
\end{align*}
\end{proof}

\subsection{The long-time asymptotics of solution for the modified SK-KK equation}

The modified SK-KK equation (\ref{msk-equation}) has Lax pair
\begin{equation}\label{msk-lax-pair}
\left\{\begin{array}{l}
	\Phi_x=\mathcal{M} \Phi, \\
	\Phi_t=\mathcal{N} \Phi,
\end{array}\right.
\end{equation}
where
\begin{align*}
	\mathcal{M} & =\left(\begin{array}{ccc}
		0 & 1 & 0 \\
		0 & -w & 1 \\
		\lambda & 0 & w
	\end{array}\right), \\
	\mathcal{N} & =\left(\begin{array}{lll}
		-6 \lambda w^2+6 \lambda w_x & N_{12} & -3 w_{x x}+9 \lambda+6 w_x w \\
		-6 \lambda w w_x+9 \lambda^2+3 \lambda w_{x x} & N_{22} & N_{23} \\
		N_{31} & 9 \lambda^2 & N_{33}
	\end{array}\right),
\end{align*}
with $N_{12}=-3w^2_x-w^4+w_{xxx}-9\lambda w-4w_x w^2+w_{xx}w, N_{22}=-3\lambda w_{x}+3\lambda w^2+w_{xxxx}w+w^5-5w_{xx}w^2-5w_xw_{xx}-5w_{x}^2w, N_{23}=-w^4+3w_{x}^2-2w_{xxx}+2w_{x}w^2+4w_{xx}w, N_{31}=-3\lambda w_{x}^2-\lambda w^4+\lambda w_{xxx}+9w^2\lambda-4\lambda w_{x}w^2+\lambda w_{xx}w$, $N_{33}=-w_{xxxx}+3\lambda w^2-3\lambda w_{x} -w^5+5w_{xx}w^2+5w_{x}w_{xx}+5w_x^2w$.
\par
By means of the same procedure as before, it is found that the Riemann-Hilbert problem associated with the modified SK-KK equation (\ref{msk-equation}) is just the one in Subsection \ref{Riemann-Hilbert-Problem} and the reconstruction formula for potential function $w(x,t)$ is
\begin{equation}
	w(x, t)=3 \lim _{k \rightarrow \infty}(k M(x, t, k))_{13}.
\end{equation}
\par
Following the similar way of Deift-Zhou steepest-descent method \cite{Deift-Zhou-1993}, the long-time asymptotics of the modified SK-KK equation (\ref{msk-equation}) is formulated below
\begin{equation} 
	\begin{split}
		\begin{aligned}\label{longtime-solution-msk}
			w(x,t)=&-\frac{1}{3^{\frac{1}{4}}2\sqrt{5t}k_0^{\frac{3}{2}}}[\sqrt{\nu_1}\cos \left(\frac{19\pi }{12}-\left(\arg y_1+\arg\Gamma(i\nu_1)\right)-(36\sqrt{3}tk_0^5)+\nu_1\ln({{3^{\frac{7}{2}}20tk_0^{5}}})\right.\\
			&\left.+\nu_4\ln(4)+\frac{1}{\pi }\int_{-k_0}^{-\infty}\log_{\pi}\frac{|s-\omega k_0|}{|s- k_0|}d\ln(1-|r_2(s)|^2)+\frac{1}{\pi }\int_{k_0}^\infty\log_0\frac{|s-k_0|}{|s-\omega k_0|}d\ln(1-|r_1(s)|^2)\right)\\	
			&+\sqrt{\nu_4}\cos \left(\frac{11\pi }{12}-\left(\arg y_4+\arg\Gamma(i\nu_4)\right)-(36\sqrt{3}tk_0^5)+\nu_4\ln({{3^{\frac{7}{2}}20tk_0^{5}}})\right.\\
			&\left.+\nu_1\ln(4)+\frac{1}{\pi }\int_{k_0}^{\infty}\log_{0}\frac{|s+\omega k_0|}{|s+ k_0|}d\ln(1-|r_1(s)|^2)+\frac{1}{\pi }\int_{-k_0}^{-\infty}\log_{\pi}\frac{|s+ k_0|}{|s+\omega k_0|}d\ln(1-|r_2(s)|^2)\right)]\\&+O\left(\frac{\ln t}{t}\right).\\
		\end{aligned}
	\end{split}
\end{equation}
\par
Moreover, the Miura transformations in (\ref{miura-SK}) can recover the long-time asymptotics of the SK equation (\ref{SK}) and KK equation (\ref{KK}) once again.

\begin{figure}
\centering
\includegraphics[width=12cm]{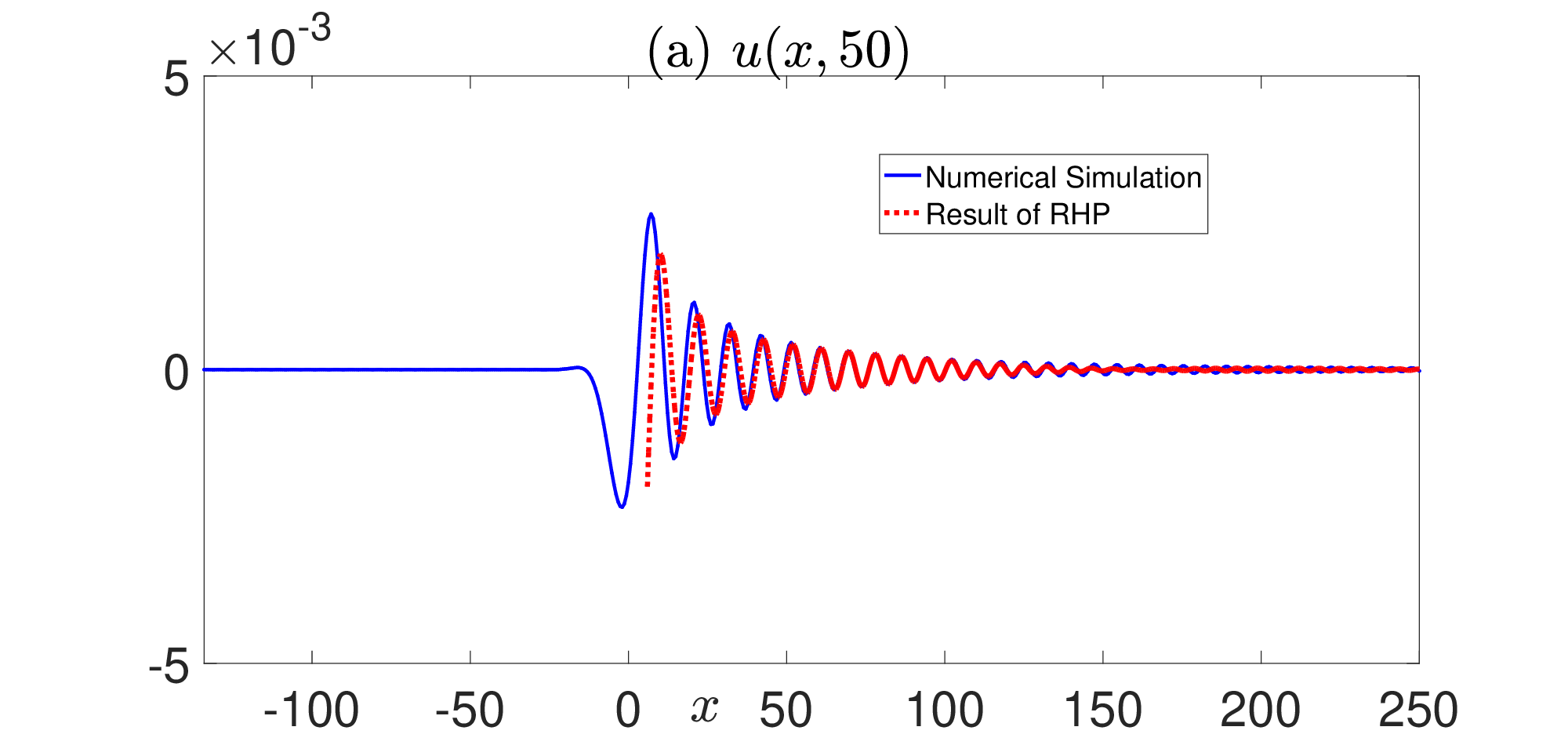}
\includegraphics[width=12cm]{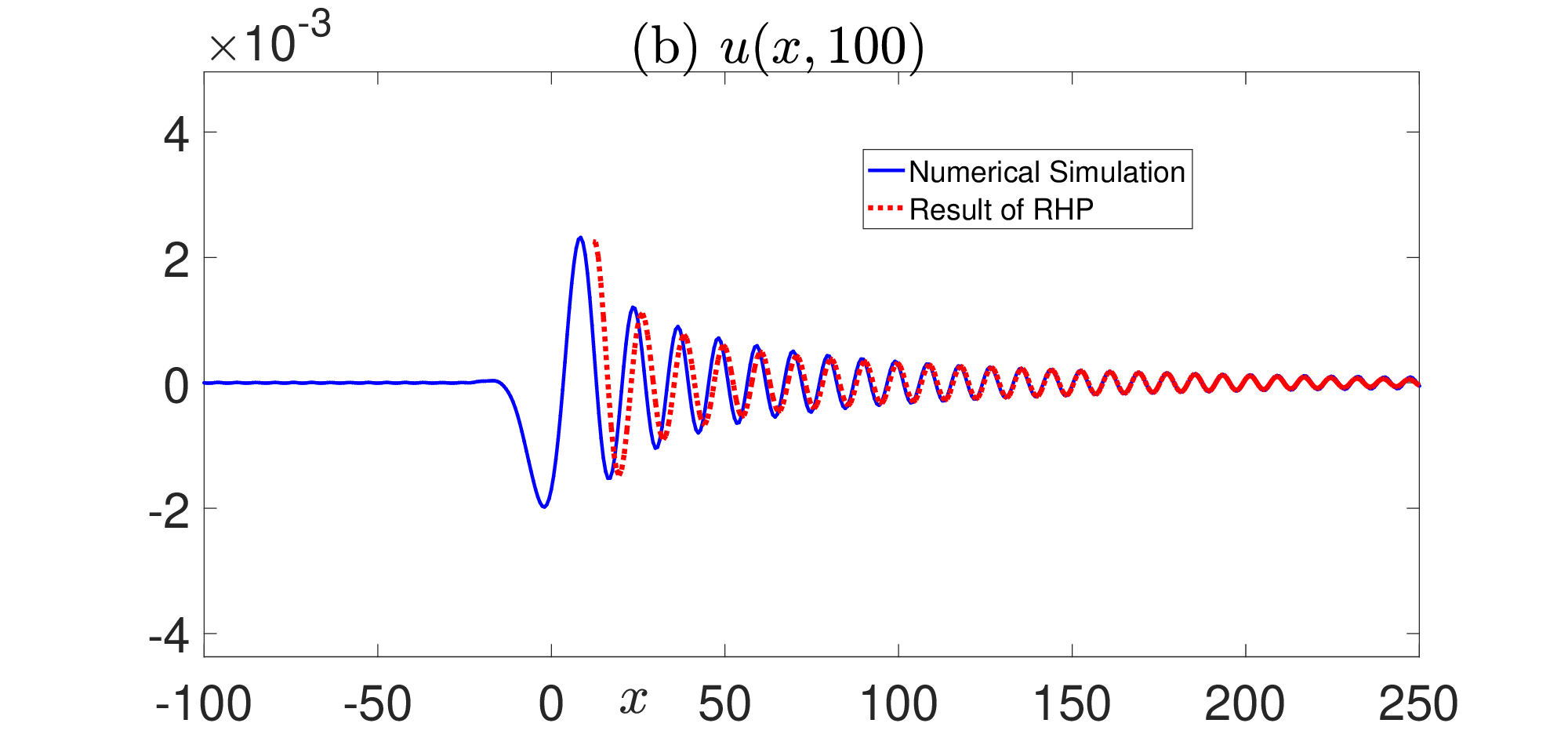}
\caption{{\protect\small The comparisons of the leading-order asymptotic approximation from Riemann-Hilbert problem and direct numerical simulations of the SK equation (\ref{SK}) with initial data (\ref{initial-SK-equation}) at time $t= 50$ and $t=100$, respectively. }}
\label{SK-comparisons}
\end{figure}

\begin{figure}
\centering
\includegraphics[width=12cm]{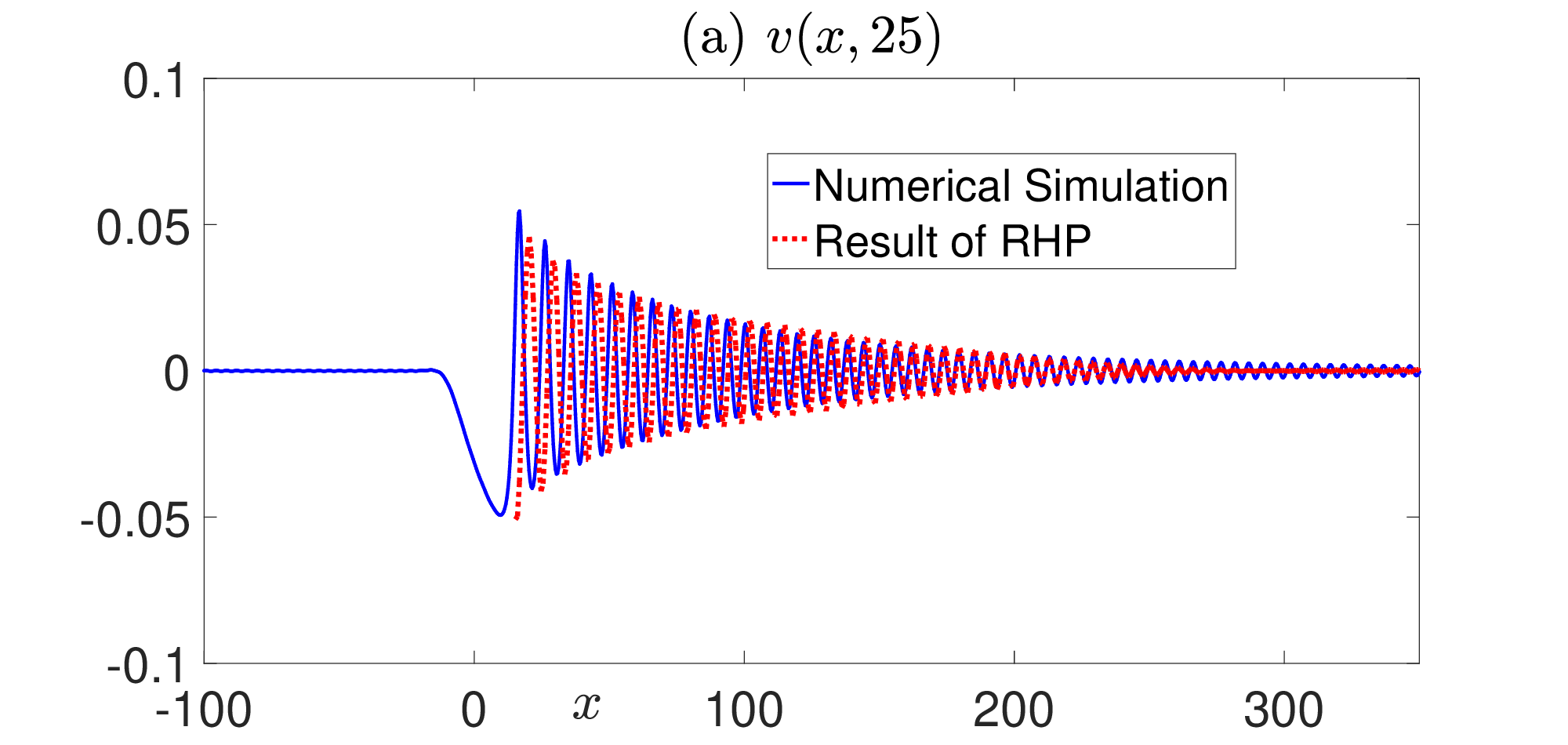}
\includegraphics[width=12cm]{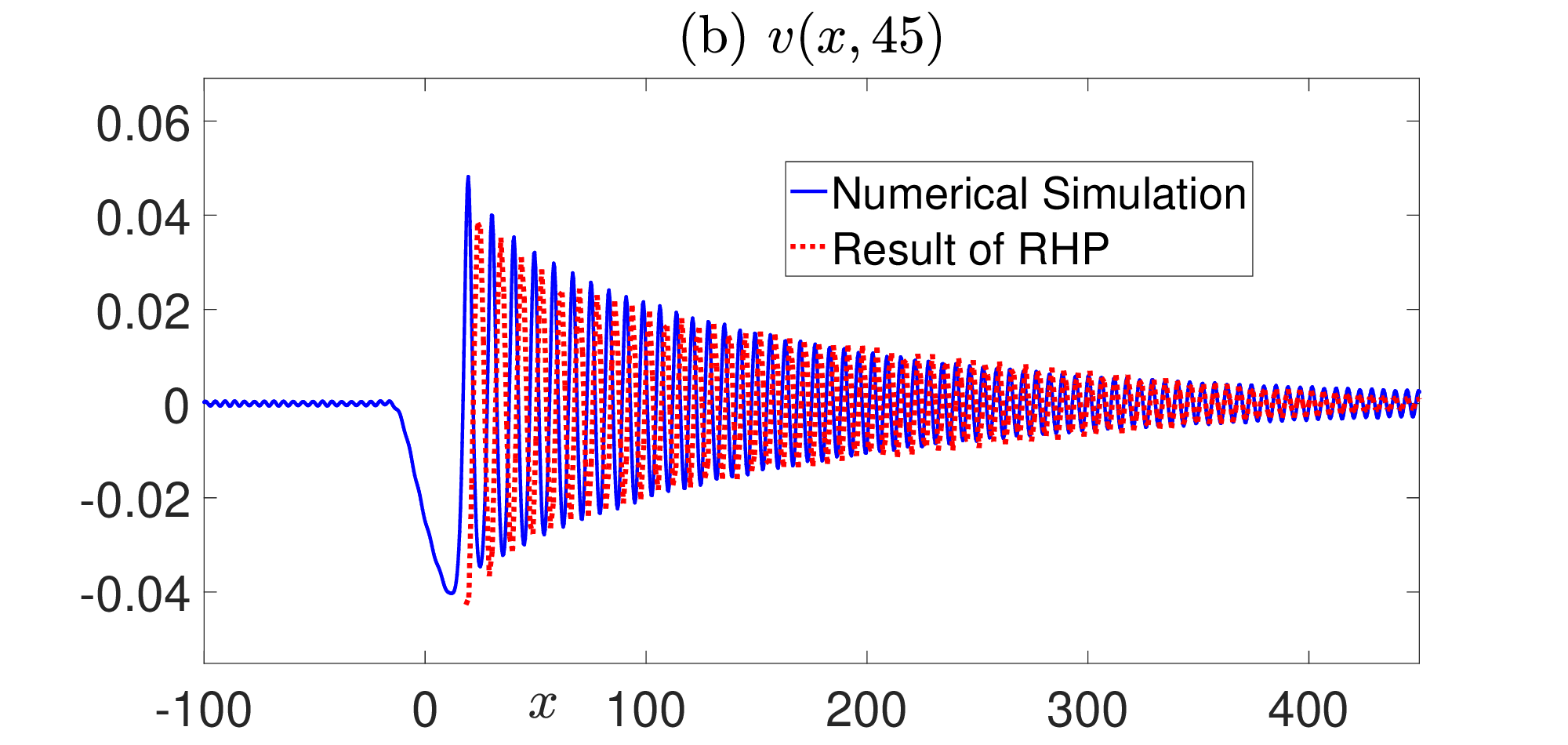}
\caption{{\protect\small The comparisons of the leading-order asymptotic approximation from Riemann-Hilbert problem and direct numerical simulations of the KK equation (\ref{KK}) with initial gaussian wavepacket (\ref{initial-KK-equation}) at time $t=25$ and $t=45$, respectively. }}
\label{KK-comparisons}
\end{figure}

\section{The verification of the theoretical results: Numerical simulations}

Now it is time to verify the theoretical results of the long-time asymptotics for the SK equation (\ref{SK}) by direct numerical simulations. To do so, for the SK equation (\ref{SK}), take the initial-value condition of the form
\begin{equation}\label{initial-SK-equation}
u(x,0)=u_0(x)=\frac{1}{600}(xe^{-\frac{x^2}{20}}-e^{-\frac{x^2}{10}}).
\end{equation}
\par
Fig. \ref{SK-comparisons} demonstrates the evolutions of the solution $u(x, t)$ to the SK equation (\ref{SK}) with initial data (\ref{initial-SK-equation}) at time $t= 50$ and $t=100$ by two different ways, where the dashed red line shows the leading-order asymptotics from the Riemann-Hilbert formulation and the solid blue line shows the wave profile obtained by numerical simulation. It is seen that the theoretical
results agree very well with the direct numerical simulation, which determines the reliability of the
Deift-Zhou steepest-descent method \cite{Deift-Zhou-1993}. From Fig. \ref{SK-comparisons}(a) and Fig. \ref{SK-comparisons}(b), it is expected that the asymptotic formula provides a better and better approximation as $t$ increases. The convergence is weak for small values of $x$, which is consistent with the fact that the asymptotic estimate (\ref{longtime-solution}) is not uniform near $x=0$. 
\par
For the KK equation (\ref{KK}), consider the initial gaussian wavepacket of the form
\begin{equation}\label{initial-KK-equation}
v(x,0)=v_0(x)=-\frac{1}{10}e^{-\frac{x^2}{20}}.
\end{equation}
\par
Fig. \ref{KK-comparisons} displays the evolutions of the solution $v(x, t)$ to the KK equation (\ref{KK}) with initial gaussian wavepacket (\ref{initial-KK-equation}) at time $t=25$ and $t=45$ by two different ways, where the dashed red line shows the leading-order asymptotics from the Riemann-Hilbert formulation and the solid blue line shows the wave profile obtained by numerical simulation. It is also observed that the theoretical results agree very well with the direct numerical simulation. From Fig. \ref{KK-comparisons}(a) and Fig. \ref{KK-comparisons}(b), it is also expected that the asymptotic formula provides a better and better approximation as $t$ increases. 
\par
For the modified SK-KK equation (\ref{msk-equation}), take the initial gaussian wavepacket of the form
\begin{equation}\label{initial-msk-equation}
w(x,0)=w_0(x)=-\frac{1}{10}e^{-\frac{x^2}{20}}.
\end{equation}
\par

\begin{figure}
\centering
\includegraphics[width=12cm]{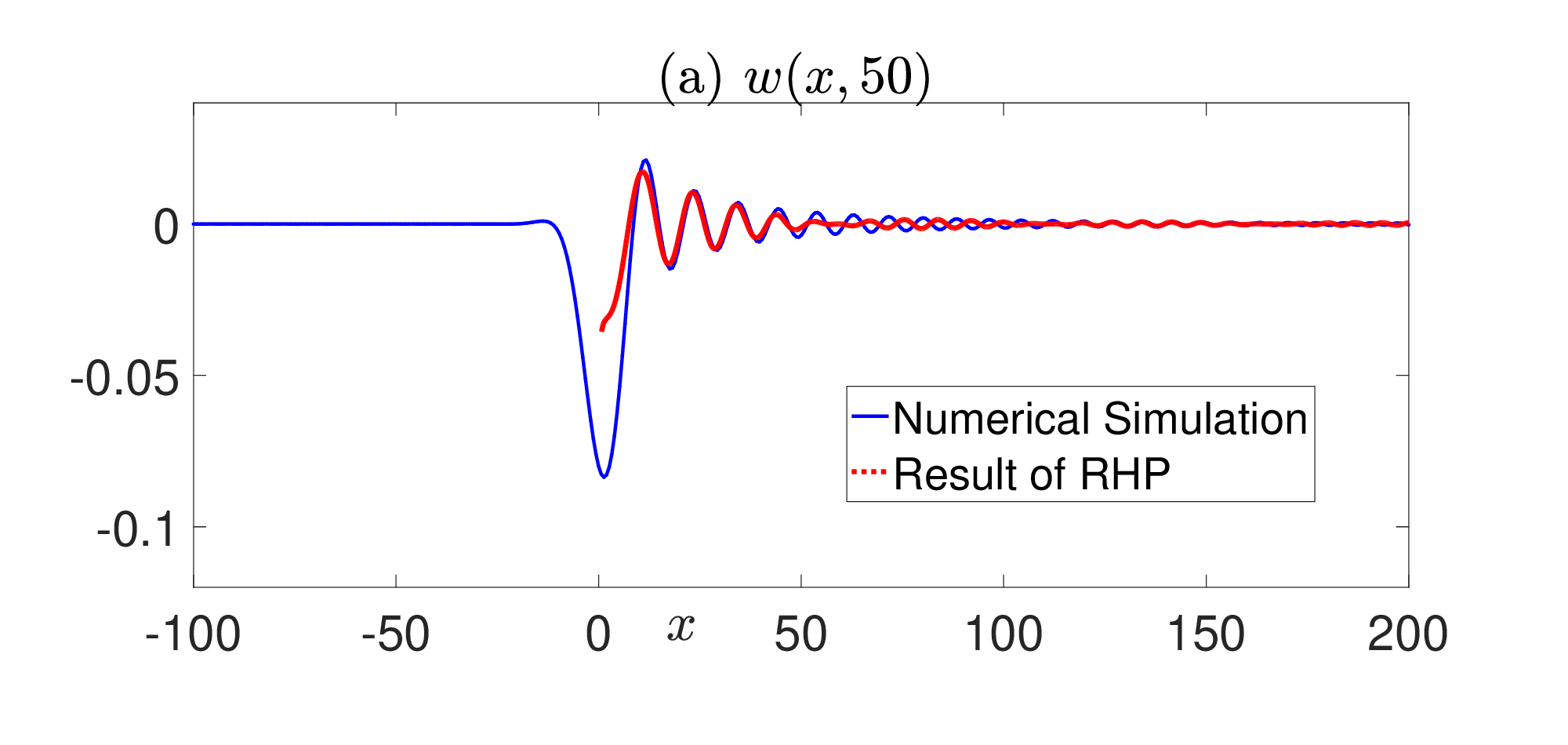}
\includegraphics[width=12cm]{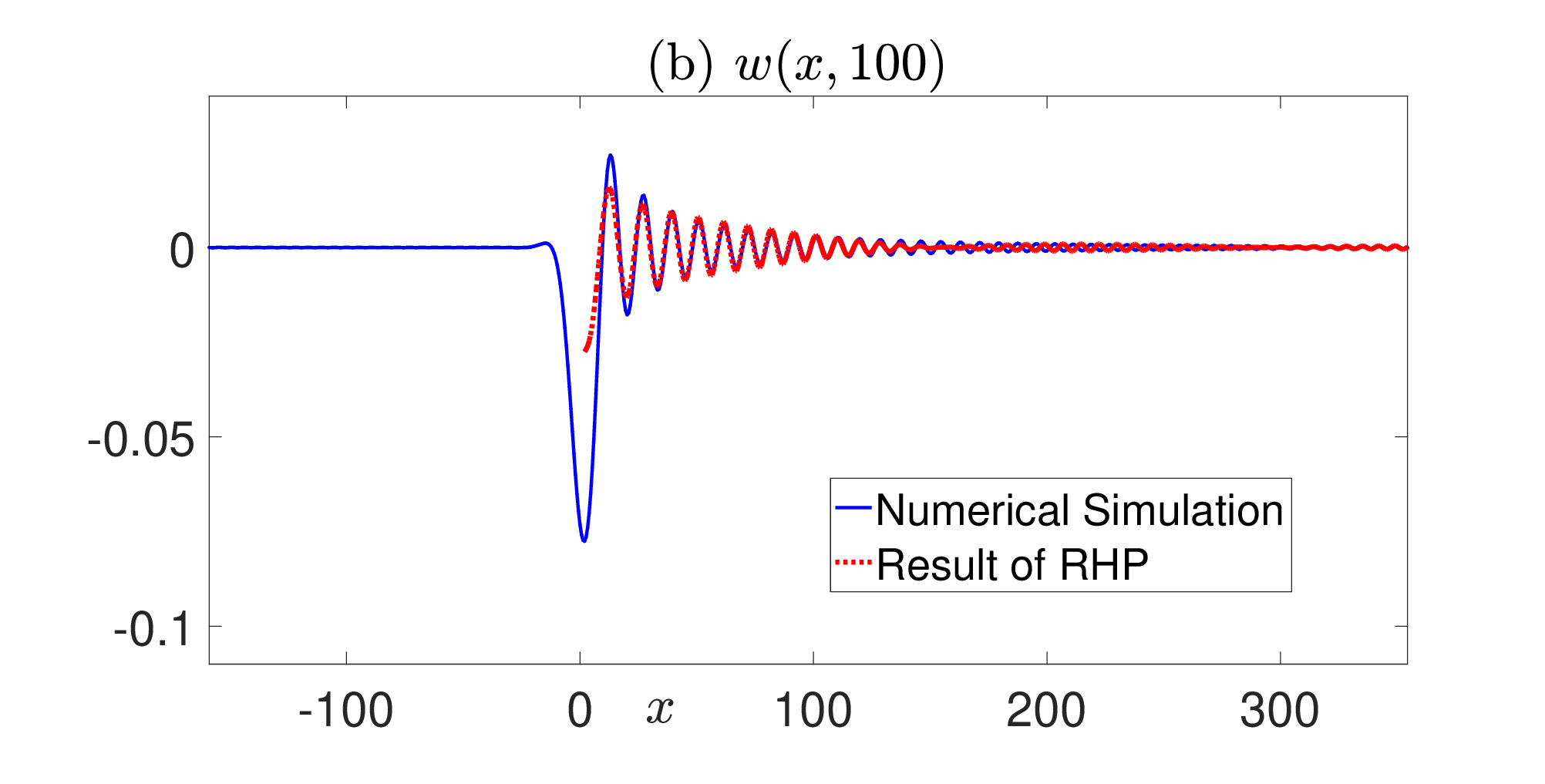}
\caption{{\protect\small The comparisons of the leading-order asymptotic approximation from Riemann-Hilbert problem and direct numerical simulations of the modified SK-KK equation (\ref{msk-equation}) with initial gaussian wavepacket (\ref{initial-msk-equation}) at time $t=50$ and $t=100$, respectively. }}
\label{msk-comparisons}
\end{figure}
\par
\par
Fig. \ref{msk-comparisons} displays the evolutions of the solution $w(x, t)$ to the modified SK-KK equation (\ref{msk-equation}) with initial gaussian wavepacket (\ref{initial-msk-equation}) at time $t=50$ and $t=100$ by two different ways, where the dashed red line shows the leading-order asymptotics from the Riemann-Hilbert formulation and the solid blue line shows the wave profile obtained by numerical simulation.  
\par

\begin{figure}
\centering
\includegraphics[width=12cm]{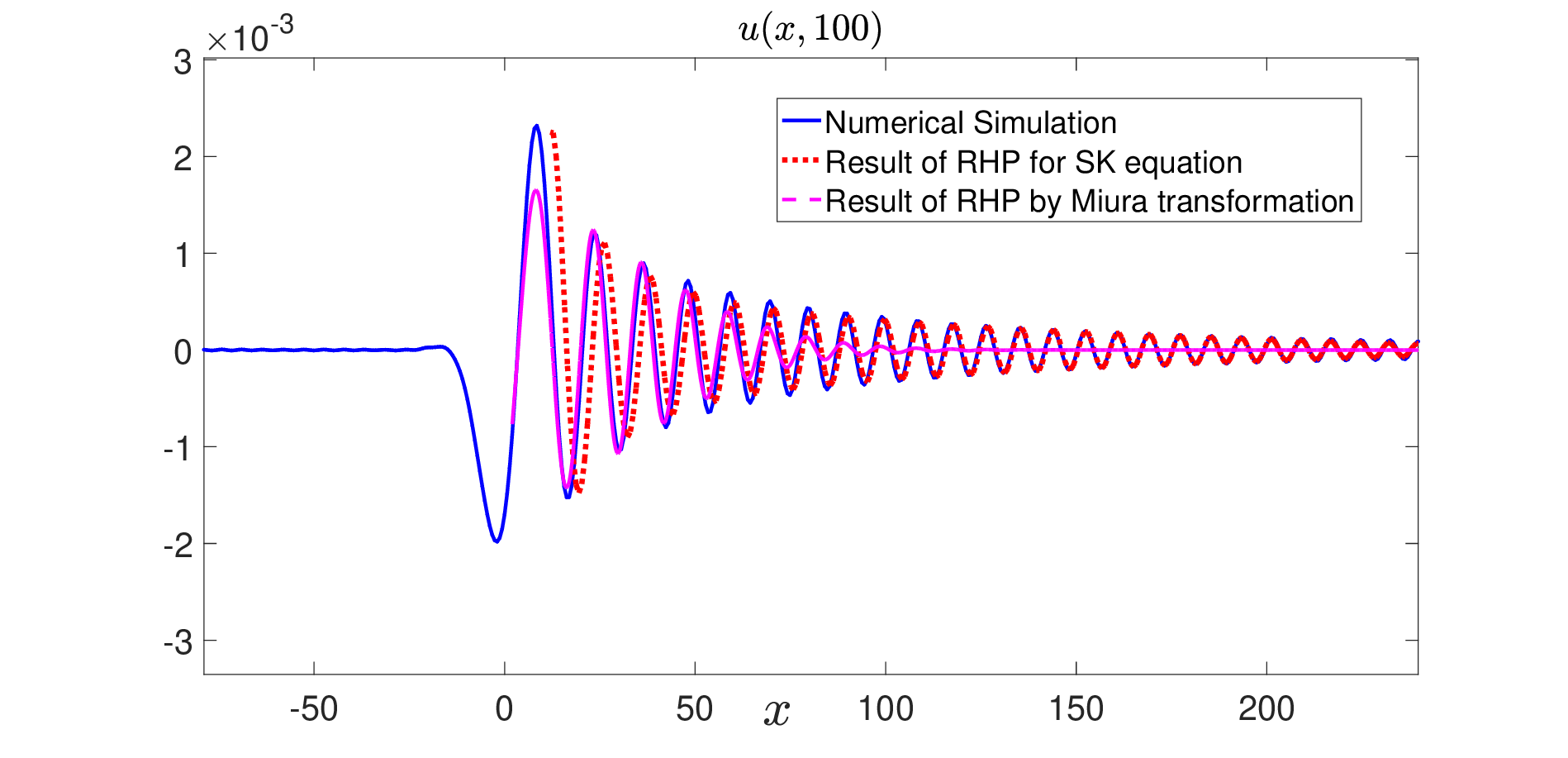}
\caption{{\protect\small The comparisons of the leading-order asymptotic approximation from Riemann-Hilbert problem, asymptotic solution from the Miura transformation $u=(w_x-w^2)/6$ and direct numerical simulations of the SK equation (\ref{SK}) with initial data (\ref{initial-SK-equation}) at time $t=100$. }}
\label{SK-msk-comparisons}
\end{figure}
\par
It is obvious that the Miura transformation $u=(w_x-w^2)/6$ in (\ref{miura-SK}) can also recover the solution of the SK equation (\ref{SK}). Thus we give the comparisons of direct numerical simulation and  
the leading-order asymptotic approximation (\ref{longtime-solution}) along with the long-time asymptotics for modified SK-KK equation (\ref{msk-equation}) and the Miura transformation in (\ref{miura-SK}). 
Fig. \ref{SK-msk-comparisons} shows these comparisons by considering the initial data of $u(x, 0)$ in (\ref{initial-SK-equation}), where the solid blue line shows the wave profile obtained by numerical simulation, the dashed red line shows the leading-order asymptotics (\ref{longtime-solution}) from the Riemann-Hilbert formulation of the SK equation (\ref{SK}), while the dotted purple line displays the solution of SK equation (\ref{SK}) from the Miura transformation $u=(w_x-w^2)/6$ in (\ref{miura-SK}) and leading-order asymptotic approximation (\ref{longtime-solution-msk}) of the modified SK-KK equation (\ref{msk-equation}). It is seen that the theoretical results agree very well with the direct numerical simulation, which determines the reliability of the Deift-Zhou steepest-descent method \cite{Deift-Zhou-1993}.

\section{The Painlev${\rm\acute{e}}$ Region}

It is seen from Figs. \ref{SK-comparisons}-\ref{SK-msk-comparisons} that the long-time asymptotic solutions are invalid near $x=0$. It is conjectured that this region can be expressed by the solution of the fourth-order Painlev${\rm\acute{e}}$ I equation. The self-similar transformation motivates this conjecture.
\par
For the region $|\frac{x}{t^{1/5}}|\le C$, letting $k\to\frac{k}{t^{1/5}}$ and denoting $\tau=tk_0^5$, the $\theta_{21}$ becomes into
$$
\theta_{21}(k)=9\sqrt{3}i(k^5-5k_0^4t^{4/5}k)\\
	=9\sqrt{3}i(k^5-5\tau^{4/5}k)
$$
Take the self-similar transformation $w(x,t)=(5t)^{-\frac{1}{5}}y(s)$ with $s=\frac{x}{\sqrt[5]{5} \sqrt[5]{t}}$, then one can get the equation
$$
y^{(5)}-5 y^{(3)} y^2-5 y''^2
+5 y^4 y'-5 y'^3
+\left(-5 y^{(3)}-s\right) y'-y \left(20 y' y''+1\right)
$$
$$
=y^{(5)}-5 y'^3-10y y' y''-5 y''^2-5 y' y^{(3)}-s y'-y-5 y^{(3)} y^2-10y y' y''+5 y^4 y'.
$$
\par
Integrating this equation yields
\begin{equation}
y^{(4)}=5y(y')^2+5y'y''+sy+5y^2y''-y^5,
\end{equation}
which is just the first Painlev${\rm\acute{e}}$ transcendence according to the fourth-order Painlev${\rm\acute{e}}$ I equation in \cite{Painleve-2006}.

\bibliographystyle{amsplain}

\end{document}